\documentclass[onefignum,onetabnum, final]{siamonline250211}

\usepackage{lipsum}
\usepackage{amsfonts}
\usepackage{graphicx}
\usepackage{algorithmic}
\usepackage{subcaption}
\usepackage{amsmath}
\usepackage{amssymb}
\usepackage{dsfont}
\usepackage{mathrsfs}
\usepackage{geometry}
\usepackage{xcolor}
\usepackage{enumitem}

\usepackage{mathtools}

\usepackage{tcolorbox}
\tcbset{colback=blue!05!white}
\tcbset{colframe=blue!30!black}
\ifpdf
  \DeclareGraphicsExtensions{.eps,.pdf,.png,.jpg}
\else
  \DeclareGraphicsExtensions{.eps}
\fi

\usepackage{enumitem}
\setlist[enumerate]{leftmargin=.5in}
\setlist[itemize]{leftmargin=.5in}

\newsiamremark{remark}{Remark}
\newsiamremark{example}{Example}
\newsiamremark{hypothesis}{Hypothesis}
\newsiamremark{assumption}{Assumption}
\crefname{hypothesis}{Hypothesis}{Hypotheses}
\newsiamthm{claim}{Claim}
\newsiamremark{fact}{Fact}
\crefname{fact}{Fact}{Facts}

\headers{Scalability of SORM for SDEs with multiplicative noise}{T. Schorlepp and T. Grafke}

\title{Scalability of the second-order reliability method for\\stochastic differential equations with multiplicative noise\thanks{Resubmitted to the editors on November 25, 2025.
\funding{T.G.\ acknowledges the support received from the EPSRC projects EP/T011866/1 and EP/V013319/1.}}}

\author{Timo Schorlepp\thanks{Courant Institute of Mathematical Sciences, New York University, New York, NY 10012, USA 
  (\email{timo.schorlepp@nyu.edu}).}
\and Tobias Grafke\thanks{Mathematics Institute, University of Warwick, Coventry CV4 7AL, UK
  (\email{T.Grafke@warwick.ac.uk}).}}

\usepackage{amsopn}

\usepackage{tikz}
\usetikzlibrary{shapes.geometric, arrows, positioning}
\tikzstyle{startstop} = [rectangle, rounded corners, minimum width=3cm, minimum height=1cm, text centered, draw=black, text width=3cm] \tikzstyle{startstopnoborder} = [minimum width=3cm, minimum height=1cm, text centered, text width=3cm] \tikzstyle{decision} = [diamond, minimum width=3cm, minimum height=1cm, text centered, draw=black, text width=3cm] \tikzstyle{arrow} = [thick,->,>=stealth]

\newcommand{\EE}{\mathbb{E}}
\newcommand{\RR}{\mathbb{R}}

\newcommand{\NN}{\mathbb{N}}
\newcommand{\PP}{\mathbb{P}}
\newcommand{\eps}{\varepsilon}
\DeclareMathOperator*{\argmin}{arg\,min}

\DeclareMathOperator{\ppr}{pr}
\DeclareMathOperator{\Id}{Id}
\newcommand{\dd}{\mathrm{d}}
\DeclareMathOperator{\rank}{rank}
\DeclareMathOperator{\trace}{tr}
\newcommand{\dv}[2]{\frac{\dd#1}{\dd#2}}

\newcommand{\fdv}[2]{\frac{\delta #1}{\delta #2}}
\newcommand{\nfdv}[3]{\frac{\delta^{#1}#2}{\delta#3^{#1}}}
\newcommand{\abs}[1]{\left\lvert#1\right\rvert}
\newcommand{\norm}[1]{\left\lVert#1\right\rVert}

\ifpdf
\hypersetup{
  pdftitle={Scalability of SORM for SDEs with multiplicative noise},
  pdfauthor={T. Schorlepp and T. Grafke}
}
\fi

\begin{document}

\maketitle

\begin{abstract}
We show how to efficiently compute asymptotically sharp estimates of extreme event probabilities in stochastic differential equations (SDEs) with small multiplicative Brownian noise. The underlying approximation is known as sharp large deviation theory or precise Laplace asymptotics in mathematics, the second-order reliability method (SORM) in reliability engineering, and the instanton or optimal fluctuation method with 1-loop corrections in physics. It is based on approximating the tail probability in question with the most probable realization of the stochastic process, and local perturbations around this realization. We first recall and contextualize the relevant classical theoretical result on precise Laplace asymptotics of diffusion processes [Ben Arous~(1988), Stochastics, 25(3), 125-153], and then show how to compute the involved infinite-dimensional quantities -- operator traces and Carleman--Fredholm determinants -- numerically in a way that is scalable with respect to the time discretization and remains feasible in high spatial dimensions. Using tools from automatic differentiation, we achieve a straightforward black-box numerical computation of the SORM estimates in JAX. The method is illustrated in examples of SDEs and stochastic partial differential equations, including a two-dimensional random advection-diffusion model of a passive scalar. We thereby demonstrate that it is possible to obtain efficient and accurate SORM estimates for very high-dimensional problems, as long as the infinite-dimensional structure of the problem is correctly taken into account. Our JAX implementation of the method is made publicly available.
\end{abstract}

\begin{keywords}
extreme events, stochastic differential equations, precise Laplace asymptotics, Carleman-Fredholm determinant, automatic differentiation, scalability, passive scalar
\end{keywords}

\begin{MSCcodes}
60F10, 28C20, 60H15, 35R60, 35Q35, 65F40, 65C50
\end{MSCcodes}


\section{Introduction}

Estimating the probabilities of rare but impactful events, an important problem throughout science and engineering, is usually done through -- generally expensive -- Monte Carlo methods such as importance sampling~\cite{bucklew:2013} or importance splitting methods~\cite{au-beck:2001,cerou-guyader-rousset:2019}. A simple alternative approach, which is principled and sampling-free but only asymptotically exact under certain assumptions, consists of using a Laplace approximation, see e.g.~\cite{piterbarg-fatalov:1995,friz-klose:2022,spokoiny:2023,katsevich:2024} for general overviews. In the present context, the Laplace approximation is also known as the second-order reliability method (SORM) in the reliability engineering literature~\cite{breitung:1984,rackwitz:2001,breitung:2006,hu-mansour-olsson-etal:2021,der-kureghian:2022}. For a standard normal random parameter vector $\eta \sim {\cal N}\left(0, \text{Id} \right)$ in $\RR^N$ and a parameter-to-event map $F \colon \RR^N \to \RR$, this approximation estimates rare failure probabilities as
\begin{align}
\PP\left[F(\eta) \geq z \right] \approx (2 \pi)^{-1/2} C(z) \exp \left\{-I(z) \right\}\,,
\label{eq:sorm-finite}
\end{align}
where $I(z) = \min_{\eta} \tfrac{1}{2} \norm{\eta}^2$ subject to $F(\eta) \geq z$. We write $I(z) = \tfrac{1}{2} \norm{\eta_z}^2$ with the minimizer or design point $\eta_z$, which can be regarded as the most likely realization of the random variable $\eta$ that realizes the event $F(\eta) \geq z$. The SORM prefactor $C(z)$, which is necessary to get an asymptotically exact estimate and accounts for Gaussian fluctuations around $\eta_z$, is given by
\begin{align}
C(z) = \left[2 I(z) \det \left(\text{Id}_N - \lambda_z \text{pr}_{\eta_z^\perp} \nabla^2 F(\eta_z)\text{pr}_{\eta_z^\perp} \right) \right]^{-1/2}
\label{eq:sorm-prefac-finite}
\end{align}
in terms of the weighted and projected Hessian $\nabla^2 F$ of the map $F$ at the minimizer $\eta_z$, with $\text{pr}_{\eta_z^\perp}$ denoting the orthogonal projection onto the subspace $\eta_z^\perp$, and $\lambda_z = I'(z)$ the Lagrange multiplier for the constraint $F(\eta) \geq z$. Calculating $C(z)$ hence requires finding and calculating the determinant of an $N \times N$ matrix. Importantly, the Laplace approximation can provide valuable qualitative insights into the extreme event under study by representing its ``most likely realization'', which can then serve as the basis for further Monte Carlo-based approaches, or it may very well in itself already yield accurate approximations of tail probabilities. Yet, a commonly advocated point of view in the reliability community seems to be that the asymptotic expression~\eqref{eq:sorm-finite} is hopeless to use in high-dimensional and nonlinear settings, both because of its inaccuracy and the expensive computations involved~\cite{schueller-iwo-pradlwarter:2004,valdebenito-pradlwarter-schueller:2010,adams-bohnhoff-dalbey-etal:2020,olivier-giovanis-aakash-etal:2020,hu-mansour-olsson-etal:2021,breitung:2024}.\\

In a previous work~\cite{schorlepp-tong-grafke-etal:2023}, approaching the problem from the perspective of large deviation theory in mathematics~\cite{varadhan:1984,iltis:1995,dembo-zeitouni:1998,freidlin-wentzell:2012,grafke-vanden-eijnden:2019,tong-vanden-eijnden-stadler:2021} and the instanton method in physics~\cite{falkovich-kolokolov-lebedev-etal:1996,zinn-justin:2021,grafke-grauer-schaefer:2015}, it was demonstrated to the contrary that SORM can in fact be extended to the infinite-dimensional setup of path space in a scalable way under certain assumptions.
As we recall in more detail below, what was shown in~\cite{schorlepp-tong-grafke-etal:2023} is that directly analogous expressions to~\eqref{eq:sorm-finite} and~\eqref{eq:sorm-prefac-finite}, involving a Fredholm determinant~\cite{mckean:2011} instead of the finite-dimensional determinant in~\eqref{eq:sorm-prefac-finite}, can still be used asymptotically and remain feasible to evaluate numerically for stochastic processes $(X_t)$ described by stochastic differential equations (SDEs) on $\RR^n$ -- and even well-posed stochastic partial differential equations (SPDEs) -- driven by additive Brownian noise, i.e.\ with a state-independent diffusion matrix $\sigma = \text{const}$. The underlying mathematical theory for precise Laplace asymptotics of general diffusion processes on $\mathbb{R}^n$ has in fact been known since the 1980s~\cite{azencott:1982,ben-arous:1988}, see also the review~\cite{piterbarg-fatalov:1995}. Above, and in the following, by referring to a method as ``scalable'', we mean that the computational cost does not grow, as the temporal resolution $n_t$ is refined, beyond the increased cost of solving differential equations of the same size as the original SDE itself. In other words, the number of \textit{differential equation solves} is thus stable under refinement of the temporal resolution, which, as we will see below, is a crucial property of a computational method such as SORM to be applicable in high dimensions.\\

Our goal here is to extend the approach of~\cite{schorlepp-tong-grafke-etal:2023} to general SDEs with any type of multiplicative Brownian noise with $\sigma = \sigma(X_t)$, providing a recipe that can be applied to tail probability estimates via SORM for any given diffusion process (under certain smoothness and nondegeneracy assumptions). Notably, beyond just completing the program initiated in~\cite{schorlepp-tong-grafke-etal:2023}, it turns out that the resulting SORM estimate in infinite dimensions is \textit{not} directly analogous to~\eqref{eq:sorm-prefac-finite} for general multiplicative-noise SDEs, and requires a certain kind of \textit{renormalization}. This can already be seen in the seminal work~\cite{ben-arous:1988}, and we provide a mostly self-contained exposition and intuitive explanation of the central result of~\cite{ben-arous:1988} in this paper, which then leads us to theorem~\ref{thm:prefac} of the present work for the generalization of~\eqref{eq:sorm-prefac-finite}. This generalization involves a so-called Carleman--Fredholm (CF) determinant~\cite{simon:1977,simon:2005} of the projected Hessian operator, as well as the trace of a regularized operator, and we address the intuition for and numerical computation of these quantities in a matrix-free and scalable way. The latter is important because evaluating~\eqref{eq:sorm-prefac-finite} or generalizations thereof involves determinants and traces of large matrices of size $N \times N$ for $N = n \cdot n_t$, whose size grows with the time discretization dimension $n_t$, and can quickly become infeasible. As a high-dimensional (spatially extended) example where this is relevant, we consider the advection and diffusion of a passive scalar in a two-dimensional random fluid flow. Furthermore, we show numerically that if instead of using theorem~\ref{thm:prefac}, equation~\eqref{eq:sorm-prefac-finite} is applied ``naively'' to discretized SDEs with multiplicative noise, then scalability is lost, and evaluating SORM estimate for high-dimensional systems hence becomes hopeless.\\

Beyond the practical aspect of reliability analysis for SDEs, the results presented here should also be of interest for the mathematics and physics communities. Results such as theorem~\ref{thm:prefac} for the leading coefficient in precise Laplace asymptotics are known theoretically to hold for different stochastic processes beyond diffusion processes on $\RR^n$, such as fractional Brownian motions or rough differential equations~\cite{inahama:2013,friz-gassiat-pigato:2021,yang-xu-pei:2025}, and singular SPDEs~\cite{berglund-di_gesu-weber:2017,friz-klose:2022,klose:2022,berglund:2022}. While the analysis and application of numerical methods to compute the involved Fredholm or CF determinants has mostly focused on low-dimensional state-spaces so far~\cite{bornemann:2010,friz-gassiat-pigato:2022}, we show that directly computing these objects in high-dimensional state spaces is in fact feasible and merits further study. For instance, large deviations of fractional Brownian motion have been of recent interest in the statistical physics literature~\cite{walter-pruessner-salbreux:2022,meerson-benichou-oshanin:2022,hartmann-meerson:2024}, and the computation of pre-exponential factors could be addressed with the methods presented here. Further possible applications in physics include the macroscopic fluctuation theory~\cite{bertini-de-sole-gabrielli-etal:2015}, fluid dynamics and turbulence~\cite{falkovich-kolokolov-lebedev-etal:1996,balkovsky-falkovich-kolokolv-etal:1997,fuchs-herbert-rolland:2022,burekovic-schaefer-grauer:2024} and growth processes~\cite{meerson-katzav-vilenkin:2016,krajenbrink-doussal-prolhac:2018,hartmann-meerson-sasorov:2021}.\\

Our main \textit{contribution} is to show how the SORM estimate in finite dimensions needs to be generalized to obtain scalability, when applied to any SDE driven by Brownian motion. This is a common source of high-dimensional rare event probability estimation problems in practice. In particular, we demonstrate that the high dimensionality of these systems in itself is not a problem for SORM. We make the resulting method accessible through a publicly available JAX implementation~\cite{Schorlepp-github}, which can be used in a black-box way for tail probability estimates via SORM of any diffusion process in $\mathbb{R}^n$ on a finite time interval, examples of which arise e.g.\ in physics, structural engineering, or mathematical finance.\\

As for potential \textit{limitations}, as usual, the presented SORM estimate, while asymptotically exact, can for any given example with a complicated failure region yield an inaccurate approximation of the true tail probability. It should hence be considered as a first step of a rare event analysis in such a case, and not as a replacement for sampling methods. On the other hand, because the SORM estimate has very few tuning parameters, and the number of noise-to-event map and gradient/Hessian evaluations only scales weakly with the rareness of the event, it is ideally suited for this role. On the more theoretical side, we assume for simplicity that there is a unique and nondegenerate large deviation minimizer or design point throughout the paper, and only consider finite time horizon problems. This excludes questions of metastability~\cite{berglund:2011,bouchet-reygner:2016,heller-limmer:2024}, as well as degeneracies~\cite{iltis:1995} e.g.\ due to other continuous symmetries~\cite{ellis-rosen:1981}, from the present work.\\

The \textit{outline} of the paper as follows: In section~\ref{sec:motiv}, we first recall the relevant results of~\cite{schorlepp-tong-grafke-etal:2023} for additive noise SDEs, demonstrate through two examples that directly applying those results fails to give a scalable method for general multiplicative noise SDEs, and then state the main result of the present paper in theorem~\ref{thm:prefac}, which fixes this issue. We discuss numerical methods to evaluate the SORM estimate from theorem~\ref{thm:prefac} in practice for potentially high-dimensional state spaces in section~\ref{sec:num-methods}, apply them in examples, including a random advection-diffusion problem, in section~\ref{sec:examples}, and conclude with a discussion and outlook in section~\ref{sec:concl}. Appendix~\ref{app:ricc} contains additional remarks on the computation of the prefactor $C(z)$ via differential Riccati equations. In appendix~\ref{app:index}, we list explicit coordinate expressions for the first and second variation of the forward map in case of multiplicative noise. Appendix~\ref{app:bip:laplace} addresses parallels between the construction for theorem~\ref{thm:prefac} and the Gauss--Newton approximation. Finally, in appendix~\ref{app:theory}, we give an explanation of how~\cite{ben-arous:1988} proves the main theoretical ingredient for theorem~\ref{thm:prefac}, and further provide a heuristic derivation, intuition and an introduction of all necessary mathematical concepts for theorem~\ref{thm:prefac}.

\section{Motivation and statement of theoretical result}
\label{sec:motiv}

\subsection{SORM approximation for additive noise SDEs}
\label{sec:sorm-additive}

In this subsection, we recall the main results of~\cite{schorlepp-tong-grafke-etal:2023}, the extension of which we will subsequently discuss. We consider an SDE with additive noise on the time interval $[0,T]$ with state space $\RR^n$
\begin{align}
\begin{cases}
\dd X^\eps_t = b(X^\eps_t) \dd t + \sqrt{\eps} \sigma \dd W_t\,,\\
X^\eps_0 = x \in \RR^n,
\end{cases}
\label{eq:sde-additive}
\end{align}
with small noise strength $\eps > 0$. Here $b \colon \RR^n \to \RR^n$ is a deterministic ``drift'' vector field, and $\sigma \in \RR^{n\times n}$ is a constant ``diffusion'' matrix, which we do not assume to necessarily have full rank (i.e.\ the noise can be ``degenerate''). By $W = (W_t)_{t \in [0,T]}$ we denote a standard $n$-dimensional Brownian motion. We are interested in quantifying extreme event probabilities of the real-valued random variable~$f(X_T^\eps)$, where $f \colon \RR^n \to \RR$ is a final-time ``observable'' (we could also consider more general observables that are functions of the full trajectory, without significant changes). The central object for the following results is the so-called noise-to-observable map, which is defined as
\begin{align}
F \colon L^2([0,T],\RR^n) \to \RR\,, \quad &F[\eta] = f(\phi(T)) \text{ for } \begin{cases}
\dot{\phi} = b(\phi) + \sigma \eta\,,\\
\phi(0) = x\,.
\end{cases} 
\label{eq:noise-to-obs-additive}
\end{align}
We would like to estimate the following tail probability for $z > F(0)$, which we expect to be small as $\eps \downarrow 0$:
\begin{align*}
  P^\eps(z) :=
\PP[F(\sqrt{\eps}\xi) \geq z]\,.
\end{align*}
Here, $\left( \xi(t) \right)_{t \in [0,T]}$ denotes standard Gaussian white noise in $\RR^n$, i.e.\ formally $\xi(t) = \dd W_t / \dd t$, for whose realizations the SDE~\eqref{eq:sde-additive} is solved (note the slight abuse of notation: we defined $F$ on $L^2$ above, which is important for the subsequent calculation of derivatives, but still evaluate it on white noise here). In the following, we use $\overset{\eps\downarrow0}{\sim}$ to denote asymptotic equivalence, i.e.\ we write $g(\eps) \overset{\eps\downarrow0}{\sim} h(\eps)$ if and only if $\lim_{\eps \downarrow 0} g(\eps)/h(\eps) = 1$. Then we have the sharp asymptotic expansion for the tail probability
\begin{align}
P^\eps(z) \overset{\eps\downarrow0}{\sim}\eps^{1/2} (2 \pi)^{-1/2}
  \, C(z) \,\exp\left\{-\frac1\eps I(z)\right\}\,.
  \label{eq:p-eps-asymp-form}
\end{align}
On the right-hand side of~\eqref{eq:p-eps-asymp-form}, the observable rate function $I$ is given by
\begin{align}
I \colon
\RR \to \RR\,, \quad I(z) := \tfrac{1}{2}
\norm{\eta_z}^2_{L^2([0,T], \RR^n)}
\label{eq:rate-function-def}
\end{align}
where the optimal ``instanton'' noise $\eta_z$ (also called large deviation minimizer, or design point) is the solution of the constrained optimization problem
\begin{align}
\eta_z =  \argmin_{\begin{subarray}{c}\eta\in L^2([0,T],\RR^n)\\
  \text{s.t. }F[\eta] = z
    \end{subarray}}  \;
  \frac{1}{2}
\norm{\eta}_{L^2([0,T],\RR^n)}^2\,.
\label{eq:min-prob-eta}
\end{align}
We assume that $\eta_z$ is uniquely defined in the following.
The problem~\eqref{eq:min-prob-eta} corresponds to finding the most probable noise realization that leads to the prescribed outcome for the final-time observable. One can solve the constrained optimization problem~\eqref{eq:min-prob-eta} with gradient-based methods. As a necessary ingredient for this, the gradient of the noise-to-observable map $F$ at any $\eta$, using the adjoint state method~\cite{plessix:2006,schorlepp-grafke-may-etal:2022}, can be evaluated as
\begin{align}
\fdv{\left(\lambda F \right)}{\eta}  = \sigma^\top \theta\,,
\label{eq:first-var-additive}
\end{align}
where $(\phi, \theta)$ are found by solving
\begin{align}
\begin{cases}
\dot{\phi} = b(\phi) + \sigma \eta\,, \quad
&\phi(0) = x\,,\\
\dot{\theta} = -\nabla b(\phi)^\top \theta\,,
\quad &\theta(T) = \lambda \nabla f(\phi(T))\,.
\label{eq:grad-additive}
\end{cases}
\end{align}
The optimal state space trajectory, and associated adjoint variable/conjugate momentum are denoted by $(\phi_z, \theta_z)$. Here, $\lambda > 0$ denotes a Lagrange multiplier for the constraint $F[\eta] = z$, and we will write $\lambda_z$ for the Lagrange multiplier at optimality from the first-order necessary optimality conditions $\eta_z = \lambda_z \; \delta F/\delta \eta |_{\eta_z}$.\\

We now address the computation of the leading-order prefactor $C(z)$ in~\eqref{eq:p-eps-asymp-form}, which is required to obtain an asymptotically sharp estimate instead of a log-asymptotic estimate. The main ingredient for $C(z)$ is a Fredholm determinant of the symmetric and trace-class (TC) Hessian of the noise-to-observable map $F$. We assume, as a necessary condition for the expansion to hold, that the scaled Hessian $A_{\lambda_z}$ of $F$ at the instanton noise $\eta_z$, defined as
\begin{align*}
A_{\lambda_z} :=  \left. \nfdv{2}{\left(\lambda_z F\right)}{\eta} \right|_{
\eta_z}\,,
\end{align*}
is such that the operator
\begin{align}
\begin{aligned}
&\text{Id} - A_{\lambda_z} \colon L^2\left([0,T], \RR^n\right) \to L^2\left([0,T], \RR^n\right) \text{ is positive definite}\\
&\qquad \text{on the codimension 1 subspace }  \eta_z^\perp \subset L^2([0,T], \RR^n).
\end{aligned}
\label{eq:sufficent-optimal}
\end{align}
This is nothing else than a second-order sufficient optimality condition for the minimizer $\eta_z$ in the constrained optimization problem~\eqref{eq:min-prob-eta}, and hence a natural assumption. Assuming the existence of a unique instanton $\eta_z$ with such a Hessian, we have
\begin{align}
C(z) = \left[2 I(z) \det \left(\Id -
\ppr_{\eta_z^\perp}A_{\lambda_z} \ppr_{
\eta_z^\perp} \right) \right]^{-1/2}\,,
\label{eq:tail-prob-prefac-sde-additive}
\end{align}
with $\ppr_{\eta_z^\perp}$ denoting the orthogonal projection in $L^2([0,T],\RR^n)$ onto the subspace $\eta_z^\perp$. Here, the operator being TC means that the eigenvalues $\mu_z^{(i)}$ of $B_z = \ppr_{
\eta_z^\perp} A_{\lambda_z} \ppr_{
\eta_z^\perp}$ are absolutely summable, i.e.\ $\sum_{i = 1}^\infty \abs{\mu_z^{(i)}} < \infty$. Then the Fredholm determinant $\det$ can be defined via $\det \left(\text{Id} - B_z \right) = \prod_{i= 1}^\infty \left(1 - \mu_z^{(i)} \right) \in (0, \infty)$, and we can approximate this infinite product for example through just finitely many leading eigenvalues with largest absolute value. We come back to the definition and properties of TC operators and Fredholm determinants in appendix~\ref{app:theory}. Intuitively speaking, TC operators are among the ``nicest'' ones to work with, for which e.g.\ notions of trace and determinant immediately carry over from finite dimensional matrices to infinite dimensions. This is in line with the fact that the precise Laplace asymptotics for additive noise SDEs in equations~\eqref{eq:p-eps-asymp-form} and \eqref{eq:rate-function-def} and in particular the prefactor~\eqref{eq:tail-prob-prefac-sde-additive} look exactly like their corresponding expressions~\eqref{eq:sorm-finite} and~\eqref{eq:sorm-prefac-finite} for a finite-dimensional random parameters space, as in~\cite{tong-vanden-eijnden-stadler:2021} and SORM in the engineering literature~\cite{breitung:1984,rackwitz:2001,breitung:2006,hu-mansour-olsson-etal:2021}.\\

To find the leading eigenvalues iteratively in a matrix-free and scalable way, we rely on matrix-vector products, or, in other words, evaluate the second variation $A_\lambda$ in a direction $\delta \eta$ with second order adjoints via
\begin{align}
A_\lambda \delta \eta = \sigma^\top \zeta
\label{eq:second-var-additive}
\end{align}
where
\begin{align}
\begin{cases}
\dot \gamma^{(1)} =\nabla b(\phi) \gamma^{(1)} + \sigma \delta \eta\,,\\
\dot \zeta = -\left \langle \nabla^2 b(\phi), \theta \right
\rangle \gamma^{(1)} - \nabla b^\top(\phi) \zeta\,,
\end{cases}
\text{ with } \begin{cases}
\gamma^{(1)}(0) = 0\,,\\
\zeta(T)
= \lambda \nabla^2 f(\phi(T)) \gamma^{(1)}(T)\,,
\end{cases}
\label{eq:second-order-adj-eq-additive}
\end{align}
with notation $\left[ \left \langle \nabla^2 b(\phi), \theta \right
\rangle \right]_{ij} = \sum_{k = 1}^n \theta_k \partial_i \partial_j b_k(\phi)$. The equations~\eqref{eq:second-var-additive} and~\eqref{eq:second-order-adj-eq-additive} can be obtained by linearizing~\eqref{eq:first-var-additive} and~\eqref{eq:grad-additive}. The strategy of approximating $C(z)$ in this way is scalable with respect to the time discretization and remains practically feasible in high spatial dimensions $n \gg 1$ if the forcing is low-rank, i.e.\ $\rank \sigma \ll n$, as was demonstrated in~\cite{schorlepp-tong-grafke-etal:2023} for a stochastic Korteweg--De Vries as well as the three-dimensional Navier--Stokes equations, see also~\cite{burekovic-schaefer-grauer:2024} for an application to a stochastic nonlinear Schr{\"o}dinger equation. We briefly summarize an alternative approach to prefactor computations using differential Riccati equations in appendix~\ref{app:ricc}.

\subsection{Breakdown of ``naive'' SORM scalability in SDEs with multiplicative noise}

We now discuss what happens if we transition from additive noise SDEs~\eqref{eq:sde-additive} to general multiplicative noise It{\^o} SDEs
\begin{align}
\begin{cases}
  \dd X^\eps_t = b\left(X^\eps_t\right)\,\dd t + \sqrt{\eps} \sigma\left(X^\eps_t\right)\,\dd W_t\,,\\
  X_0^\eps = x \in \RR^n
\end{cases}
 \label{eq:SDE}
\end{align}
with $\sigma = \sigma(x) \colon \RR^n \to \RR^{n \times n}$ not constant. The leading-order term does not change structurally~\cite{freidlin-wentzell:2012}, and we still have the rate function $I$ as defined above in~\eqref{eq:rate-function-def} with noise-to-observable map $F$ defined as
\begin{align}
F \colon L^2([0,T],\RR^n) \to \RR\,, \quad &F[\eta] = f(\phi(T)) \text{ for } \begin{cases}
\dot{\phi} = b(\phi) + \sigma(\phi) \eta\,,\\
\phi(0) = x\,.
\end{cases} 
\label{eq:noise-to-obs-mult}
\end{align}
The prefactor~\eqref{eq:tail-prob-prefac-sde-additive} does change significantly, however. By a similar computation as for the additive-noise case, we may evaluate the gradient of the noise-to-observable map $F$ at any $\eta \colon [0,T] \to \RR^n$ using an adjoint variable $\theta$ via
\begin{align*}
\fdv{\left(\lambda F \right)}{\eta}  = \sigma(\phi)^\top \theta\,,
\end{align*}
where $(\phi, \theta)$ are now found by solving
\begin{align}
\begin{cases}
\dot{\phi} = b(\phi) + \sigma(\phi) \eta\,, \quad
&\phi(0) = x\,,\\
\dot{\theta} = -\nabla b(\phi)^\top \theta - \left(\nabla \sigma(\phi) \eta\right)^\top \theta\,,
\quad &\theta(T) = \lambda \nabla f(\phi(T))\,.
\end{cases}
\label{eq:first-order-adjoint-mult-noise}
\end{align}
Similarly, for the second-order adjoint equations, we find
\begin{align}
A_\lambda \delta \eta =  \nfdv{2}{\left(\lambda F\right)}{\eta} \delta \eta  =  \sigma(\phi)^\top \zeta + \left \langle \theta,\left(\nabla \sigma(\phi) \cdot \right) \gamma^{(1)} \right \rangle
 \,,
 \label{eq:second-var-mult}
\end{align}
to evaluate the second variation operator $A_\lambda$ in the direction $\delta \eta$. Here, $(\gamma^{(1)}, \zeta)$ now solve the second-order tangent and adjoint equations
\begin{align}
\begin{aligned}
&\begin{cases}
\dot{\gamma}^{(1)} =  L[\eta, \phi] \gamma^{(1)} +  \sigma(\phi) \delta \eta\,,\\
\dot{\zeta} \quad= -L[\eta, \phi]^\top \zeta- 
\left(\nabla \sigma (\phi) \delta \eta \right)^\top \theta- \left \langle \nabla^2 b(\phi),
\theta \right \rangle \gamma^{(1)} - 
\left(\nabla^2 \sigma(\phi) \eta\right) \gamma^{(1)} \theta\,,
\end{cases}\\
\text{ with } &\begin{cases}
\gamma^{(1)}(0) = 0\,,\\
\zeta(T) = \lambda \nabla^2 f(\phi) \gamma^{(1)}(T)\,,
\end{cases}
\end{aligned}
\label{eq:second-order-adjoint-mult-noise}
\end{align}
with time-dependent and in general non-symmetric  $n \times n$ matrix $L[\eta, \phi] := \nabla b(\phi) + \nabla \sigma(\phi) \eta \colon \allowbreak [0,T] \to \RR^{n \times n}$. Explicit component expressions for these equations and the notation used are given in appendix~\ref{app:index}.
 Note that these are still just the linearizations of the first order adjoint equations~\eqref{eq:first-order-adjoint-mult-noise}. Importantly, however, it turns out that it is not enough to simply insert the modified second variation $A_{\lambda_z}$ into the additive-noise prefactor $C(z)$ as stated in~\eqref{eq:tail-prob-prefac-sde-additive}. One possible concern as to why~\eqref{eq:tail-prob-prefac-sde-additive} may need to be modified is the need to correctly account for It{\^o} calculus in the Taylor expansion to derive the prefactor, but there is in fact a more fundamental problem as well.\\

One could argue that a time discretization of the SDE~\eqref{eq:tail-prob-prefac-sde-additive} with $[0,T]$ into~$n_t$ points reduces the problem to a finite-dimensional one, so that finite-dimensional SORM~\eqref{eq:sorm-prefac-finite} should apply. After all, one always has to discretize the problem numerically, and is effectively using a Laplace approximation in high but finitely many dimensions, with a discretized Hessian $A_{\lambda_z}^{(n_t)} \in \RR^{(n \cdot n_t) \times (n \cdot n_t)}$ of the discrete noise-to-event map $F^{(n_t)} \colon \RR^{n \cdot n_t} \to \RR$. In the following two toy examples, we demonstrate that this strategy does indeed give correct results for the prefactor $C(z)$, but is ultimately doomed to fail in high-dimensional applications, since one needs the \textit{full} determinant in
\begin{align}
C^{(n_t)}(z) = \left[ 2 I^{(n_t)}(z) \det \left(\text{Id}_{n \cdot n_t} - \ppr_{\eta_z^\perp}  A_{\lambda_z}^{(n_t)} \ppr_{
\eta_z^\perp} \right) \right]^{-1/2}
\label{eq:prefac-wrong-discretize}
\end{align}
to get the answer and \textit{cannot} truncate after using just the leading eigenvalues. Hence, this prefactor computation strategy is \textit{not} scalable for multiplicative noise SDEs in general, thereby limiting the applicability of the Laplace approximation in high discretization or state space dimensions, unless it is modified in some way.

\subsubsection{Example: predator-prey model}
\label{sec:predprey}

In this subsection, we want to demonstrate how the application of finite-dimensional SORM to a numerical discretization of the continuous process is undesirable, and in particular not scalable as $n_t$ becomes large. For this, we consider a modified Lotka--Volterra model for $\left(X_t^\eps, Y_t^\eps\right)_{t \in [0,T]}$ with state space~$\RR^2$, where $X_t^\eps$ is the prey concentration and $Y_t^\eps$ the predator concentration at time $t \in [0,T]$, with~\cite{grafke-vanden-eijnden:2019}
\begin{align}
  \label{eq:predator-prey-SDE}
  \begin{cases}
   \dd X_t^\eps = \left(-\beta X_t^\eps Y_t^\eps + \alpha X_t^\eps + \delta\right)\,\dd t + \sqrt{\eps} \sqrt{\beta X_t^\eps Y_t^\eps + \alpha X_t^\eps + \delta} \,\dd W_x\,, & X_0^\eps = x_0\,,\\
    \dd Y_t^\eps = \left(+\beta X_t^\eps Y_t^\eps - \gamma Y_t^\eps + \delta\right)\,\dd t + \sqrt{\eps}\sqrt{\beta X_t^\eps Y_t^\eps + \gamma Y_t^\eps + \delta}\, \dd W_y\,, & Y_0^\eps= y_0\,.
  \end{cases}  
\end{align}
The process starts at the fixed point of the drift vector field $(x_0, y_0) \in \RR^2$. The equations can be
derived as the law of large numbers of a stoichiometric reaction
network modeling reproduction of prey, predation, and the death of
predators (with positive rates $\alpha$, $\beta$, and $\gamma$, respectively),
as well as additional migration of predators and prey from neighboring
habitats at rate $\delta > 0$. The migration term is added to avoid the
otherwise absorbing state of complete extinction of both species. The
multiplicative noise is consistent with the fluctuations obtained from
a central limit theorem.\\

We are interested in estimating the probability of an atypically high concentration of prey $X_T^\eps$ at time $T = 10$. Hence, the observable here is $f \colon \RR^2 \to \RR$, $f(x,y) = x$. We consider an event with $z = 0.5$, and set the rates to be $\alpha = 1$, $\beta = 5$, $\gamma = 1$, and $\delta = 0.1$. Taking a noise strength of $\eps = 0.01$ as an example, we perform $5.2 \cdot 10^5$ Monte Carlo simulations of~\eqref{eq:predator-prey-SDE} with $n_t = 1000$ equidistant Euler--Maruyama steps. This yields a tail probability estimate of $P^\eps(z) \in \left[1.56 \cdot 10^{-4},2.32 \cdot 10^{-4} \right]$ at $95\%$ asymptotic confidence by counting the number of samples that satisfy $X_T^\eps \geq z = 0.5$. A few sample paths of the SDE~\eqref{eq:predator-prey-SDE} that do end up in the extreme event set are shown in the top left of figure~\ref{fig:pred-prey-demo}.\\

\begin{figure}
\includegraphics[width = \textwidth]{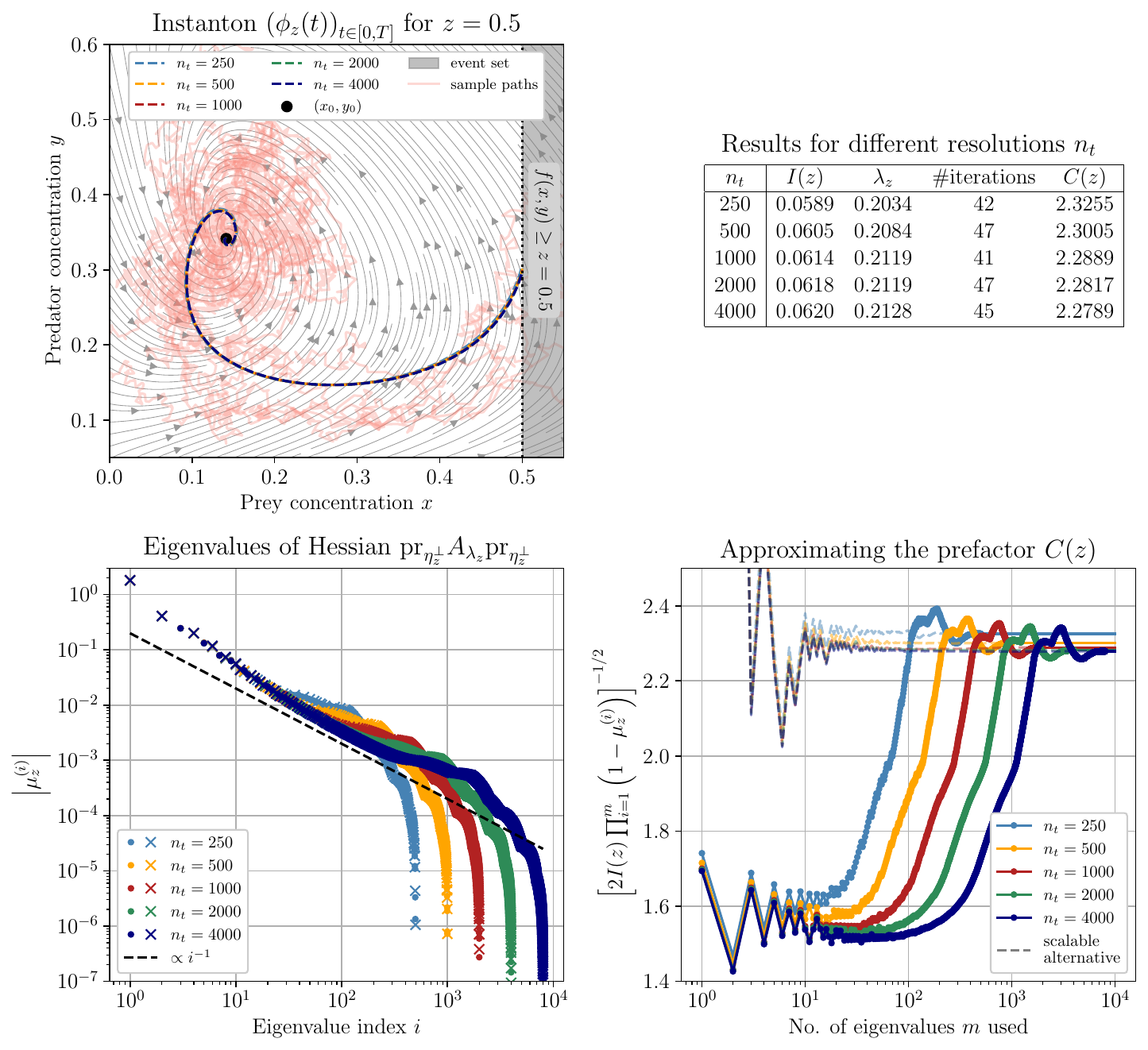}
\caption{A numerical example -- the probability estimation of an extreme prey concentration event $f(X_T^\eps, Y_T^\eps) = X_T^\eps \geq z = 0.5$ in the stochastic Lotka--Volterra model with multiplicative noise~\eqref{eq:predator-prey-SDE} -- demonstrating that ``naively'' trying to compute the prefactor $C(z)$ in~\eqref{eq:p-eps-asymp-form} using a discretized version of~\eqref{eq:tail-prob-prefac-sde-additive} fails to be scalable with respect to the time discretization dimension $n_t$. Top left: state space instanton trajectories $\phi_z$ for different $n_t$ (dashed) found from numerical optimization, and a few sample paths of the SDE~\eqref{eq:predator-prey-SDE} at noise strength $\eps = 0.01$ that reach the extreme event set at final time $T = 10$. The grey field lines visualize the drift vector field of the SDE~\eqref{eq:predator-prey-SDE}. Top right: table that summarizes the numerical results for the observable rate function $I(z)$, Lagrange multiplier $\lambda_z$, number of L-BFGS iterations in the instanton optimization, and prefactor $C(z)$ from~\eqref{eq:prefac-wrong-discretize} using all eigenvalues of the discrete projected Hessian $\text{pr}_{\eta_z^\perp} A_{\lambda_z} \text{pr}_{\eta_z^\perp}$. Bottom left: Full spectrum of all eigenvalues $\mu_z^{(i)}$ of the discretized projected Hessian at different resolutions. Positive eigenvalues are drawn as dots, negative eigenvalues as crosses. Bottom right: Prefactor $C(z)$ when approximating~\eqref{eq:prefac-wrong-discretize} using only the $m$ leading eigenvalues. No discretization-independent convergence is observed, and the number of necessary eigenvalues increases with $n_t$. The dashed lines, for references, are obtained by discretizing the correct continuum limit from theorem~\ref{thm:prefac} instead of~\eqref{eq:tail-prob-prefac-sde-additive}, as explained later in section~\ref{sec:num-methods}.}
\label{fig:pred-prey-demo}
\end{figure}

Now, with the Monte Carlo result as our ground truth, we turn to the instanton and prefactor calculation, which we will perform at a number of different time resolutions $n_t$ to investigate scalability. We note that in contrast to Monte Carlo simulations, the computations for these asymptotic estimates are independent of $\eps$, and the only $\eps$-dependence is explicit in equation~\eqref{eq:p-eps-asymp-form}. As a consequence, the asymptotic estimate, once computed, can be used to estimate probabilities for any~$\eps$ without any additional computation. The first step is the calculation of the instanton noise $\eta_z$ with associated state space trajectory $\phi_z$, which in particular yields the observable rate function~$I(z)$ in the exponent in~\eqref{eq:p-eps-asymp-form}. We follow~\cite{schorlepp-grafke-may-etal:2022} for this, and use an augmented Lagrangian formulation where a number of functionals
\begin{align}
{\cal L} = \tfrac{1}{2} \norm{\cdot}^2_{L^2} - \lambda (F - z) + \tfrac{\mu}{2}(F - z)^2
\label{eq:augmented-functional}
\end{align}
are minimized for an increasing sequence of penalty parameters $\mu>0$. Here, and for all other computations in this paper, we use the L-BFGS method of scipy.optimize for these unconstrained minimizations, and compute gradients of the noise-to-observable map $F$ defined in~\eqref{eq:noise-to-obs-mult} (and discretized with Euler steps) automatically via JAX~\cite{bradbury-frostig-hawkins-etal:2018}, cf.\ section~\ref{sec:implement} and~\cite{Schorlepp-github}. The resulting optimal trajectories for different time discretizations $n_t$ are shown in the top left of figure~\ref{fig:pred-prey-demo}, and show that the instanton trajectory does not change much as the time resolution is increased. More quantitatively, the table in the top right of figure~\ref{fig:pred-prey-demo} shows that $I(z)$ and $\lambda_z$ converge as the resolution is increased, while the total number of L-BFGS iterations remains constant.\\

The instanton is then used as an input for the prefactor calculation of $C(z)$ according to~\eqref{eq:prefac-wrong-discretize} (which we stress is \textit{not} the strategy that one \textit{should} use, as we demonstrate here). We also note that for the present example with small state space dimension $n = 2$, the Riccati approach from appendix~\ref{app:ricc} would be more efficient to evaluate the prefactor, or one could proceed along the lines of~\cite{psaros-kougioumtzoglou:2020,zhao-psaros-petromichelakis-etal:2022}. Nevertheless, computing the discretized Hessian $A_{\lambda_z}^{(n_t)}$ automatically via JAX (cf.\ section~\ref{sec:implement}), we obtain the results in the table in figure~\ref{fig:pred-prey-demo} for $C(z)$ when taking the full discrete determinant, i.e.\ using all eigenvalues of the discrete Hessian to compute the determinant. With this, we obtain an asymptotic estimate~\eqref{eq:p-eps-asymp-form} of $P^\eps(z) \approx 1.85 \cdot 10^{-4}$ for the tail probability from the instanton and prefactor calculation at the highest resolution $n_t = 4000$. This matches the Monte Carlo estimate, indicating that indeed, the prefactor calculation via~\eqref{eq:prefac-wrong-discretize} yields correct results, that also converge as the time discretization is refined (as evident from the table). However, this only holds if all eigenvalues of the discretized second variation are taken into account for the determinant calculation when using~\eqref{eq:prefac-wrong-discretize}, as the bottom right plot in figure~\ref{fig:pred-prey-demo} shows. In other words, to get anywhere near a good estimate, the number of eigenvalues needed scales linearly with $n_t$ and the problem becomes unfeasible to solve for large enough resolution.\\

A possible explanation is suggested by the $i^{-1}$ scaling of the eigenvalue spectrum of the discretized Hessians in the bottom left of figure~\ref{fig:pred-prey-demo}: The second variation operator is \textit{not TC} here, and the Fredholm determinant in~\eqref{eq:tail-prob-prefac-sde-additive} is hence not well-defined in general for multiplicative noise SDEs! Since we are not discretizing the correct continuum object, we cannot expect scalability of~\eqref{eq:prefac-wrong-discretize}, and cannot approximate the discrete determinant in~\eqref{eq:prefac-wrong-discretize} using just the leading eigenvalues of the discrete projected Hessian~$\text{pr}_{\eta_z^\perp} A_{\lambda_z} \text{pr}_{\eta_z^\perp}$.

\subsubsection{Example: geometric Brownian motion}
\label{sec:geombb}

The need to use a renormalized estimate instead of the classical one goes beyond arguments of numerical scalability. In this subsection, we demonstrate how a naive application of the finite-dimensional result in the infinite-dimensional setup leads to the wrong prediction even in cases where the computation can be carried out in an analytical way. For this, we consider the arguably simplest It{\^o} SDE with multiplicative noise, geometric Brownian motion in one dimension $n = 1$ with~\cite[example B.4]{schorlepp-grafke-grauer:2023}
\begin{align}
\begin{cases}
\dd X_t^\eps = -\beta X_t^\eps \dd t + \sqrt{2 \eps} X_t^\eps \dd W_t\,,\\
X_0^\eps = 1\,,
\end{cases}
\label{eq:geom-bm}
\end{align}
i.e.\ $b(x) = - \beta x$ and $\sigma(x) = \sqrt{2} x$ here.
To enable straightforward analytical calculations, we consider the observable $f(x) = \tfrac{1}{2} \left( \log x \right)^2$ (since $X_T^\eps$ follows a log-normal distribution), and instead of calculating a tail probability $P^\eps(z) = \PP \left[f(X_T^\eps) \geq z \right]$ as in the previous example, we directly consider the so-called moment-generating function (MGF) $J^\eps(\lambda) := \EE \left[\exp \left\{\frac{\lambda}{\eps} f\left(X_T^\eps\right) \right\} \right] \in [0, \infty]$ of the observable. This is the central quantity that underlies the derivation of~\eqref{eq:tail-prob-prefac-sde-additive} for additive noise, as well as the corresponding extension to multiplicative noise SDEs which we will introduce in subsection~\ref{sec:theory-result} and derive in appendix~\ref{app:theory} using the MGF. Indeed, the sharp asymptotics of $J^\eps(\lambda)$ as $\eps \downarrow 0$ that are described in~\cite{schorlepp-tong-grafke-etal:2023} for additive noise read $J^\eps(\lambda) \sim R_\lambda \exp \left\{+ \eps^{-1} I^*(\lambda) \right\}$ where $R_\lambda = \det \left(\text{Id} - A_\lambda \right)^{-1/2}$, and $I^*$ is the Legendre--Fenchel (LF) transform of the observable rate function $I$. In this simple example, the exact answer is known without any asymptotic estimates. Solving the SDE~\eqref{eq:geom-bm} using It{\^o}'s lemma and integrating the PDF of $X_T^\eps$ yields
\begin{align}
J^\eps(\lambda) = \underbrace{\left[1 - 2 \lambda T\right]^{-1/2} \exp
\left\{\frac{\beta \lambda T^2}{1 - 2 \lambda T} \right\}}_{= R_\lambda}
\underbrace{\exp \left\{\frac{\eps \lambda T^2}{2 \left( 1 - 2 \lambda T \right)}
\right\}}_{= 1 + {\cal O}(\eps)}  \underbrace{\exp \left\{ \frac{\lambda}{\eps} \frac{\beta^2
T^2}{2 \left( 1 - 2 \lambda T \right)} \right\}}_{=\exp \left\{+ \eps^{-1} I^*(\lambda) \right\}}\,.
\label{eq:mgf-geom-bm-ex}
\end{align}
as the exact result for the MGF for $\lambda < 1/(2T)$. We want to compare this exact result to the Laplace approximation with prefactor. As explained in~\cite{schorlepp-grafke-grauer:2023}, the instanton noise for the MGF, which solves $\eta_\lambda = \argmin_{\eta \in L^2} \tfrac{1}{2} \norm{\eta}^2_{L^2} - \lambda F[\eta]$, is given by $\eta_\lambda(t) = -\sqrt{2} \beta \lambda T / (1 - 2 \lambda T) = \text{const}$ here, with corresponding optimal state space trajectory $\phi_\lambda(t) = \exp \left\{-\beta t / (1 - 2 \lambda T) \right\}$ and momentum $\theta_\lambda(t) = - \beta \lambda T / (1 - 2 \lambda T) \exp \left\{\beta t / (1 - 2 \lambda T) \right\}$. This gives the correct exponential term from $\tfrac{1}{\eps} \left( \tfrac{1}{2} \norm{\eta_\lambda}^2_{L^2} - \lambda F[\eta_\lambda]\right)$ in~\eqref{eq:mgf-geom-bm-ex}. For the prefactor $R_\lambda$, however, we note that after a straightforward calculation, we can see that  the second variation operator $A_\lambda$ from~\eqref{eq:second-var-mult} and~\eqref{eq:second-order-adjoint-mult-noise} acts as
\begin{align*}
\left( A_\lambda \delta \eta \right)(t) = 2 \lambda \int_0^T \delta \eta(t') \dd t' \quad \forall t \in [0,T]\,.
\end{align*}
Notably, the right-hand side is independent of $t$, and hence the only nonzero eigenvalue of $A_\lambda$ is $\mu^{(1)}_\lambda=2 \lambda T$ with constant eigenfunction. Consequently, $A_\lambda$ \textit{is} TC here, but $\det(\text{Id} - A_\lambda)^{-1/2} = \left[1 - 2 \lambda T\right]^{-1/2}$ only recovers part of the leading order prefactor in the exact result~\eqref{eq:mgf-geom-bm-ex}. We conclude that the naive generalization of the finite-dimensional prefactor estimate~(\ref{eq:sorm-prefac-finite}) yields an incorrect answer here. The issue is clearly a missing It{\^o}--Stratonovich correction, which will be confirmed in the next subsection, and will yield the missing exponential term in~$R_\lambda$.\\

\begin{figure}
\includegraphics[width = \textwidth]{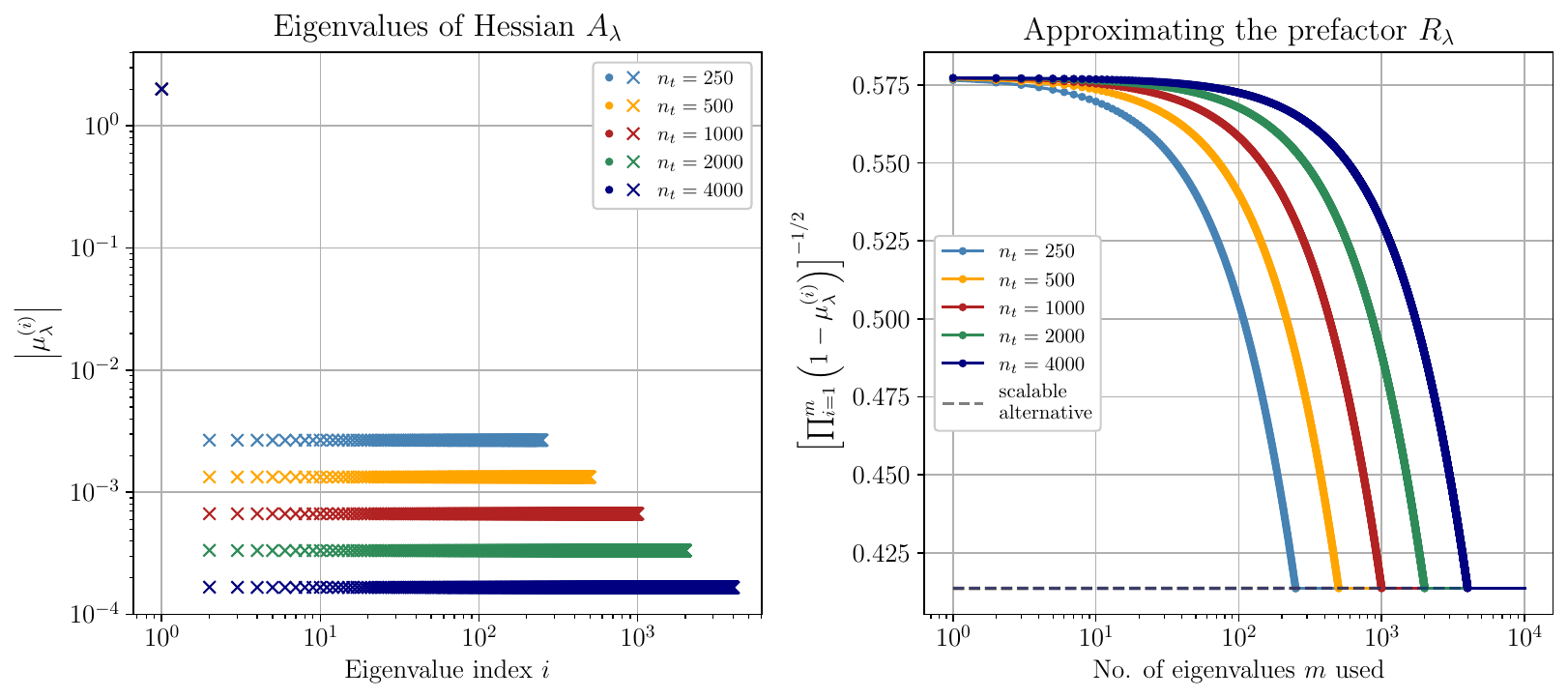}
\caption{Geometric Brownian motion toy example, where we want to estimate the MGF $J^\eps(\lambda)$ of geometric Brownian motion~\eqref{eq:geom-bm} at $\lambda = -1$, and see that again the prefactor calculation according to the MGF equivalent of~\eqref{eq:prefac-wrong-discretize} is not scalable here. The other parameters are $\beta = T = 1$, such that the second variation in continuous time $A_\lambda$ has exactly one nonzero eigenvalue $2 \lambda T = -2$. Left: Eigenvalue spectrum of the discrete Hessians at different resolutions $n_t$, with the crosses indicating negative eigenvalues. Importantly, in addition to the predicted eigenvalue $\approx -2$, all other eigenvalues of the discrete Hessian are small but nonzero here. Right: Approximation of $R_\lambda$ through $m$ leading eigenvalues as $\left[\prod_{i = 1}^m \left(1 - \mu_\lambda^{(i)} \right) \right]^{-1/2}$. We see that this is not scalable, since all of the small but nonzero eigenvalues of the discrete Hessian are needed to ``make up for the missing exponential term'' compared to the exact result~\eqref{eq:mgf-geom-bm-ex} for the MGF prefactor $R_\lambda$. Dashed lines plotted according to the numerical evaluation of the correct prefactor expression~\eqref{eq:mgf-prefac} for the MGF, which immediately gives the correct result after one eigenvalue.}
\label{fig:geom-bb}
\end{figure}

However, with the reasoning presented in the beginning of this section, one should again still be able to discretize the problem in time and calculate $R_\lambda^{(n_t)} = \det \left(\text{Id}_{n_t} -  A_{\lambda}^{(n_t)}\right)^{-1/2}$ for the resulting finite-dimensional problem, and should expect convergence of $R_\lambda^{(n_t)}$ to the correct $R_\lambda$ as $n_t \to \infty$. We show in figure~\ref{fig:geom-bb} that this is true, but that an approximation of the discrete determinant in terms of just the leading eigenvalues of $A_{\lambda}^{(n_t)}$ again fails. For this, we also implement an explicit Euler time stepper for the noise-to-observable map $F$ as defined in~\eqref{eq:noise-to-obs-mult} for the SDE~\eqref{eq:geom-bm} with $n_t$ time steps in JAX~\cite{bradbury-frostig-hawkins-etal:2018}, and differentiate automatically to find the gradient for minimization of $\tfrac{1}{2} \norm{\cdot}^2_{L^2} - \lambda F[\cdot]$ (no augmented Lagrangian~\eqref{eq:augmented-functional} is necessary here) as well as to find the full discrete Hessian at the large deviation minimizer $\eta_\lambda$ afterwards. The intuitive reason for the failure of approximating the discrete determinant by the leading eigenvalues is transparent here, since we know the exact result for $R_\lambda$: We are also not discretizing the correct continuum form of the prefactor $R_\lambda$ in this instance.

\subsection{Underlying theoretical result: SORM prefactor for multiplicative noise SDEs}
\label{sec:theory-result}

The examples from the previous subsection show that we need a proper infinite-dimensional theory for precise Laplace asymptotics of multiplicative noise SDEs, and can only hope for scalability of numerical methods when taking into account properties of the infinite-dimensional operators. While the classical reference for the present case of in\-finite-di\-men\-sional Laplace asymptotics of diffusion processes~\cite{ben-arous:1988} solves this problem on a theoretical level, our goal here is to provide further intuition and explanations of the corresponding abstract result, and ultimately turn it into a viable and scalable computational tool. We note that philosophically, our perspective of developing an infinite-dimensional theory and only discretizing in the end is similar to the modern treatment of Bayesian inverse problems in infinite dimensions~\cite{stuart:2010,ghattas-wilcox:2021,alexanderian:2021}. We provide further comments on the parallels of the result presented here to Laplace approximations in Bayesian inverse problems in appendix~\ref{app:bip:laplace}.

\begin{theorem}
\label{thm:prefac}
For a multiplicative noise SDE~\eqref{eq:SDE} and $z > f \left(X^0_T \right)$, assuming the existence of a unique and nondegenerate instanton noise $\eta_z$ in~\eqref{eq:min-prob-eta}, such that the second-order sufficient optimality condition~\eqref{eq:sufficent-optimal} holds, as well as sufficient differentiability and boundedness of the drift $b$, volatility $\sigma$ and observable $f$, we have the asymptotic expansion
\begin{align}
  P^\eps(z) = \PP\left[f\left(X_T^\eps\right) \geq z \right] \overset{\eps\downarrow0}{\sim}\eps^{1/2} (2 \pi)^{-1/2}
  \, C(z) \,\exp\left\{-\frac1\eps I(z)\right\}
  \label{eq:tail-asymp}
\end{align}
for the probability of a tail event, with observable rate function $I$ at $z$ defined as in~\eqref{eq:rate-function-def}. The leading-order prefactor $C(z) \in (0, \infty )$ is given by
\begin{align}
\begin{aligned}
C(z) = &\left[2 I(z) {\det}_2 \left(\Id - \text{pr}_{\eta_z^\perp} A_{\lambda_z} \text{pr}_{\eta_z^\perp} \right) \right]^{-1/2} \times \\
&\quad \times \exp \left\{\tfrac{1}{2}\trace \left[\text{pr}_{\eta_z^\perp}  \left(A_{\lambda_z} - \tilde{A}_{\lambda_z}\right) \text{pr}_{\eta_z^\perp}  \right] - \tfrac{1}{2} \left \langle e_z, \tilde{A}_{\lambda_z} e_z \right \rangle_{L^2} \right\}\,.
\end{aligned}
\label{eq:tail-prefac-projected}
\end{align}
\end{theorem}

We derive this result from~\cite{ben-arous:1988} (which treats expectations such as MGFs instead of tail probabilities) and explain it in detail in appendix~\ref{app:theory}. Here, ${\det}_2$ denotes a so-called Carleman--Fredholm (CF) determinant, instead of a Fredholm determinant $\det$ in the additive noise case~\eqref{eq:tail-prob-prefac-sde-additive}. The CF determinant is defined as ${\det}_2\left(\text{Id} - B\right) := \det \left( \left(\text{Id} - B \right) \exp \left\{B \right\} \right)$ for Hilbert--Schmidt (HS) operators~$B$ (cf.~\cite{simon:1977,simon:2005}), meaning that we only demand the eigenvalues of $B$ to be square-summable here, instead of absolutely summable for TC operators. The additional $\exp\left\{B \right\}$ in the definition of ${\det}_2$ regularizes the determinant, such that it is well-defined for HS operators. Roughly, this is because, when neglecting higher-order terms in the exponential, we have ${\det}_2\left(\text{Id} - B\right) \approx \det \left(\left(\text{Id} - B\right)\left(\text{Id} + B\right) \right) = \det \left(\text{Id} - B^2\right)$, and~$B^2$ \textit{is} TC. If~$B$ itself does happen to be TC, the CF determinant factorizes into ${\det}_2\left(\text{Id} - B\right) = \det \left(\text{Id} - B \right) \exp \left\{\trace \left[ B \right] \right\}$. We provide a more comprehensive introduction of these regularized determinants and their properties in appendix~\ref{app:theory}. Note that the regularized determinant ${\det}_2$ is exactly the necessary object to make sense of the predator-prey example we encountered in subsection~\ref{sec:predprey}, where numerical evidence in figure~\ref{fig:pred-prey-demo} suggested that the second variation operator $A_\lambda$ is not TC but HS.\\

In~\eqref{eq:tail-prefac-projected}, we defined the unit vector $e_z = \eta_z / \norm{\eta_z}_{L^2} \in L^2([0,T], \RR^n)$ in the instanton noise direction. The operator $\tilde{A}_\lambda$, as given below in~\eqref{eq:atilde-return} and~\eqref{eq:a-tilde-adjoint-mult-noise}, collects the terms appearing in $A_\lambda$ that correspond to iterated It{\^o} integrals in the ``stochastic'' Taylor expansion to obtain~\eqref{eq:tail-prefac-projected}, which ultimately lead to a less regular second variation operator for multiplicative noise. Substracting this part, the remaining operator $A_\lambda - \tilde{A}_\lambda$ is trace class. In contrast to this, the actual second variation operator $A_\lambda = (A_\lambda - \tilde{A}_\lambda) + \tilde{A}_\lambda$ is only guaranteed to be HS for multiplicative noise because of the presence of $\tilde{A}_\lambda$. Comparing with~\eqref{eq:atilde-return} and~\eqref{eq:a-tilde-adjoint-mult-noise} below, the result~\eqref{eq:tail-prefac-projected} correctly reduces to the previous prefactor~\eqref{eq:tail-prob-prefac-sde-additive} for additive noise, as $\tilde{A}_\lambda = 0$ in that case.\\

\begin{figure}
\centering
\begin{tikzpicture}[node distance=1.5cm,scale=0.8, every node/.style={scale=0.8}]
\node (deta) [startstop] {$\delta \eta$};
\node (gamma) [startstop, right=of deta,] {$\gamma^{(1)}$};
\node (zeta) [startstop, right=of gamma] {$\zeta$};
\node (Adeta) [startstop, right=of zeta] {$A_\lambda \delta \eta$};
\draw [arrow] (deta) -- (gamma) node[midway,above] () {$\int \dd t$};
\draw [arrow] (gamma) -- (zeta)  node[midway,above] () {$\int \dd t$};
\draw [arrow] (zeta) -- (Adeta);
\draw [->,out=90,in=90,looseness=0.3,red,dashed] (deta.north) to node[above]{$\int \dd t$}  (zeta.north);
\draw [->,out=270,in=270,looseness=0.3,red,dashed] (gamma.south) to  (Adeta.south);
\end{tikzpicture}
\caption{Flow chart that visualizes how the second variation $A_\lambda$ of the noise-to-event map $F$ acts on $\delta \eta \in L^2\left([0,T], \RR^n \right)$ for additive noise according to~\eqref{eq:second-var-additive} and~\eqref{eq:second-order-adj-eq-additive} and multiplicative noise as in~\eqref{eq:second-var-mult} and~\eqref{eq:second-order-adjoint-mult-noise}. Here, each integral symbol $\int \dd t$ above an arrow corresponds to solving an ODE with input given by the previous box, and arrows without integral symbols correspond to directly plugging the entry of the previous box into an equation to calculate the next box. The dashed red arrows are only present for multiplicative noise.}
\label{fig:flowcharts}
\end{figure}
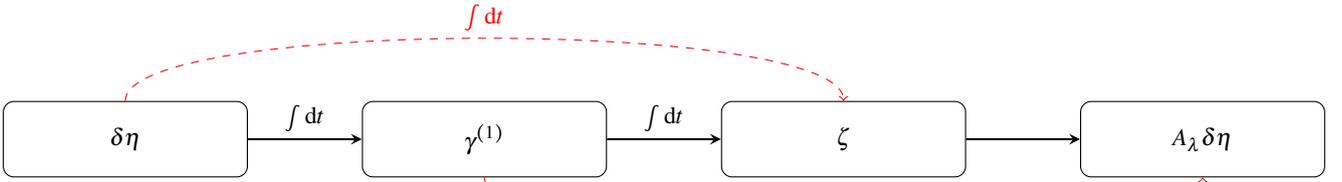

What has changed \textit{structurally} compared to the additive noise case here? Consider the second variation operator for multiplicative noise as given in~\eqref{eq:second-var-mult} and~\eqref{eq:second-order-adjoint-mult-noise}, compared to the corresponding operator for additive noise in~\eqref{eq:second-var-additive} and~\eqref{eq:second-order-adj-eq-additive}. For additive noise, the flow chart of operations to map a noise perturbation $\delta \eta \in L^2([0,T], \RR^n)$ to the second variation output $A_\lambda \delta \eta$ according to~\eqref{eq:second-var-additive} and~\eqref{eq:second-order-adj-eq-additive} is shown by the black arrows in figure~\ref{fig:flowcharts}. Notably, here, there are two subsequent integrations of the input. Intuitively, each integration or linear ODE solution should add regularity to the output, and we obtain a TC operator $A_\lambda$ for additive noise. The rightmost black arrow is just plugging $\zeta$ into $A_\lambda \delta \eta = \sigma^\top \zeta$ without further integrations. In general, for multiplicative noise, the flow chart for $\delta \eta \mapsto A_\lambda \delta \eta$ given by~\eqref{eq:second-var-mult} and~\eqref{eq:second-order-adjoint-mult-noise} now looks like the full sketch in figure~\ref{fig:flowcharts}. This means that at two points in the flow chart, indicated by the red dashed arrows, there is only one integration operation between input $\delta \eta$ and output $A_\lambda \delta \eta$. From this perspective, it is intuitive to expect that for multiplicative noise, the second variation operator will generally be less regular than for additive noise, and this needs to be taken into account when finding an alternative to~\eqref{eq:tail-prob-prefac-sde-additive}, since the Fredholm determinant in~\eqref{eq:tail-prob-prefac-sde-additive} will no longer be well-defined in general, cf.\ appendix~\ref{app:theory}. We can also immediately read off from the flow chart how $\tilde{A}_\lambda$, the regularizing subtrahend from $A_\lambda$, will have to look like: It needs to cut off the additional arrows, shown in red in figure~\ref{fig:flowcharts}, in the flow chart that lead to just one integration from input to output, such that we set
\begin{align}
\tilde{A}_\lambda \delta \eta = \sigma(\phi)^\top \zeta_{\text{sing}} + \left \langle \theta,\left(\nabla \sigma(\phi) \cdot \right) \gamma^{(1)} \right \rangle
 \,,
 \label{eq:atilde-return}
\end{align}
where
\begin{align}
\begin{cases}
\dot{\gamma}^{(1)} =  L[\eta, \phi] \gamma^{(1)} +  \sigma(\phi) \delta \eta\,, \quad &\gamma^{(1)}(0) = 0\,,\\
\dot{\zeta}_{\text{sing}} = -L[\eta, \phi]^\top \zeta_{\text{sing}}- 
\left(\nabla \sigma (\phi) \delta \eta \right)^\top \theta\,, \quad 
&\zeta_{\text{sing}}(T) = 0\,,
\end{cases}
\label{eq:a-tilde-adjoint-mult-noise}
\end{align}
where $\gamma^{(1)}$ the same as above in~\eqref{eq:second-order-adjoint-mult-noise}.
We see that the intuitive reasoning presented here in fact corresponds to the proofs that $A_{\lambda} - \tilde{A}_\lambda$ is TC and $A_\lambda$ HS in general in appendix~\ref{app:theory}. We note that~\eqref{eq:atilde-return} and~\eqref{eq:a-tilde-adjoint-mult-noise} indeed lead to $\tilde{A}_\lambda = 0$ for additive noise with $\nabla \sigma = 0$.\\

However, if $\tilde{A}_{\lambda_z}$ (and hence $A_{\lambda_z}$) does happen to be TC for a multiplicative noise It{\^o} SDE~\eqref{eq:SDE}, the result~\eqref{eq:tail-prefac-projected} will not just reduce to the additive noise answer~\eqref{eq:tail-prob-prefac-sde-additive}, but to
\begin{align}
C(z) = \left[2 I(z) {\det} \left(\Id - \text{pr}_{\eta_z^\perp} A_{\lambda_z} \text{pr}_{\eta_z^\perp} \right) \right]^{-1/2}  \exp \left\{-\tfrac{1}{2}\trace \left[\tilde{A}_{\lambda_z}\right] \right\}
\label{eq:tail-prefac-projected-mult-tc}
\end{align}
instead, with
\begin{align}
\trace \left[\tilde{A}_{\lambda_z} \right] = \int_0^T \trace\left[\sigma^\top(\phi_z(t))\nabla \sigma(\phi_z(t)) \theta_z(t) \right] \dd t = \int_0^T \sigma_{jk}(\phi_z(t)) \partial_j \sigma_{ik}(\phi_z(t)) \theta_{z,i}(t) \dd t
\label{eq:trace-atilde-integral}
\end{align}
in this case, as we verify in appendix~\ref{app:theory}. This additional term is directly related to the standard It{\^o}--Stratonovich correction term when converting between both formulations, which was the missing term in the geometric Brownian motion example in subsection~\ref{sec:geombb}. More generally, we can state the following result for the prefactor $C$ if we start from a Stratonovich SDE instead of an It{\^o} SDE in the first place:

\begin{remark}
\label{rem:ito-strato}
For a Stratonovich SDE
\begin{align}
\begin{cases}
  \dd \hat{X}^\eps_t = b\left(\hat{X}^\eps_t\right)\,\dd t + \sqrt{\eps} \sigma\left(\hat{X}^\eps_t\right) \circ \dd W_t\,,\\
  \hat{X}_0^\eps = x \in \RR^n\,.
\end{cases}
\label{eq:strato-sde}
\end{align} 
the corresponding It{\^o} SDE with the same law is
\begin{align}
\begin{cases}
  \dd X^\eps_t = \left[ b_i\left(X^\eps_t\right) + \frac{\eps}{2} \sigma_{jk}\left(X_t^\eps\right) \partial_j \sigma_{ik}\left(X_t^\eps\right) \right]_{1 \leq i \leq n}\,\dd t + \sqrt{\eps} \sigma\left(X^\eps_t\right) \, \dd W_t\,,\\
  X_0^\eps = x \in \RR^n\,.
\end{cases}
\label{eq:strat-to-ito-sde}
\end{align} 
Due to the $\eps$-dependent drift term, there is an additional term in the prefactor $C(z)$, which now becomes~\cite{ben-arous:1988}
\begin{align}
\begin{aligned}
&C(z) = \left[2 I(z) {\det}_2 \left(\Id - \text{pr}_{\eta_z^\perp} A_{\lambda_z} \text{pr}_{\eta_z^\perp} \right) \right]^{-1/2} \exp \bigg\{\frac{1}{2}\trace \left[\text{pr}_{\eta_z^\perp}  \left(A_{\lambda_z} - \tilde{A}_\lambda\right) \text{pr}_{\eta_z^\perp}  \right] \\
&\qquad \qquad - \tfrac{1}{2} \left \langle e_z, \tilde{A}_{\lambda_z} e_z \right \rangle + \tfrac{1}{2}\int_0^T \trace\left[\sigma^\top(\phi_z(t))\nabla \sigma(\phi_z(t)) \theta_z(t) \right] \dd t  \bigg\}\,.
\end{aligned}
\label{eq:strato-pref-complete}
\end{align}
Notably, if we have multiplicative noise in the Stratonovich sense and $A_{\lambda_z}$ is TC for a particular SDE~\eqref{eq:strato-sde}, then, by our discussion above, the prefactor $C(z)$ does actually reduce to~\eqref{eq:tail-prob-prefac-sde-additive}, i.e.\ to the direct analogue of finite-dimensional SORM. This is in line with the general rationale that stochastic calculus for Stratonovich SDEs should follow the rules of ordinary calculus.
\end{remark}

\begin{remark}
Before turning to its numerical evaluation and computational examples, we comment on the assumptions and applicability of theorem~\ref{thm:prefac}:
\begin{enumerate}
\item
Verifying whether the abstract assumptions of theorem~\ref{thm:prefac} hold in a given example can be challenging, and we refer to~\cite{deuschel-etal:2014} for a discussion of this issue for  coordinate projection observables $f \colon \RR^n \to \RR^m$, $f(x_1, \dots, x_n) = (x_1, \dots, x_m)$ with $m \in \{1, \dots, n\}$, and further references on heat kernel asymptotics.
Regarding existence of an instanton, H{\"o}rmander conditions on~$b$ and~$\sigma$ for hypoellipticity are commonly assumed as a sufficient condition. Uniqueness of the instanton can be hard to establish, and a practical remedy is to perform multiple optimization runs from different random initializations of $\eta$ to detect possibly non-unique solutions and select the ones with minimum action.
The non-degeneracy condition~\eqref{eq:sufficent-optimal} can be verified a posteriori by numerically computing the dominant eigenvalues of the projected second variation around the instanton as shown in the next section. Alternatively, geometric conditions based on nonfocality as introduced in~\cite{deuschel-etal:2014} could be formulated.
Lastly, for a comment on dropping the boundedness assumptions of theorem~\ref{thm:prefac}, see~\cite[Remark 2.11]{deuschel-etal:2014}.
\item
The presented theory applies rigorously to SDEs with state space $\mathbb{R}^n$. 
Nevertheless, we expect direct generalizations to hold for well-posed SPDEs with spatially sufficiently smooth, white-in-time Gaussian noise, such as the advection-diffusion problem of section~\ref{sec:advec-diffuse} below. 
There, we confirm numerically that theorem~\ref{thm:prefac}, applied to semi-discretized versions of this SPDE with spatial resolution $n$, leads to convergent tail probability estimates as $n$ increases, which further match the results of direct Monte Carlo simulations of the SPDE.
For theoretical works on small-noise expansions and precise large deviations for well-posed SPDEs, leading to analogous results, we refer to~\cite{rovira-tindel:2000,rovira-tindel:2001,rovira-tindel:2011,albeverio-di-persio-mastrogiacomo:2011}. 
For a recent extension to a singular SPDE, requiring a modification of the prefactor to account for renormalization, and a discussion of further references, see~\cite{friz-klose:2022,klose:2022}.
\end{enumerate}
\end{remark}

\section{Numerical implementation}
\label{sec:num-methods}

In this section, we address how to evaluate the asymptotically sharp tail probability estimate~\eqref{eq:tail-asymp} with prefactor~\eqref{eq:tail-prefac-projected} from theorem~\ref{thm:prefac} numerically in practice, and will subsequently illustrate our approach in different examples in section~\ref{sec:examples}. We do not discuss details of the computation of the instanton itself by solving~\eqref{eq:min-prob-eta} here, which has been presented elsewhere (see e.g.~\cite{grafke-vanden-eijnden:2019,schorlepp-grafke-may-etal:2022,zakine-vanden-eijnden:2023} and references therein, as well as section~\ref{sec:predprey}), and hence assume that for any given $z$, we do have access to the instanton noise $\eta_z$, state space trajectory $\phi_z$, momentum $\theta_z$, Lagrange multiplier $\lambda_z$, and rate function $I(z)$.

\subsection{Iterative and scalable determinant and trace computations}

It remains to evaluate the three real numbers
\begin{enumerate}
\item ${\det}_2 \left(\Id - B_z \right)$ with $B_z := \text{pr}_{\eta_z^\perp} A_{\lambda_z} \text{pr}_{\eta_z^\perp}$
\item $\trace \left[D_z  \right]$ with $D_z := \text{pr}_{\eta_z^\perp}  \left(A_{\lambda_z} - \tilde{A}_{\lambda_z}\right) \text{pr}_{\eta_z^\perp}$
\item $\left \langle e_z, \tilde{A}_{\lambda_z} e_z \right \rangle_{L^2}$ with $e_z = \eta_z / \norm{\eta_z}$ (computing this just requires one operator evaluation according to~\eqref{eq:atilde-return} and~\eqref{eq:a-tilde-adjoint-mult-noise})
\end{enumerate}
to obtain a sharp tail probability estimate with $C(z)$ given by~\eqref{eq:tail-prefac-projected}. We stress that quantities 1.\ and 2.\ involve operators $B_z$ and $D_z$ defined on infinite-dimensional path space (effectively $L^2([0,T],\RR^n)$ here), and in practice one always has to discretize time to approximately calculate the (regularized) determinant and trace. This corresponds to a projection of these operators onto a finite-dimensional subspace, either through discretizing the linear differential equations that need to be solved to evaluate the operators by any time stepping scheme with $n_t$ points in time, or by explicitly applying them to any finite set of $n_t$ orthonormal basis functions. The resulting finite-dimensional matrices $B_z^{(n_t)}, D_z^{(n_t)} \in \RR^{(n \cdot n_t) \times (n \cdot n_t)}$, if found explicitly in this way, may still be large, in particular if the state space dimension $n$ is already large in itself, e.g.\ if the SDE~\eqref{eq:SDE} stems from the semidiscretization of a stochastic PDE (say with $n_t = 1000$ and $n = 1024^2$). If $n$ is small, say $n = 1$ or $2$, then using explicit quadratures and computing the full determinants and traces of $B_z^{(n_t)}, D_z^{(n_t)}$ can yield accurate results, see e.g.~\cite{bornemann:2010,friz-gassiat-pigato:2022}. Here, however, our goal is to compute the quantities 1.\ and 2.\ in an iterative and matrix-free way, that avoids having to construct and store $B_z^{(n_t)}, D_z^{(n_t)}$ explicitly, and scales to very high dimensions.\\

Hence, we only define $B_z^{(n_t)}, D_z^{(n_t)}$ implicitly through their action on vectors, without ever storing these matrices. Note that the projection operators act as
\begin{align}
(\ppr_{\eta^\perp_z} \delta \eta)(t) = \delta \eta(t) -
\frac{\left \langle \eta_z, \delta \eta \right \rangle_{L^2([0,T],
\RR^n)}}{\norm{\eta_z}^2_{L^2([0,T],\RR^n)}} \eta_z(t)\,,
\label{eq:proj}
\end{align}
the action of $A_{\lambda_z}$ on $\delta \eta$ is defined through~\eqref{eq:second-var-mult} and~\eqref{eq:second-order-adjoint-mult-noise}, and for convenience, we also note here that for the regularized operator
\begin{align}
\left(A_{\lambda_z} - \tilde{A}_{\lambda_z}\right) \delta \eta =  \sigma(\phi_z)^\top \zeta_{\text{reg}}\,,
\label{eq:a-atilde-action}
\end{align}
where
\begin{align}
\begin{aligned}
&\begin{cases}
\dot{\gamma}^{(1)} =  L[\eta_z, \phi_z] \gamma^{(1)} +  \sigma(\phi_z) \delta \eta\,,\\
\dot{\zeta}_{\text{reg}} \quad= -L[\eta_z, \phi_z]^\top \zeta_{\text{reg}}- \left \langle \nabla^2 b(\phi_z),
\theta_z \right \rangle \gamma^{(1)} - 
\left(\nabla^2 \sigma(\phi_z) \eta_z\right) \gamma^{(1)} \theta_z\,,
\end{cases}\\
\text{ with }&\begin{cases}
\gamma^{(1)}(0) = 0\,,\\
\zeta_{\text{reg}}(T) = \lambda_z \nabla^2 f(\phi_z) \gamma^{(1)}(T)\,,
\end{cases}
\end{aligned}
\label{eq:a-atilde-diff-eq}
\end{align}
with time-dependent $n \times n$ matrix $L[\eta, \phi] = \nabla b(\phi) + \nabla \sigma(\phi) \eta$. Assuming that $\eta_z$, $\phi_z$, and $\theta_z$ are all known, in addition to calculating projections~\eqref{eq:proj}, each evaluation of $B_z^{(n_t)}$ or $D_z^{(n_t)}$ acting on a test vector $\delta \eta \in \RR^{n_t \cdot n}$ hence amounts to solving two differential equations in $\RR^n$, each of comparable cost to integrating the original SDE~\eqref{eq:SDE}: one (tangent equation) forward in time and one ((regularized) second order adjoint equation) backwards in time. For the discretization of these forward and backward differential equations, one can choose arbitrary stable and consistent time stepping schemes. As we will discuss below, we actually evaluate the actions of both $A_{\lambda_z}^{(n_t)}$ and $A_{\lambda_z}^{(n_t)} - \tilde{A}_{\lambda_z}^{(n_t)}$ through automatic differentiation for this paper, so~\eqref{eq:second-var-mult}, \eqref{eq:second-order-adjoint-mult-noise}, \eqref{eq:a-atilde-action} and~\eqref{eq:a-atilde-diff-eq} do not need to be implemented explicitly, and their time-stepping simply follows from how we choose to discretize the forward map~\eqref{eq:noise-to-obs-mult}.\\

Equipped with a way of evaluating how the discretized operators act on vectors, we proceed to approximate the CF determinant in~\eqref{eq:tail-prefac-projected} via
\begin{align*}
{\det}_2 \left(\text{Id} - B_z \right) = \prod_{i =1}^\infty \left(1 - \mu_z^{(i)} \right) \exp \left\{\mu_z^{(i)}  \right\} \approx 
\underbrace{\prod_{i =1}^M \left(1 - \mu_z^{(n_t),(i)} \right) \exp \left\{\mu_z^{(n_t),(i)}  \right\}}_{=:{\mathfrak B}_z(n_t, M)}\,,
\end{align*}
with the $M$ leading eigenvalues (in terms of absolute value) $\mu_z^{(n_t),(i)} \in (- \infty, 1)$ of $B_z^{(n_t)}$. To lighten the notation, we will omit the superscript $(n_t)$ for the eigenvalues in the following. Note, however, that we only have access to the eigenvalues of $B_z^{(n_t)}$ in practice -- which can be expected to converge to those of $B_z$ as $n_t \to \infty$ -- and hence there are two sources of error here through finite $M$ and finite $n_t$ (and a third one through finite $n$ for semi-discretized SPDEs), cf.\ remark~\ref{rem:error} below. To actually compute the leading eigenvalues $\mu_z^{(i)}$ in a scalable way, we use an iterative eigenvalue solver as in~\cite{schorlepp-tong-grafke-etal:2023}, relying only on matrix-vector products. Specifically, we use the implicitly restarted Arnoldi method of ARPACK~\cite{lehoucq-sorensen-yang:1998}, but randomized algorithms could also be employed~\cite{halko-martinsson-tropp:2011}.
Importantly, the number $M$ of necessary eigenvalues for a good approximation of the CF determinant can be considered as fixed, i.e.\ independent of $n_t$, as soon as all relevant modes are resolved by the time discretization. Iterative eigenvalue solvers then typically require a number of matrix-vector evaluations that is linear in $M$, but will be independent of the resolution $n_t$. All of this is discussed in~\cite{schorlepp-tong-grafke-etal:2023}; except for the fact that we evaluate CF determinants instead of Fredholm determinants, the same statements apply here. We find empirically in the examples considered in this paper that taking $M$ of the order of a few hundred eigenvalues is typically sufficient to use in practice.

\begin{remark}
\label{rem:error}
It would be valuable to establish rigorous error estimates of the approximate CF determinant in terms of $n_t$ and $M$, or explicitly incorporate the asymptotic behavior of eigenvalues $\mu_z^{(i)} \to 0$ as $i \to \infty$ into the determinant estimation: We can decompose the ratio between the approximate and true CF determinant
${\mathfrak B}_z(n_t, M) / {\det}_2 \left(\text{Id} - B_z \right) =   \epsilon_z(n_t, M) r_z(M)$
into a remainder term (with cutoff $\bar{\mu} = \bar{\mu}(M)$)
\begin{align*}
r_z(M) = \exp \left\{-\sum_{i=M+1}^\infty \log\left(1 - \mu_z^{(i)} \right) + \mu_z^{(i)} \right\} = \exp \left\{-\int_{-\bar{\mu}}^{\bar{\mu}} \left(\log(1-\mu) + \mu \right) \dd \rho_z(\mu) \right\}\,,
\end{align*}
accounting for the truncation after $M$ eigenvalues, and an error term
\begin{align*}
\epsilon_z(n_t, M) = \exp \left\{\sum_{i=1}^M \log\left(1 - \mu_z^{(i),(n_t)} \right) - \log\left(1 - \mu_z^{(i)} \right) + \mu_z^{(i),(n_t)} - \mu_z^{(i)}  \right\}
\end{align*}
which corresponds to the time discretization error for the first $M$ eigenvalues. 
Here, $\rho_z$ denotes the spectral measure of~$B_z$, and $\bar{\mu} > 0$ is a threshold such that $\left \lvert \mu_z^{(1)} \right \rvert, \dots, \left \lvert \mu_z^{(M)} \right \rvert > \bar{\mu}$, and all further eigenvalues satisfy $\left \lvert \mu_z^{(i)} \right \rvert < \bar{\mu}$ for $i > M$.
Knowledge of the asymptotic behavior of $\rho_z(\mu)$ as $\mu \to 0$ could yield useful truncation criteria to determine~$M$, and allow for correcting for eigenvalues $\mu_z^{(i)}$, $i > M$, by integrating the asymptotic spectral density~\cite{garcia-hofmann:2024}.
Regarding the discretization error~$\epsilon_z(n_t, M)$, \cite{atkinson:1975} shows that for compact integral operators with eigenvalues of multiplicity $1$, the order of convergence of eigenvalues of discretized approximations of the operator to the true eigenvalues is inherited from the quadrature rule used.
Suppose, e.g., $\mu_z^{(i),(n_t)} = \mu_z^{(i)} (1 + c_z^{(i)} \Delta t + O(\Delta t^2))$ for a first-order discretization, then $\epsilon_z(n_t, M) = \exp \left\{-\left(\sum_{i=1}^M c_z^{(i)} \left(\mu_z^{(i)} \right)^2 / \left(1 - \mu_z^{(i)} \right) \right) \Delta t + O(\Delta t^2) \right\}$. Another source of discretization errors is the instanton itself, found from discrete versions of the control problem~\eqref{eq:min-prob-eta}, which enters into the coefficient functions of $B_z$. Already on the level of the rate function, this leads to an error of, e.g., $I^{(n_t)}(z) = I(z) + O(\Delta t)$ (cf.~\cite{dontchev-hager:2001}), which may dominate the overall numerical error for small~$\varepsilon$ when computing $\exp\{-I^{(n_t)}(z) / \varepsilon \}$ for~\eqref{eq:tail-prefac-projected}. We leave a more detailed error analysis as future work.
\end{remark}

Lastly, for the estimation of the operator trace $\trace \left[D_z \right]$ in a matrix-free way, a popular way of approximating the trace is through Monte Carlo methods using the Hutchinson estimator~\cite{hutchinson:1989}
$
\trace \left[D_z^{(n_t)} \right] \approx M^{-1} \sum_{m = 1}^M \left \langle \xi_m, D_z^{(n_t)} \xi_m \right \rangle \Delta t
$
with $M$ discrete white noise realizations $\xi_m \sim {\cal N}\left(0, (\Delta t)^{-1} \text{Id}_{n_t \cdot n} \right)$. It turns out that for the predator-prey example from section~\ref{sec:predprey} that we will further elaborate on in section~\ref{sec:pred-prey-continuation}, because of the approximate pairing symmetry of positive and negative eigenvalues of $D_z$, the Hutchinson estimator has a high variance, and we will hence just approximate $\trace \left[D_z^{(n_t)} \right]$ in terms of the sum of $M$ leading eigenvalues of $D_z^{(n_t)}$ for this paper, which are computed analogously to the previous paragraph. Recently introduced improved randomized trace estimators, such as Hutch++~\cite{meyer-musco-musco-etal:2021,persson-cortinovis-kressner:2022} or XTRACE~\cite{epperly-tropp-webber:2024} would also be interesting alternatives to the simple leading eigenspace projection method we use here.
Remark~\ref{rem:error} applies, mutatis mutandis, to the estimation of $\trace \left[D_z \right]$ and associated errors as well.

\subsection{Implementation via automatic differentiation}
\label{sec:implement}

While it does not necessarily yield optimal performance without further fine-tuning, we will illustrate our simple implementation, available under~\cite{Schorlepp-github}, of the strategy outlined in the previous subsection in JAX~\cite{bradbury-frostig-hawkins-etal:2018} by leveraging automatic differentiation. Despite the considerable effort required to \textit{prove} and provide intuition for theorem~\ref{thm:prefac}, \textit{evaluating} the asymptotic tail probability estimate~\eqref{eq:tail-asymp} with prefactor~\eqref{eq:tail-prefac-projected} numerically can actually be rather straightforward in just a few lines of Python code.\\

All we require in principle is an automatically differentiable implementation of the discretized forward map $F$ from noise to observable~\eqref{eq:noise-to-obs-mult}, that takes as its input the discretized noise vector $\eta$ in $\RR^{n \cdot n_t}$, solves $\dot{\phi} = b(\phi) + \sigma(\phi) \eta$ starting from given $\phi(0) = x$ with any choice of time stepper (e.g.\ forward Euler steps as the simplest option), and returns the real number~$f(\phi(T))$. Collecting this procedure in a function \texttt{integrate\_forward\_obs\_jax}, to find the instanton from~\eqref{eq:min-prob-eta} for a given $z$, we implement~\eqref{eq:augmented-functional} as a function \texttt{target\_func}, and use \texttt{target\_func\_grad = jax.jacrev(target\_func)} to get the gradient automatically. By supplying these two functions to \texttt{scipy.optimize.minimize}, we easily have access to the instanton $\eta_z$ and Lagrange multiplier $\lambda_z$ as we increase the penalty parameter $\mu$. Note that the reverse-mode automatic differentiation solves exactly a discrete version of~\eqref{eq:first-order-adjoint-mult-noise} to find the gradient, but we do not need to implement this manually.\\

Given $\eta_z$ (\texttt{eta}) and $\lambda_z$ (\texttt{lbda}), defining the action of the matrix $A_{\lambda_z}^{(n_t)}$ on vectors $\delta \eta$ (\texttt{deta}) is achieved in one line of code:
\begin{align*}
&\texttt{Adeta = lambda deta: lbda/dt *}\\
&\qquad \qquad \quad \texttt{ jax.jvp(jax.grad(integrate{\_}forward{\_}obs{\_}jax),(eta,),(deta,))[1]}\,,
\end{align*}
where the only noteworthy point is the factor of $1/\Delta t$ for a time stepper with equidistant time steps $\Delta t = T / n_t$, to account for the discretization $\delta / \delta \eta \approx (\Delta t)^{-1} \nabla_\eta$ of the functional derivative. The function \texttt{Adeta}, composed with projections~\eqref{eq:proj}, is then wrapped as a \texttt{scipy.sparse.linalg.{\allowbreak}LinearOperator} and plugged into the iterative eigenvalue solver \texttt{scipy.sparse.linalg.eigs} to approximate ${\det}_2 \left(\Id - B_z \right)$.\\

For the trace computation, to implement $\left( A_{\lambda_z} - \tilde{A}_{\lambda_z} \right) \delta \eta$ automatically, we use the $A_{\lambda_z} \delta \eta$ implementation that is already available, and subtract $\tilde{A}_{\lambda_z} \delta \eta$ according to~\eqref{eq:atilde-return}. To obtain $\tilde{A}_{\lambda_z} \delta \eta$ -- which is also needed to calculate $\left \langle e_z, \tilde{A}_{\lambda_z} e_z \right \rangle_{L^2}$ in~\eqref{eq:tail-prefac-projected} -- from automatic differentiation of our implementation of the ODE solution map $\eta \mapsto \phi[\eta]$, we first differentiate $\eta \mapsto \left \langle \theta_z, \sigma(\phi[\eta]) \delta \eta \right \rangle_{L^2}$ with \texttt{jax.jvp} at $\eta$ in the tangent direction $\delta \eta$ to get $\left \langle \delta \eta,  \left \langle \theta_z,\left(\nabla \sigma(\phi) \cdot \right) \gamma^{(1)} \right \rangle \right \rangle_{L^2}$, and then do another \texttt{jax.grad} (divided by $\Delta t$) with respect to $\delta \eta$ on this expression to get $\tilde{A}_{\lambda_z} \delta \eta$. While this strategy does rely on having easy access to the optimal conjugate momentum $\theta_z$, e.g.\ from solving $\eta_z = \sigma(\phi_z)^\top \theta_z$ for invertible $\sigma$, as long as this is the case, it means that the \textit{only} system-specific function that needs to be implemented for a given SDE is the forward map~$F$, and everything else is generic.

\section{Examples}
\label{sec:examples}

We present numerical results of applying theorem~\ref{thm:prefac} in two examples in this section: First, in subsection~\ref{sec:pred-prey-continuation}, we show further results for the predator-prey model that was already considered in section~\ref{sec:predprey}. Afterwards, we move on to a high-dimensional problem in subsection~\ref{sec:advec-diffuse}: Estimating the probability of high concentrations of a passive scalar in a two-dimensional advection-diffusion model, where the advecting velocity field has a random component.

\subsection{More on the predator-prey model}
\label{sec:pred-prey-continuation}

Having introduced and derived theorem~\ref{thm:prefac} now, we return to the introductory predator-prey example from section~\ref{sec:predprey} here. With the same parameters as given in section~\ref{sec:predprey}, we show the results of estimating tail probabilities of high prey concentrations for different thresholds $z$ in figure~\ref{fig:pred-prey-prob}, with $9 \cdot 10^8$ samples for each noise strength $\eps \in \{0.001, 0.004, 0.01 \}$ considered. For a range of $43$ equidistant values of $z \in [0.15, 1]$, we calculate the instanton and evaluate the corresponding prefactor estimate~\eqref{eq:tail-prefac-projected} as described in the previous section. The rate function $I$ and leading prefactor $C$ for the predator-prey system are shown in the left column of figure~\ref{fig:pred-prey-prob}, and the resulting tail probability estimates in the right subfigure. We see that the asymptotic estimate from~\ref{thm:prefac} describes the tail probabilities, as obtained from Monte Carlo simulations, very well, and allows to estimate them far beyond the directly observable regime for each $\eps$. In particular, with regard to our discussion of convexity of the rate function in appendix~\ref{sec:convex-rf}, we see that $I''(z)<0$ for a range of $z$'s in the present example, but that this does not pose a problem for the asymptotic estimate according to~\eqref{eq:tail-prefac-projected}. More detailed numerical results for one particular observable value of $z = 1$ are shown in figure~\ref{fig:pred-prey-spec} and table~\ref{tab:pred-prey-z-1.0}. In particular, figure~\ref{fig:pred-prey-spec} nicely shows the effect of subtracting $\tilde{A}_{\lambda_z}$ from the full Hessian $A_{\lambda_z}$, which, in agreement with the theory, results in a TC operator instead of HS.

\begin{figure}
\centering
\includegraphics[width = .8\textwidth]{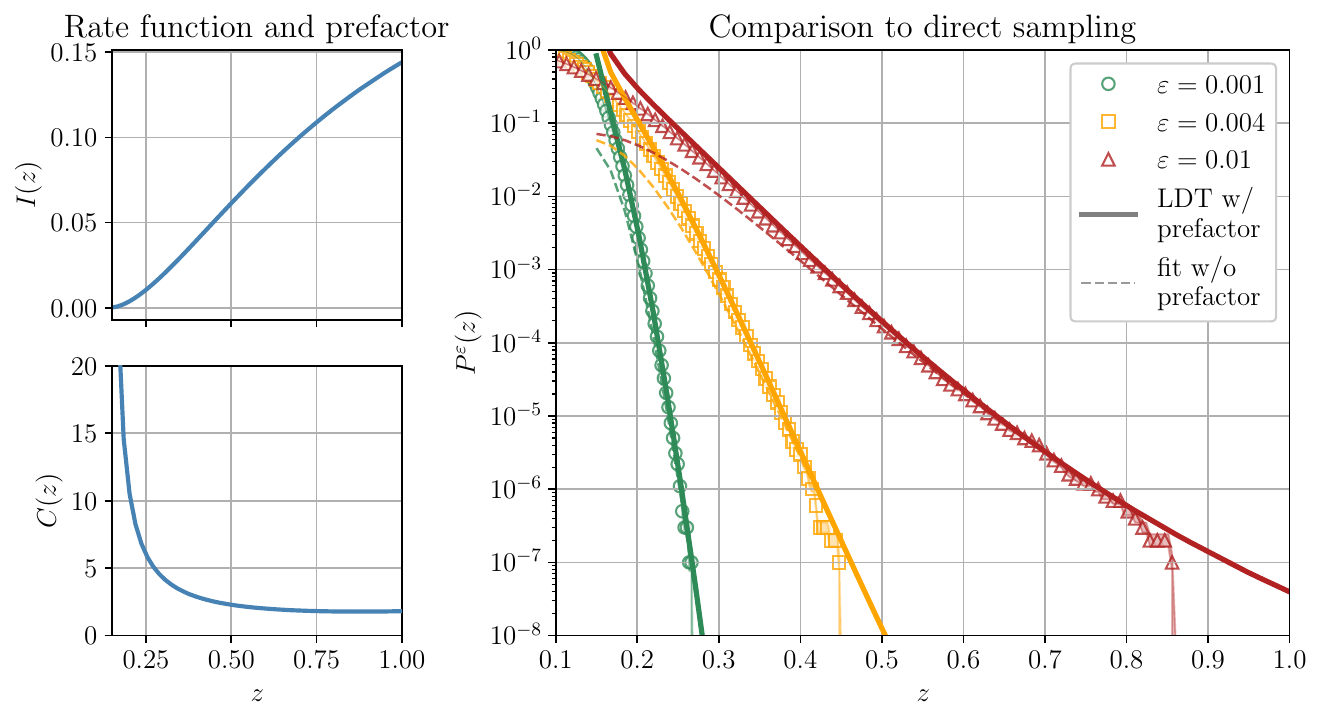}
\caption{Rate function $I$ and prefactor $C$ (left column) as well as probability $P^\eps(z)$ (right) of observing a large
    concentration of prey $f(x(T), y()T)) = x(T)$ in the Lotka--Volterra
    model~(\ref{eq:predator-prey-SDE}), exceeding $z$ at time $t=T$
    when starting at the fixed point $(x_0, y_0)$ of the system. The right subfigure shows a comparison between
    the sharp large deviation estimate~\eqref{eq:tail-asymp} from theorem~\ref{thm:prefac} (solid lines) with
    prefactor according to~\eqref{eq:tail-prefac-projected} against Monte Carlo estimates (data points, with $99\%$ Wilson score intervals as shaded area) with
    $9 \cdot 10^8$ samples for each noise strength $\eps \in \{0.001, 0.004, 0.01\}$. Temporal resolution $n_t = 1000$ for sampling and all asymptotic estimates, and all other parameters are as described in section~\ref{sec:predprey}. The dashed lines are $\text{const} \cdot \exp \left\{-I(z) / \eps \right\}$ with $\text{const}$ chosen in such a way that the curves match at large $z$. This is the best we can do without knowledge of the prefactor, and is not an a priori estimate without Monte Carlo data. The result is not too different from the full prediction including the prefactor $C(z)$ in this particular example, since $C(z)$ does not depend strongly on $z$ for large $z$ in this example.}
\label{fig:pred-prey-prob}
\end{figure}

\begin{figure}
\centering
\includegraphics[width = \textwidth]{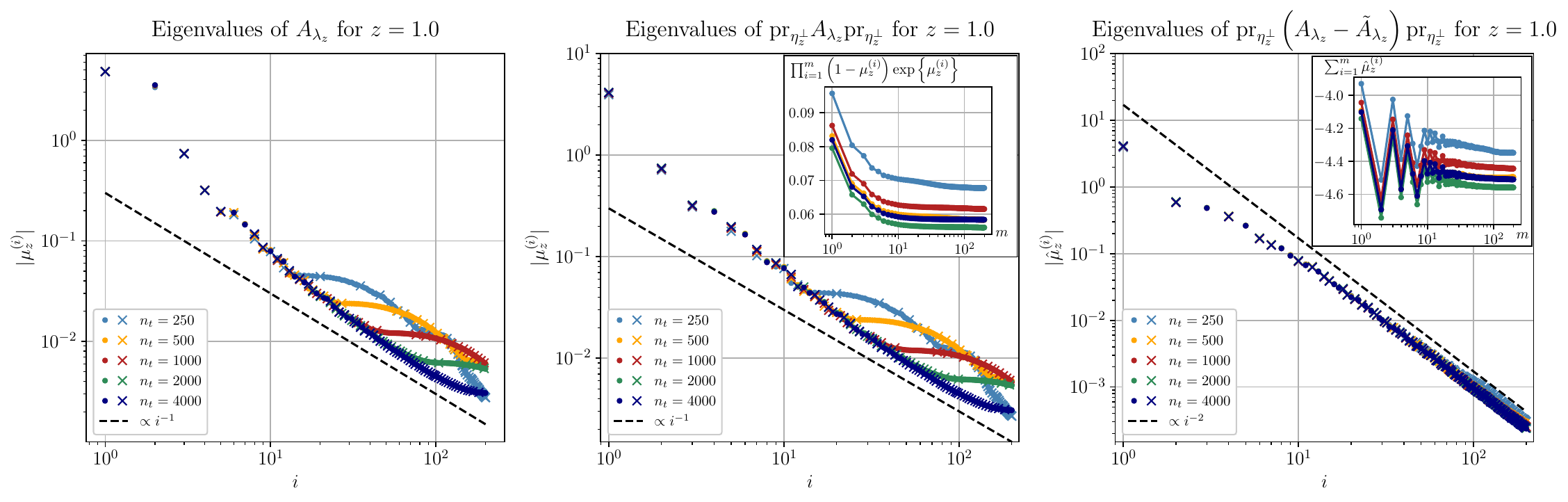}
\caption{Numerical results for the leading $M=200$ eigenvalues of the spectra of the operators that are needed to evaluate the tail probability estimate~\eqref{eq:tail-asymp} from theorem~\ref{thm:prefac} for the predator-prey model~\eqref{eq:predator-prey-SDE} for an observable value of $f(x,y) = z = 1.0$ at different time resolutions $n_t$. In all spectra, dots correspond to positive eigenvalues and crosses to negative eigenvalues.
Left: Spectrum of the Hessian $A_{\lambda_z}$ without projection operators. Note that $I''(z = 1) < 0$ in this example, as evident from  figure~\ref{fig:pred-prey-prob}, and there is hence indeed a single eigenvalue $\mu_z^{(2)} > 1$ in the spectrum, and $\det_2 \left(\text{Id} - A_{\lambda_z}\right) < 0$. Center: Spectrum of the projected Hessian $\text{pr}_{\eta_z^\perp}A_{\lambda_z} \text{pr}_{\eta_z^\perp}$, as it actually enters into the prefactor formula~\eqref{eq:tail-prefac-projected}. The projection operators change the spectrum and in particular remove the eigenvalue $>1$, and hence the CF determinant shown in the inset is positive now. Right: Subtracting the HS part $\tilde{A}_{\lambda_z}$ from the full Hessian leads to a trace-class operator $\text{pr}_{\eta_z^\perp} \left( A_{\lambda_z} - \tilde{A}_{\lambda_z} \right) \text{pr}_{\eta_z^\perp}$ with eigenvalue decay $\propto i^{-2}$ instead of $\propto i^{-1}$, and the trace is indeed seen to converge in the inset. See table~\ref{tab:pred-prey-z-1.0} for further details.}
\label{fig:pred-prey-spec}
\end{figure}

\begin{table}
\caption{Numerical results for all necessary quantities to evaluate the tail probability estimate~\eqref{eq:tail-asymp} for the predator-prey model~\eqref{eq:predator-prey-SDE} for an observable value of $f(x,y) = z = 1.0$ at different time resolutions $n_t$ in tabular form. Here, $I''(z) < 0$ and $\det_2 \left(\text{Id} - A_{\lambda_z}\right) < 0$, but~\eqref{eq:tail-prefac-projected} can still be used to calculate the well-defined tail probability prefactor $C(z)$. All quantities in this table should converge to a well-defined continuum limit as $n_t$ increases. $M = 200$ eigenvalues were used to approximate all of the listed determinants and traces.}
\label{tab:pred-prey-z-1.0}
\centering
\begin{tabular}{|c|c|c|c|c|}
\hline
$n_t$ & $\lambda_z$ & $I(z)$ & $\det_2 \left(\text{Id} - A_{\lambda_z}\right)$ & $\det_2 \left(\text{Id} -\text{pr}_{\eta_z^\perp} A_{\lambda_z} \text{pr}_{\eta_z^\perp}\right)$ \\[.2cm]
\hline
$250$  & $0.115443$ & $0.139546$ & $-2.353707$ & $0.067781$ \\
$500$  & $0.119596$ & $0.142610$ & $-2.689994$ & $0.058361$ \\
$1000$ & $0.117907$ & $0.144161$ & $-2.883130$ & $0.061549$ \\
$2000$ & $0.121072$ & $0.144943$ & $-2.985856$ & $0.056080$ \\
$4000$ & $0.119295$ & $0.145329$ & $-3.024294$ & $0.058393$ \\\hline
\end{tabular}\\[.1cm]
\begin{tabular}{|c|c|c|c|}
\hline
$n_t$ & $\trace \left[\text{pr}_{\eta_z^\perp} \left( A_{\lambda_z} - \tilde{A}_{\lambda_z} \right) \text{pr}_{\eta_z^\perp}\right]$& $\left \langle e_z, \tilde{A}_{\lambda_z} e_z \right \rangle_{L^2}$ & $C(z)$\\[.2cm]
\hline
$250$  & $-4.348229$ & $-1.537991$ & $1.783767$\\
$500$  &  $-4.494852$ & $-1.581076$ & $1.805635$\\
$1000$ &  $-4.444308$ & $-1.600456$ & $1.810986$\\
$2000$ &  $-4.558893$ & $-1.611645$ & $1.796790$\\
$4000$ &  $-4.509487$ & $-1.619085$ & $1.809190$\\\hline
\end{tabular}
\end{table}

\subsection{Advection-diffusion equation}
\label{sec:advec-diffuse}

The scalability of the algorithm is of particular importance when the
total number of degrees of freedom is very large. The typical example
of this is a semidiscretization of a stochastic PDE, especially if the
PDE has multiple spatial dimensions on top of the temporal one. We  consider the transport of a
pollutant in a $(d = 2)$-dimensional spatial domain $\Omega\subset\RR^2$ over
time. Emitting the pollutant at a localized point in the domain
$x_{\text{inj}}\in\Omega$ with a constant rate, we are interested in estimating how random fluctuations
in a background velocity field lead to an exceedance of the pollutant
concentration of some threshold $z>0$ at a target point $x_{\text{meas}}\in\Omega$. Concrete
examples in this class of problems might be to estimate the risk that
an oil spill in the ocean reaches a specific point on the coast~\cite{yang-khan-lye:2013,ji-johnson-wikel:2014,zhang-schaefer-kress:2021}, or
that nuclear fallout from a disaster area reaches a city.\\

For this purpose, we consider the following SPDE for a scalar
concentration field $c^\eps:\Omega\times[0,T]\to\RR$ given by
\begin{align}
  \label{eq:advdiff-spde-strat}
  \begin{cases}
    \partial_t c^\eps = -(v\cdot\nabla) c^\eps - \sqrt{\eps} (w\odot\nabla) c^\eps + D_0 \Delta c^\eps + s\,,\\
    c^\eps(\cdot, 0) = 0\,,
  \end{cases}
\end{align}
where $v:\Omega\to\RR^2$ is a divergence-free background velocity field, $D_0>0$ the
pollutant's (bare) diffusivity, $s:\Omega\to\RR$ is the pollutant source,
$\odot$ is the Stratonovich dot product, and
$w:\Omega\times[0,T]\to\RR^2$ is a white-in-time, divergence-free stationary Gaussian
random vector field, with
\begin{align*}
\EE \left[w(x,t) \right] = 0\,, \quad \EE \left[w(x,t) \otimes w(x', t')  \right] = R_w(x-x') \delta(t-t')\,.
\end{align*}
The Stratonovich interpretation is chosen here because it corresponds to
the limit of infinite time scale separation between the velocity field
$w$ being much faster than the concentration field
$c$~\cite{donev-fai-vanden-eijnden:2014}. For simplicity, we consider
$\Omega=[-L/2,L/2]^2$ with periodic boundary conditions, and $L=2\pi$.\\

In contrast to the It{\^o} SDE~\eqref{eq:SDE}, equation~(\ref{eq:advdiff-spde-strat}) is a Stratonovich-type S(P)DE, see remark~\ref{rem:ito-strato} for the prefactor computation. In the general notation used in this paper, and with $R_w^{1/2} * R_w^{1/2} = R_w$, where $*$ denotes convolution, the correction term in the drift, needed to convert between the It{\^o} and Stratonovich versions of~(\ref{eq:advdiff-spde-strat}), for the diffusion operator $\sigma$ with
\begin{align*}
\left(\sigma[c]\eta\right)(x) &=
-\left( \left(R_w^{1/2} * \eta \right)(x)\cdot \nabla\right) c(x) \\&= - \sum_{j,k = 1}^2 \int_\Omega \dd x' \; R_{w,jk}^{1/2}(x-x') \partial_j c(x) \eta_k(x') =: \sum_{k = 1}^2 \int_\Omega \dd x' \tilde{\sigma}_k(x,x') \eta_k(x')\,,
\end{align*}
becomes
\begin{align*}
  \frac{\eps}{2} \sum_{k = 1}^2 \int_\Omega \dd x' \int_\Omega \dd x'' \; \tilde{\sigma}_k[c](x',x'') \frac{\delta \tilde{\sigma}_{k}}{\delta c(x')} (x,x'')  = \frac{\eps}{2} R_w(0) : \nabla^2 c(x)
\end{align*}
after a short calculation. The It{\^o} version of the SPDE~(\ref{eq:advdiff-spde-strat}) is hence given by
\begin{align}
  \label{eq:advdiff-spde-ito}
  \begin{cases}
    \partial_t c^\eps = -(v\cdot \nabla)c^\eps -\sqrt{\eps}(w\cdot\nabla)c^\eps + \left(D_0 I_d + \tfrac12\eps R_w(0)\right):\nabla^2 c^\eps + s\,,\\
    c^\eps(\cdot, 0) = 0\,.
  \end{cases}
\end{align}
In particular, as is well-known, the It{\^o}--Stratonovich correction term ``renormalizes'' the diffusivity in this example, resulting in an $\eps$-dependent effective diffusivity.\\

For the Gaussian random velocity field $w$, we choose a spatially
smooth field, because otherwise the SPDE~(\ref{eq:advdiff-spde-ito})
may become singular and would require renormalization. Note that such a smooth forcing, effectively only in terms of a few nonzero Fourier modes, corresponds to a low-rank diffusion matrix $\sigma$, which we will exploit in our code for both the instanton and prefactor computation, in the same way as detailed in~\cite{schorlepp-tong-grafke-etal:2023}. For a symmetric,
isotropic, and divergence-free smooth velocity field, the correlation function
$R_w\colon \Omega \to \RR^{d \times d}$ must necessarily be of the form~\cite{robertson:1940}
\begin{align*}
R_{w,ij}(x) = r(\norm{x}) \delta_{ij} + \frac{\norm{x} r'(\norm{x})}{d -1} \left(\delta_{ij} - \frac{x_i x_j}{\norm{x}^2} \right)\,,
\end{align*}
for a scalar function $r \colon [0,\infty) \to \RR$. We use $r(\norm{x}) = R_{w,0} \exp \left\{-\norm{x}^2 / (2 L_w^2)
\right\}$ as a generic choice, with a typical length scale $L_w>0$ of the random velocity
field $w$, such that
\begin{align}
R_{w,ij}(x) = R_{w,0} \exp \left\{-\frac{\norm{x}^2}{2 L_w^2} \right\} \left[\delta_{ij} - \frac{1}{d-1} \left(\frac{\norm{x}^2}{L_w^2} \delta_{ij} - \frac{x_i x_j}{L_w^2} \right) \right]\,.
\label{eq:forcing-correl-mex}
\end{align}

For both the source $s \colon \Omega \to \RR$ and the observable $f\colon \Omega \to \RR$, we introduce a
Gaussian mollifier on a small length scale $\ell$,
$\varphi_\ell(x) = (\pi\ell^2)^{-1}\exp\{-\norm{x}^2/\ell^2\}$,
and consider a constant and unit-rate source $s$ localized around point $x_{\text{inj}}\in\Omega$ on a scale $\ell$, $s(x) = \varphi_\ell(x-x_{\text{inj}})$,
and an observable $f$ sensitive around a point $x_{\text{meas}}\in\Omega$, also on the scale
$\ell$, so that the observable becomes
\begin{align}
  f(c(\cdot,T)) = (\varphi_\ell * c(\cdot,T))\left( x=x_{\text{meas}} \right)\,.
  \label{eq:adv-diff-obs}
\end{align}
This choice ensures that pollutant is injected into the domain at a
constant rate in a region of size $\ell$ around $x_{\text{inj}}\in\Omega$,
and we are asking for its concentration averaged over an
$\ell$-sized region around the point of measurement $x_{\text{meas}}\in\Omega$ at time $T$.\\

\begin{figure}
  \begin{center}
    \includegraphics[width=0.42\textwidth]{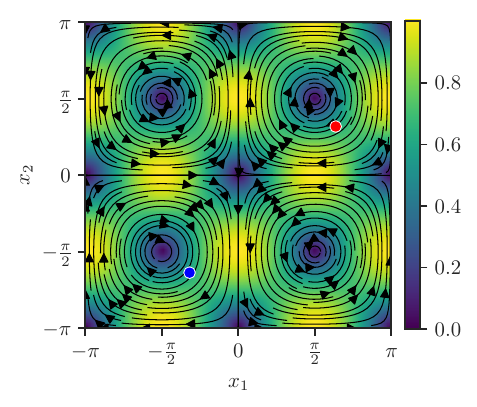}
    \includegraphics[width=0.42\textwidth]{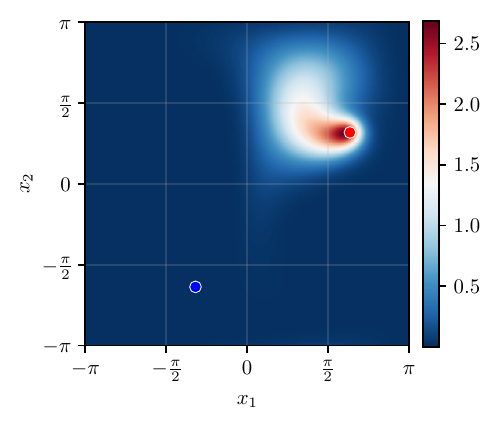}
  \end{center}
  \caption{Left: Streamlines and magnitude of the deterministic velocity field
    $v$ that advects the pollutant. There are four disconnected cells of flow, and traversal
    between them is only possible due to diffusion and
    noise in the model~\eqref{eq:advdiff-spde-strat}. Pollutant is constantly emitted at the source $x_{\text{inj}}$ (red), and we
    are monitoring its concentration at the target $x_{\text{meas}}$ (blue). Right: The
    shading shows the concentration of pollutant~$c^0(\cdot, T)$ at $T=5$ for the
    deterministic case, $\eps=0$, where no additional random velocity
    is present. Brighter values correspond to higher
    concentration, as indicated by the colorbar. Grid lines are shown to help visualize the quadrants/flow cells. The pollutant can only slowly leave its cell, since
    the transport is advection-dominated, and at time $T$ almost nothing of
    the concentration has left the upper right quadrant.}
  \label{fig:advdiff-setup}
\end{figure}

With these concrete choices, we perform numerical simulations of the SPDE~\eqref{eq:advdiff-spde-ito} with parameters
\begin{align*}
  v(x_1,x_2) = \nabla^\perp (\sin x_1\sin x_2) = 
  \begin{pmatrix}
    -\sin x_1\cos x_2\\
    \cos x_1\sin x_2
  \end{pmatrix}\,, \quad x_{\text{inj}} = 
  \begin{pmatrix}
    2\\1
  \end{pmatrix}\,, \quad x_{\text{meas}} =
  \begin{pmatrix}
    -1\\-2
  \end{pmatrix}\,,
\end{align*}
as well as $\eps \in \{ 0.05, 0.1\}$, $L_w=1.0$, $\ell=0.2$, $D_0=5\cdot
10^{-2}$, $T = 5$. The deterministic velocity field $v$ is shown in figure~\ref{fig:advdiff-setup} (left). In particular, this
choice of background velocity field $v$ introduces four independent
vortices that can only be bridged by diffusion and noise. The
pollutant emitted at the source around $x_{\text{inj}}$ (red) can only reach the target around $x_{\text{meas}}$ (blue)
under favorable realizations of the noise. If noise is absent, as in
figure~\ref{fig:advdiff-setup} (right), diffusion alone leads to
essentially no pollutant leaving the upper-right quadrant in the specified time interval.\\

For the numerical solution of~\eqref{eq:advdiff-spde-ito}, we use a pseudo-spectral discretization in terms of~$n_x$ spatial Fourier modes for each of the two directions~$x_1$ and~$x_2$, and an Euler-Maruyama integrator with an integrating factor for the diffusion term in Fourier space for time stepping, with~$n_t$ equidistant time steps within the time interval $[0,T]$. The Gaussian random field~$w = R_w^{1/2}*\xi$ that advects the pollutant is sampled directly in Fourier space by means of the convolution theorem, where~$\xi$ is space-time white noise (for implementation details, we refer to~\cite{lang-potthoff:2011,schorlepp-grafke-may-etal:2022}).
While more sophisticated numerical approaches would be possible for the advection-diffusion problem at hand, we focus on these simple choices here to implement the forward simulation from the noise~$\xi$ or~$\eta$ to the final-time observable value $\hat{F}^\eps[\sqrt{\eps} \xi] = f\left( c^\eps(\cdot, T) \right)$ directly in JAX and be able to differentiate through it automatically for the instanton and prefactor computations. Note that if we represent~$\xi$ or~$\eta$ in terms of (the non-redundant real and imaginary parts of its) Fourier modes, because the convolution with the smooth function~$R_w$ strongly dampens all higher modes, it is enough to explicitly retain only the first few Fourier modes of the noise $\eta$ and fluctuations $\delta \eta$. This makes the size of discretized $\eta$ and $\delta \eta$ vectors independent of the spatial resolution, and makes it possible to scale the instanton optimization and eigenvalue computations for the prefactor to high spatio-temporal resolutions, cf.~\cite{schorlepp-tong-grafke-etal:2023}. Concretely, for the forcing correlation~\eqref{eq:forcing-correl-mex} and numerical parameters given above, when keeping only modes with wave numbers $\norm{k} \leq 8$ for which $\abs{\hat{R}_{w,ij}(k)} \geq 10^{-10}$, we store only $2 \cdot (2 \cdot 8)^2 = 512$ spatial degrees of freedom for~$\eta$ and~$\delta \eta$ (which still depend on time), compared to the actual $2 \cdot n_x^2$ spatial degrees of freedom of the system for the forward solve. Consequently, the total number of unknowns of the optimization problem to find the instanton noise~$\eta_z$ is~$512\cdot n_t$, as opposed to the~$2\cdot n_x^2 n_t$ degrees of freedom of the naive discretization.\\

\begin{figure}
  \begin{center}
    \includegraphics[width=0.49\textwidth]{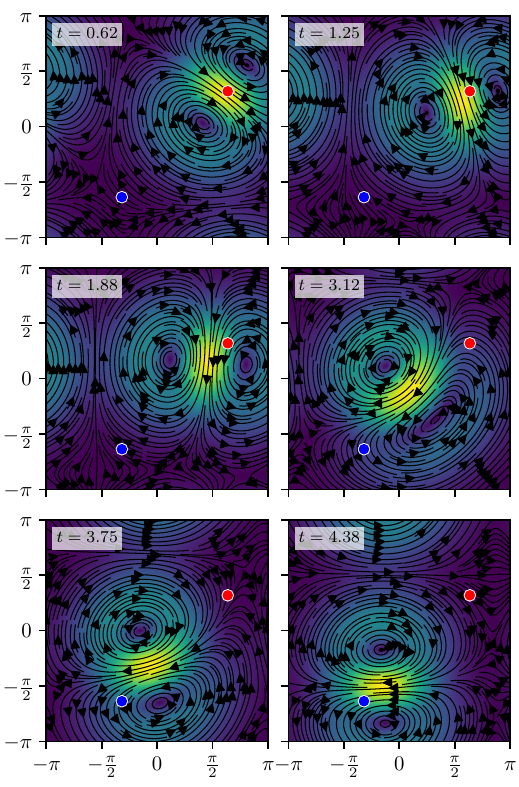}
     \includegraphics[width=0.49\textwidth]{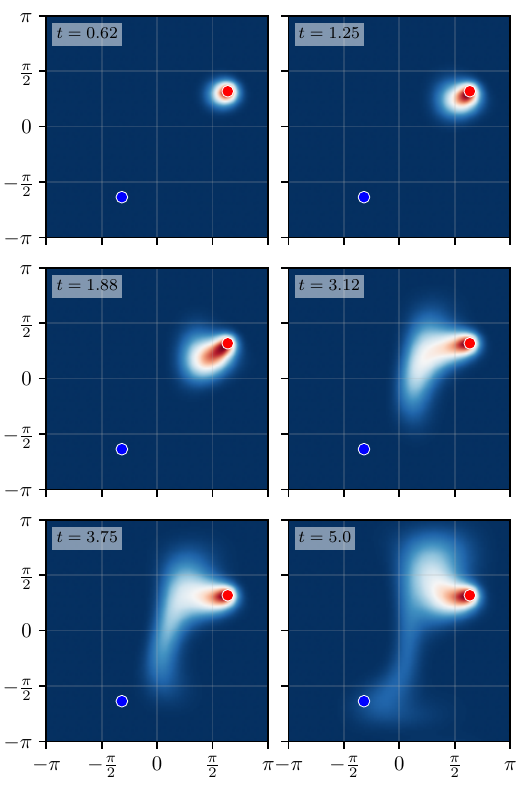}
  \end{center}
  \caption{Visualization of a numerically computed instanton for the SPDE~\eqref{eq:advdiff-spde-ito} for an observable value $f(c(\cdot,T)) = 0.15$ according to~\eqref{eq:adv-diff-obs}.
  Left 6 panels: Vector field (streamlines) and strength (shading) of the optimal velocity field realization $w_z(\cdot,t)$ facilitating transport from source $x_{\text{inj}}$ (red dot) to target $x_{\text{meas}}$ (blue dot): At small times, the forcing primarily pushes pollutant from the top right to the bottom right cell. Once the pollutant has been transported sufficiently by the background field, at late times the forcing advects it to the lower left cell. Right 6 panels: Corresponding evolution in time of the pollutant concentration $c_z(\cdot, t)$ from $t=0$ to $t=T=5.0$: Pollutant is emitted at the source, and reaches the target at final time. Numerical resolutions used for instanton computations: $n_x=128$, $n_t=1024$.}
  \label{fig:inst-time-series}
\end{figure}

\begin{figure}
  \begin{center}
    \includegraphics[height=0.315\textheight]{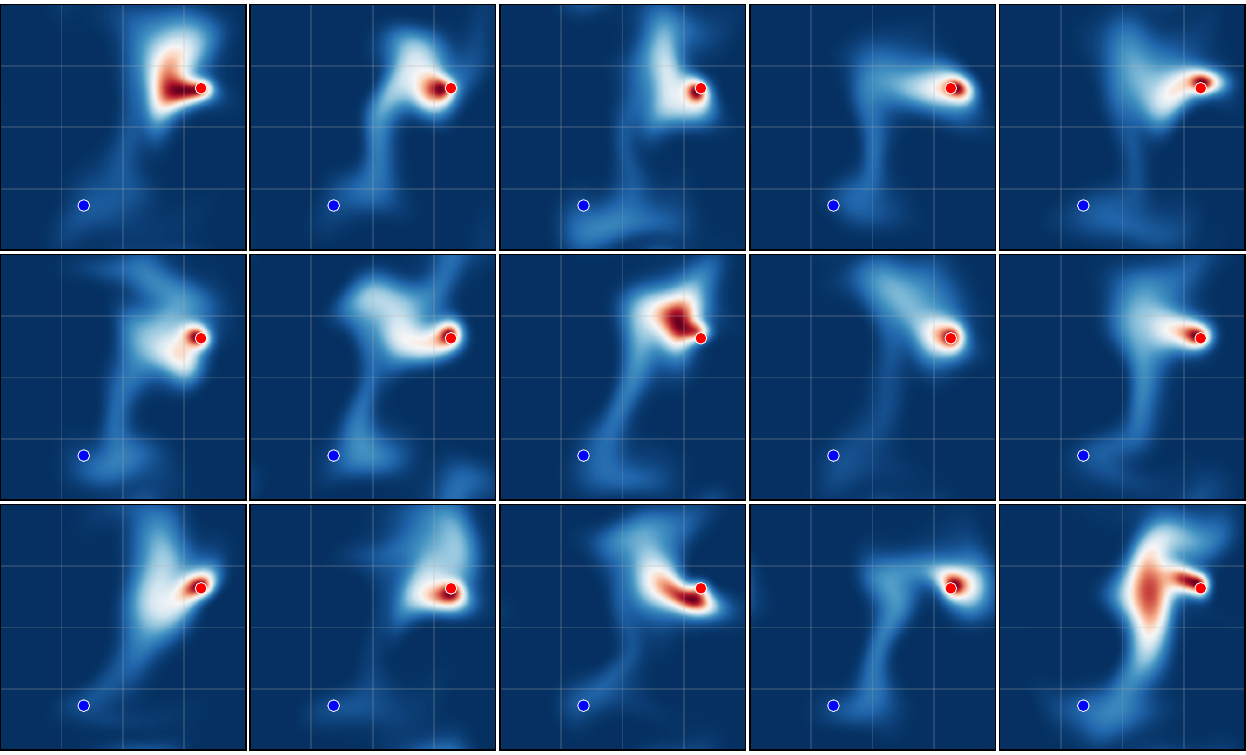}
    \includegraphics[height=0.315\textheight]{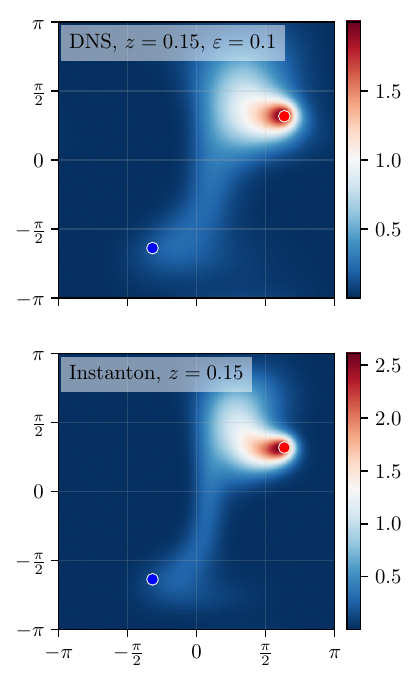}\ 
  \end{center}
  \caption{Left panels: 15 samples $c^{\eps}(\cdot, T)$ of the stochastic
    advection-diffusion equation~(\ref{eq:advdiff-spde-ito}) for
    $\eps=0.1$, conditioned on measuring a concentration of $z=0.15$
    at the target point $x_{\text{meas}}$ (blue) at final time $T=5.0$. Top right: Average of all 185
    sampled realizations of $c^{\eps}(\cdot, T)$ achieving $z=0.15$ at $T=5.0$. Bottom
    right: Concentration field $c_z(\cdot,T)$ of the instanton configuration
    for $z=0.15$. We see that the instanton has a similar structure as the conditional average for the same observable value, even though differences in terms of diffusivity and magnitude are visible. Numerical resolutions used for sampling and instanton computations: $n_x=128$, $n_t=1024$.}
    \label{fig:filter}
\end{figure}

\begin{figure}
  \begin{center}
    \includegraphics[height=0.245\textheight]{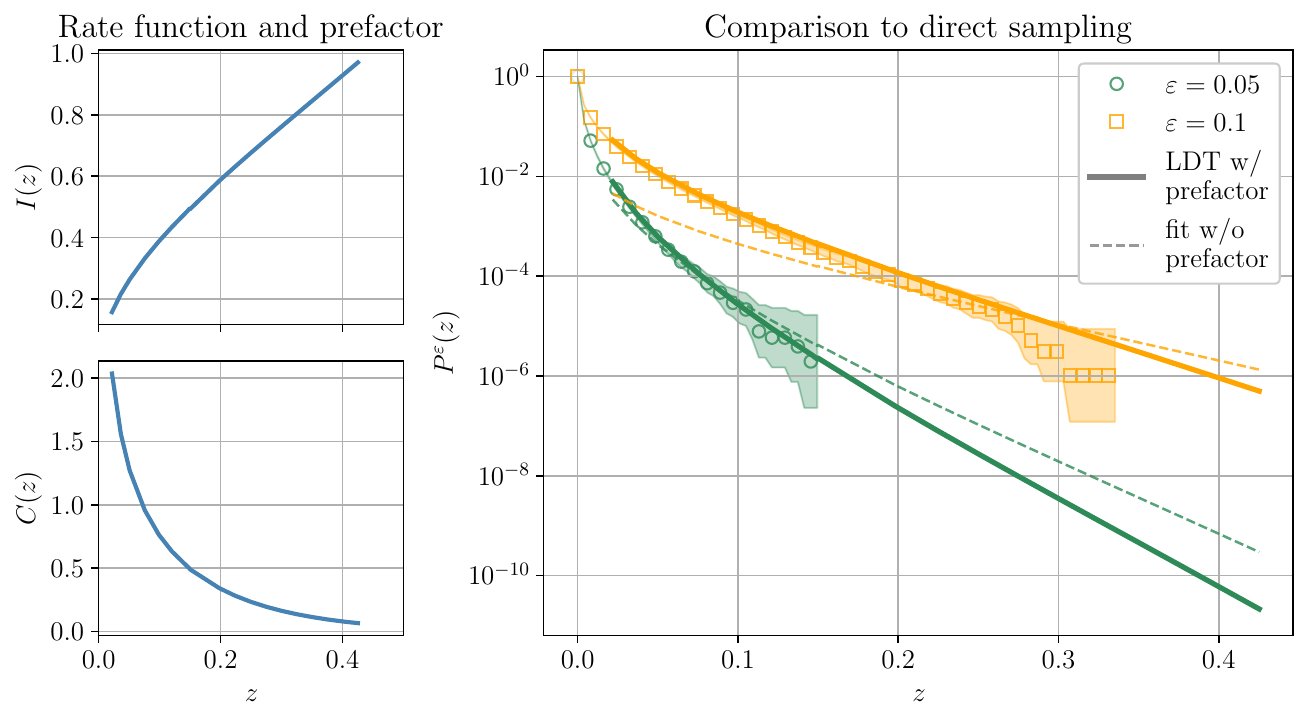}
    \includegraphics[height=0.245\textheight]{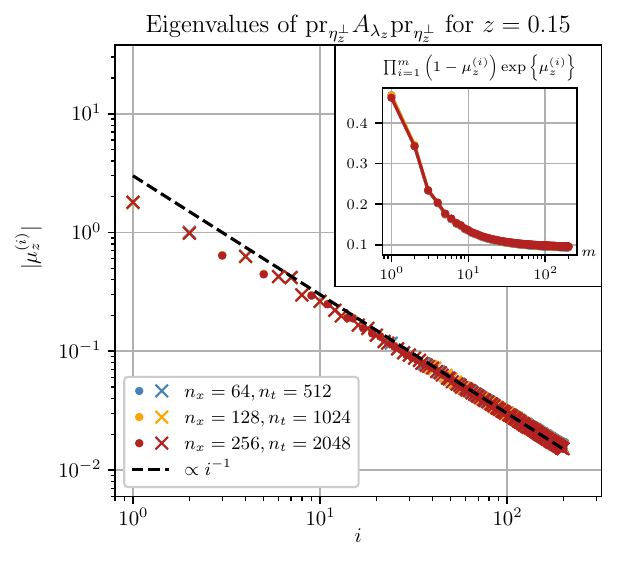}\ 
  \end{center}
  \caption{Rate function $I(z)$ (top left), prefactor $C(z)$ as given by~\eqref{eq:advdiff-prefactor-full} (bottom
    left), and probability $P^\eps(z)$ (center) of the concentration
    exceeding $z$ for noise strength $\eps$ for the stochastic
    advection diffusion equation~\eqref{eq:advdiff-spde-strat}. For the center figure, we compare
    the results of direct numerical simulation of the
    SPDE~(\ref{eq:advdiff-spde-ito}) with noise strengths $\eps=0.1$ with $\approx
    10^6$ samples and $\eps=0.05$ with $\approx 5\cdot10^5$
    samples (data points, with $99\%$ Wilson score intervals as shaded area), against the sharp large deviation
    estimate~\eqref{eq:tail-asymp} from theorem~\ref{thm:prefac}
    (solid lines) with prefactor according
    to~\eqref{eq:advdiff-prefactor-full}. In contrast to this, the dashed lines are $\text{const} \cdot \exp \left\{-I(z) / \eps \right\}$ with $\text{const}$ to match the Monte Carlo results in the tails, clearly showing worse agreement. Numerical resolutions used for these computations: $n_x = 128$, $n_t = 1024$. Also shown are the leading
    $M=200$ eigenvalues of the spectrum (right) of the projected second variation
    operator $\text{pr}_{\eta_z^\perp} A_{\lambda_z}
    \text{pr}_{\eta_z^\perp}$ for $z = 0.15$ at different spatio-temporal resolutions $(n_x, n_t)$, where dots correspond to positive and
    crosses to negative eigenvalues. The inset shows the numerical
    estimate of the CF determinant for $1\le m\le M=200$,
    demonstrating that only the first 200 eigenvalues are needed in
    this problem for the determinant to converge.}
    \label{fig:advdiff-cdf-and-spectrum}
\end{figure}

To evaluate the asymptotic estimate~\eqref{eq:tail-asymp} for the SPDE~(\ref{eq:advdiff-spde-ito}), we first
need the gradient of the noise-to-observable map $F[\eta] = f(c(\cdot, T))$,
which is given by
\begin{align*}
  \frac{\delta (\lambda F)}{\delta \eta} = \sigma(c)^\top \theta = R^{1/2}_w*(c\ \nabla\theta)
\end{align*}
with $(c, \theta)$ the solution to
\begin{align*}
  \begin{cases}
    \partial_t c = -(v\cdot\nabla) c + D_0 \Delta c - \left(\left(R_w^{1/2}*\eta\right)\cdot \nabla\right) c + s\,, & c(\cdot,0)=0\\
    \partial_t \theta =  - \left( v\cdot\nabla \right)\theta - D_0\Delta \theta - \left(\left(R_w^{1/2}*\eta\right)\cdot\nabla\right)\theta\,, & \theta(\cdot,T) = \lambda \left. \fdv{f}{c} \right|_{c(\cdot, T)} = \lambda\varphi_\ell(\cdot -x_{\text{meas}})\,.
  \end{cases}
\end{align*}
Note that here the It{\^o}-Stratonovich correction term disappears, as
formally $\eps=0$. We will need to correct for this term later in the prefactor. As in previous numerical examples, we minimize target functionals~\eqref{eq:augmented-functional} with automatically computed gradients to find instantons for prescribed concentrations $F[\eta] = z$. An instanton configuration that results from such a computation for a concentration of $z=0.15$ at final time $T=5.0$ around $x_{\text{meas}}$ is shown in figure~\ref{fig:inst-time-series} in terms of the concentration $c_z(\cdot, t)$ and advection field $w_z(\cdot, t)$ for different $t$. We can see how the optimal advection carries the pollutant efficiently from source to target by making use of the coherent background velocity field. Furthermore, in figure~\ref{fig:filter}, we compare the concentration field of the instanton for $z = 0.15$ at final time to samples from direct numerical simulation of the SPDE~\eqref{eq:advdiff-spde-ito} at $\eps = 0.1$, conditioned on displaying the same observable value, which is rare for the considered noise strength. We can see that the instanton correctly explains the mechanism for the rare event in the Monte Carlo simulations on a qualitative level. Differences between the instanton and conditional average can be attributed to the different drift terms/diffusivities that the instanton and samples are subject to, with the $\eps$-dependent part of the diffusivity in~\eqref{eq:advdiff-spde-ito} being accounted for only by the prefactor below in the asymptotic approximation of~\eqref{eq:strato-pref-complete}. In other words, even though the instanton is visibly computed at a lower effective diffusivity, the effect of the $\eps$-dependent Stratonovich correction onto the probabilities is taken care of in the prefactor computation.\\

Two technical comments on the instanton computations are in order here: (i) to solve the optimization problems for the instantons at different $z$, we initialize the optimizer with a specifically chosen initial guess for $\eta$ here, which is constant over time and transports $c$ diagonally from top right to bottom left of the domain. This is to facilitate finding the instanton in this example, as starting from $\eta = 0$, or a random guess, can lead to a ``vanishing gradient problem'' where the optimizer remains stuck at close to zero concentration around $x_{\text{meas}}$ at final time. (ii) we observed that for large penalty parameters~$\mu$, the line search method of the optimizer can sometimes struggle to find admissible step sizes in this example, possibly due to insufficient numerical accuracy of the Fourier transforms in the forward map and gradient. As a simple remedy, we note that for \textit{any} penalty parameter $\mu > 0$, minimizing the target functional~\eqref{eq:augmented-functional} gives a valid instanton for its respective observable value; \textit{increasing} $\mu$ is only necessary to reach the prescribed observable value $z$ with higher accuracy. Hence, using smaller $\mu$'s as well as good initial guess for $\lambda$ still allows for computing correct instantons reasonably close to the target observable values (we terminate at $<10\%$ discrepancy), while avoiding the aforementioned numerical problems in the optimization. This has the additional benefit of increasing the efficiency of the optimizer. For a more thorough discussion of implementation details of the augmented Lagrangian method, we refer to~\cite{bertsekas:2014}.\\

Moving on to the prefactor computation, the second variation, or equivalently the action of the operator
$A_\lambda$, is given by
\begin{align*}
  A_\lambda \delta \eta = \nfdv{2}{\left(\lambda F \right)}{\eta} \delta \eta = R_w^{1/2} * \left(\gamma^{(1)} \nabla \theta + c\nabla \zeta \right)
\end{align*}
with $(\gamma^{(1)}, \zeta)$ solution to
\begin{align*}
\begin{aligned}
  &\begin{cases}
    \partial_t \gamma^{(1)} = -(v\cdot\nabla)\gamma^{(1)} + D_0\Delta \gamma^{(1)} -\left( \left(R_w^{1/2} * \eta\right)\cdot\nabla\right) \gamma^{(1)} - \left(\left(R_w^{1/2} * \delta\eta\right)\cdot \nabla\right)c\,,\\
    \partial_t \zeta = -(v\cdot\nabla)\zeta - D_0\Delta\zeta - \left(\left(R_w^{1/2} * \zeta\right)\cdot\nabla\right)\theta - \left(\left(R_w^{1/2} * \eta\right)\cdot\nabla\right)\zeta\,,
  \end{cases}\\
  \text{ with } &\begin{cases}
   \gamma^{(1)}(\cdot,0)=0\,,\\
   \zeta(\cdot,T) = \lambda \left.\nfdv{2}{f}{c}\right|_{c(\cdot,T)} \gamma^{(1)}(\cdot,T) = 0\,.
  \end{cases}
\end{aligned}
\end{align*}
Note that since both the drift term and the noise covariance operator
are only linear in $c$ in equation~(\ref{eq:advdiff-spde-ito}), no
second derivative terms of them appear in the second order adjoint
equation. As a consequence of that, and the linearity of the observable, the action of the singular operator~$\tilde A_\lambda$ is identical to the action of the original operator
$A_\lambda$ here, i.e.\ $A_\lambda = \tilde A_\lambda$, and in particular we have
$\trace[A_\lambda-\tilde A_\lambda]=0$.\\

Due to the It\^o--Stratonovich correction term and as explained in remark~\ref{rem:ito-strato}, for SDEs in $\RR^n$, in general, compared to~\eqref{eq:tail-prefac-projected} and by rewriting~\eqref{eq:strato-pref-complete}, there is an
additional factor
$
  \exp\left\{ \left \langle \lambda\nabla f(\phi(T)),\tilde\gamma^{(1)} \right \rangle\right\}
$
for $\tilde\gamma^{(1)} \colon [0,T] \to \RR^n$ the solution to
\begin{align}
  \label{eq:advdiff-dc-generic}
  \begin{cases}\dot{\tilde\gamma}^{(1)} = \left[ \nabla b(\phi)\big|_{\eps=0} + \nabla \sigma(\phi) \eta \right] \tilde\gamma^{(1)} + \tfrac12 \sigma(\phi) : \nabla \sigma(\phi)\,,\\\tilde\gamma^{(1)}(0)=0\,,
  \end{cases}
\end{align}
in the prefactor, where it is the last term of~\eqref{eq:advdiff-dc-generic} that
accounts for the term of order $\eps$ in~\eqref{eq:strat-to-ito-sde}. For the concrete example of the
stochastic advection-diffusion equation~\eqref{eq:advdiff-spde-ito}, this translates to
\begin{align*}
  \exp\left\{\left\langle \lambda \left.\fdv fc\right|_{c(\cdot,T)},\tilde\gamma^{(1)}\right\rangle_{L^2(\Omega, \RR)}\right\} = \exp\left\{\lambda \left(\varphi_\ell * \tilde \gamma^{(1)}(\cdot,T)\right)(x=x_{\text{meas}})\right\}
\end{align*}
for $\tilde\gamma^{(1)} \colon \Omega \times [0,T] \to \RR$ the solution to
\begin{align*}
\begin{cases}
  \partial_t \tilde\gamma^{(1)} = \left[ D_0\Delta - v\cdot\nabla - \left( R_w^{1/2}*\eta \right)\cdot\nabla \right]\tilde \gamma^{(1)} + \tfrac12 R_w(0): \nabla^2 c\,,\\
  \tilde\gamma^{(1)}(\cdot, 0) = 0\,.
  \end{cases}
\end{align*}
The reason why we write the $\tfrac{1}{2}\int_0^T \trace\left[\sigma^\top(\phi(t))\nabla \sigma(\phi(t)) \theta(t) \right] \dd t$ term of the prefactor~\eqref{eq:strato-pref-complete} in this way is because in our code, which is based on automatic differentiation, we do not have easy access to the full adjoint field $\theta$ for non-invertible $\sigma$. In contrast to this, the expression above in terms of $\tilde{\gamma}^{(1)}$ can be evaluated just from knowledge of the optimal noise $\eta_z$, and the state space instanton field $c_z$ from the forward solver.\\

\begin{table}
  \caption{Numerical results for all necessary quantities to evaluate the tail probability estimate~\eqref{eq:advdiff-prefactor-full} for the advection-diffusion SPDE~\eqref{eq:advdiff-spde-ito} to reach a target concentration of $z \approx 0.15$, at different resolutions in space, $n_x$, and time, $n_t$.
    Since $A_{\lambda} = \tilde A_\lambda$, we omit listing $\trace[A_\lambda - \tilde A_\lambda]=0$. Throughout, $M = 200$ eigenvalues were used to approximate all of the listed determinants. We remark that the iteration over increasing penalty parameters in the augmented Lagrangian method has been stopped once the observable value $z$ is within $10\%$ of the desired observable value to safe computation time and avoid step size selection problems, as discussed in the main text. The last column, $N_{\text{eval}}$ denotes the total number of equation evaluations for the optimization and determinant computation together.}
\label{tab:advdiff-z-0.15}
\centering
\begin{tabular}{|c|c|c|c|c|c|}
\hline
$n_x^2\times n_t$ & $z$ & $\lambda_z$ & $I(z)$ & $\det_2 \left(\text{Id} -\text{pr}_{\eta_z^\perp} A_{\lambda_z} \text{pr}_{\eta_z^\perp}\right)$ \\[.2cm]
  \hline
  $64^2\times512$   & $0.149691$ &  $1.94158$ & $0.492514$ & $0.0942355$\\ 
  $128^2\times1024$ & $0.146423$ &  $1.97597$ & $0.488349$ & $0.0953783$\\ 
  $256^2\times2048$ & $0.149447$ &  $1.95045$ & $0.495423$ & $0.0959704$ \\\hline 
\end{tabular}\\[.1cm]
\begin{tabular}{|c|c|c|c|c|}
\hline
$n_x^2\times n_t$ & $\left \langle e_z, \tilde{A}_{\lambda_z} e_z \right \rangle_{L^2}$ & $\langle \lambda_z \nabla f(\phi(T)),\tilde\gamma^{(1)}\rangle$ & $C(z)$ & $ N_{\text{eval}}$\\[.2cm]
  \hline
  $64^2\times512$   & $1.79462$ & $-1.00453$ & $ 0.490017$ & 1763\\ 
  $128^2\times1024$ & $1.79129$ & $-0.969077$ & $0.507641$ & 1808\\ 
  $256^2\times2048$ & $1.72903$ & $-1.01467$ & $ 0.495233$ & 1919\\\hline 
\end{tabular}
\end{table}

All taken together, the asymptotic estimate of the tail probability $P^\eps(z)$ of the pollutant concentration at location $x_{\text{meas}}\in\Omega$ is given by
$
P^\eps(z) \overset{\eps \downarrow 0}{\sim} \eps^{1/2} (2 \pi)^{-1/2} C(z) \exp \left\{-\eps^{-1} I(z) \right\}
$
with prefactor
\begin{align}
  \label{eq:advdiff-prefactor-full}
  \begin{aligned}
C(z)& = \left[2 I(z) {\det}_2 \left(\Id - \text{pr}_{\eta_z^\perp} A_{\lambda_z} \text{pr}_{\eta_z^\perp} \right) \right]^{-1/2} \times \\ &\qquad \qquad \times  \exp \left\{ - \tfrac{1}{2} \left \langle e_z, A_{\lambda_z} e_z \right \rangle + \lambda_z \left(\varphi_\ell * \tilde \gamma^{(1)}(\cdot,T)\right)(x_{\text{meas}})\right\}\,.
\end{aligned}
\end{align}

The results of this computation are depicted in figure~\ref{fig:advdiff-cdf-and-spectrum}. This figure shows both the value of the rate function $I(z)$ and the full prefactor $C(z)$, computed via~(\ref{eq:advdiff-prefactor-full}), as a function of the observable $z$. We remark that in contrast to earlier examples, the prefactor varies significantly with $z$ over the whole considered interval. As a consequence, any LDT estimate without prefactor must be wildly inaccurate. This is indeed observed in the center panel of figure~\ref{fig:advdiff-cdf-and-spectrum}, where we show the sharp LDT estimate of the tail probability $P^\eps(z)$, i.e.~the probability that the concentration around $x_{\text{meas}}$ exceeds $z$. Solid and dashed lines correspond to the LDT estimate with and without prefactor correction, respectively. The uncorrected LDT estimate is a best effort fit, but clearly nowhere tangent to the actually observed probabilities. Compared to this, as circular and square markers, are the results of direct numerical simulations, simulating the SPDE~(\ref{eq:advdiff-spde-ito}) for $\eps\in\{0.05,0.1\}$ with approximately $5\cdot10^5$ and $10^6$ samples, respectively. We observe that the LDT estimate with prefactor explains the tail probability very well over the whole considered interval. Lastly, the right panel of figure~\ref{fig:advdiff-cdf-and-spectrum} shows the spectrum of the projected Hessian $\text{pr}_{\eta_z^\perp} A_{\lambda_z} \text{pr}_{\eta_z^\perp}$. The spectrum is shown for the first $M=200$ eigenvalues only, and for three different sets of numerical resolutions, with $n_x^2\times n_t \in \{64^2\times512, 128^2\times 1024, 256^2\times2048\}$, with no obvious differences between the three. This demonstrates that the lowest resolution is likely sufficient to obtain converged results. Nevertheless, all other computations in this section are performed with $n_x^2\times n_t = 128^2\times 1024$. We further observe that only the first few eigenvalues significantly differ from the limiting scaling that is empirically proportional to $\mu_i\propto i^{-1}$. The inset of the rightmost panel of figure~\ref{fig:advdiff-cdf-and-spectrum} depicts the cumulative product of eigenvalues necessary to compute the CF-determinant, for $0<m\le M=200$, to demonstrate that 200 eigenvalues are sufficient to compute the CF-determinant accurately.\\

Table~\ref{tab:advdiff-z-0.15} summarizes all terms that enter the prefactor computation, including the additional Stratonovich correction, for the observable value $z=0.15$, showing good agreement of the results for different numerical resolutions. Note in particular that the number of equation evaluations, $N_{\text{eval}}$, is roughly identical for all resolutions. Here, we count cost function evaluations as a single evaluation, gradient evaluations as two (forward and adjoint), and an operator-vector multiply for the eigenvalue computation as four (forward and adjoint, second order forward and adjoint). The number of solves for the optimization varies due to our termination condition, but remains roughly constant. The number of solves for the determinant computation is constant for all resolutions, since we always compute $M=200$ eigenvalues (which needs 401 operator computations, i.e. 1604 equation evaluations). 

\section{Conclusion}
\label{sec:concl}

In the present paper, we have shown how the classical SORM estimate for the approximation of extreme event probabilities in the spirit of Laplace's method needs to be modified when it is applied to infinite-dimensional path space for general diffusion processes~\eqref{eq:SDE} with multiplicative noise. This generalizes earlier work~\cite{schorlepp-tong-grafke-etal:2023}, and builds on the theory developed in~\cite{ben-arous:1988}. The central component of the estimate is the computation of CF determinants and regularized Hessian traces, which we have implemented in a matrix-free way to ensure numerical scalability to high dimensions. We have demonstrated the resulting simple method in two examples, a predator-prey model as an example of a low-dimensional ordinary SDE, and a random advection-diffusion equation as a high-dimensional SPDE example, where we exploited the low-rank property of the diffusion operator $\sigma$. The presented method should be useful for extreme event estimates in various other examples involving S(P)DEs. The source code for the predator-prey model is available under~\cite{Schorlepp-github}.\\

The novelty of the presented approach and algorithm is that it allows for the estimation of extremely rare probabilities in complex high-dimensional systems without any fitting parameters. Since the whole algorithm is developed in a scalable way, it is feasible to apply it to very high-dimensional systems, including SPDEs with multiple spatial dimensions. In particular, we demonstrated how both the computation of the minimizer, as well as the computation of the prefactor, including approximating the Carleman--Fredholm determinant of the second variation operator, can be implemented in a way such that the number of equation solves remains approximately constant independent of the spatial or temporal resolution. As such, computing tail probability estimates necessitates only hundreds of equation solves for the optimization (which we implement via L-BFGS, needing equation solves for gradient computation and line search) as well the computation of operator trace and determinant (where the number of operator evaluations empirically is roughly 2 or 3 times the number of computed eigenvalues). Consequently, the approach is already competitive for only mildly rare events with probabilities of $\approx10^{-3}$ (for which $\approx 1000$ stochastic samples would be needed to observe the event once), but becomes both more accurate and wildly more efficient for lower probabilities, making it possible to obtain highly accurate probability estimates for extremely rare events. In total, and including harnessing the sparsity of the forcing covariance of our largest problem, we can approximate on a single core, and in approximately an hour, the determinant of the second variation operator, which in the largest case, when naively discretized, corresponds to a matrix with $(2n_x^2 n_t)^2 \approx 7\cdot 10^{16}$ entries (for $n_x=256, n_t=2048$), which we of course never actually assemble.\\

Beyond further applications of the theory used here, such as a more in-depth study of realistic advection-diffusion problems, it would be interesting to generalize the presented method to rough differential equations according to the theory developed in~\cite{inahama:2013,yang-xu-pei:2025}. A natural next step are then singular SPDEs, where precise Laplace asymptotics have recently been considered~\cite{friz-klose:2022,klose:2022}. Being able to evaluate and analyze the theoretical results obtained there, and e.g.\ compare them against Monte Carlo simulations or other renormalization methods from physics, would be an interesting future direction.

\section{Acknowledgments}

The authors would like to thank G{\'e}rard Ben Arous, Nils Berglund, Tom Klose, Youssef Marzouk, and Georg Stadler for helpful discussions on different parts of this work. For the purpose of open access, the authors have applied a Creative Commons Attribution (CC BY) license to any Author Accepted Manuscript version arising from this submission. 


\appendix
\renewcommand{\thesection}{\Alph{section}}
\renewcommand{\thesubsection}{\Alph{section}.\arabic{subsection}}
\renewcommand{\thesubsubsection}{\Alph{section}.\arabic{subsection}.\arabic{subsubsection}}

\section{Alternative approach with Riccati differential equations}
\label{app:ricc}

We could alternatively carry out the prefactor computation with matrix Riccati differential equations~\cite{maier-stein:1996,schorlepp-grafke-grauer:2021,bouchet-reygner:2022,schorlepp-grafke-grauer:2023,grafke-schaefer-vanden-eijnden:2023}, where we have, in the case of additive noise,
\begin{align*}
C(z) &=
\frac{1}{\lambda_z} \exp\left\{\frac{1}{2} \int_0^T
\trace \left[\left \langle \nabla^2 b(\phi_z), \theta_z \right
\rangle  Q_z \right] \dd t\right\} \times \\ &\qquad \qquad \times\left[{\det} \left( U_z \right) \left \langle \nabla f(\phi_z(T)),
Q_z(T) U_z^{-1} \nabla f(\phi_z(T)) \right \rangle  \right]^{-1/2}
\end{align*}
for the prefactor $C(z)$ in~\eqref{eq:tail-prob-prefac-sde-additive} with
\begin{align*}
U_z := \text{Id}_n - \lambda_z \nabla^2 f
\left(\phi_z(T) \right) Q_z(T) \in \RR^{n \times n}\,,
\end{align*}
and a symmetric matrix-valued function $Q_z \colon [0,T] \to \RR^{n \times n}$ solving the Riccati differential equation
\begin{align*}
\begin{cases}
\dot{Q}_z = \sigma \sigma^\top + Q_z \nabla b\left(\phi_z
\right)^\top +
 \nabla b\left(\phi_z \right) Q_z + Q_z \left \langle \nabla^2
b(\phi_z), \theta_z\right \rangle Q_z \,,\\
Q_z(0) = 0 \in \RR^{n \times n}\,.
\end{cases}
\end{align*}
However, this can be numerically difficult to solve for high-dimensional state spaces with $n \gg 1$, can display pseudo-singularities unless specialized time steppers are used, and requires more effort to implement, see~\cite{schorlepp-tong-grafke-etal:2023} for a detailed comparison. The Riccati approach does generalize immediately to multiplicative noise without structural changes (whereas generalizing the expression~\eqref{eq:tail-prob-prefac-sde-additive} requires more effort, as we see in the present paper), and in fact generalizes to any continuous-time Markov process in $\RR^n$ that satisfies an LDP as $\eps \downarrow 0$, as discussed in~\cite{schorlepp-grafke-grauer:2023,grafke-schaefer-vanden-eijnden:2023}. It is less straightforward, however, to generalize the Riccati approach to non-Markovian processes (see e.g.~\cite{rosinberg-tarjus-munakata:2024}), which is possible without major changes for the operator determinant approach~\cite{inahama:2013}.


\section{Component form of the first and second order adjoint equations}
\label{app:index}

For completeness and to explicitly define the notation used in~\eqref{eq:first-order-adjoint-mult-noise}-\eqref{eq:second-order-adjoint-mult-noise}, we list the first and second order adjoint equations for the multiplicative-noise noise-to-observable map~\eqref{eq:noise-to-obs-mult} in index notation (with Einstein summation convention) here. The first variation is
\begin{align*}
\left(\fdv{\left(\lambda F \right)}{\eta}\right)_i  = \sigma_{ji}(\phi) \theta_j
\end{align*}
with
\begin{align*}
\begin{cases}
\dot{\phi}_i = b_i(\phi) + \sigma_{ij}(\phi) \eta_j\,, \quad
&\phi_i(0) = x_i\,,\\
\dot{\theta}_i = -\partial_i b_j(\phi) \theta_j - \partial_i \sigma_{jk}(\phi)  \eta_k \theta_j\,,
\quad &\theta_i(T) = \lambda \partial_i f(\phi(T))\,,
\end{cases}
\end{align*}
and the second variation acts as
\begin{align*}
(A_\lambda \delta \eta)_i =  \sigma_{ji}(\phi) \zeta_j +  \theta_k \partial_j \sigma_{ki}(\phi) \gamma^{(1)}_j
 \,,
\end{align*}
with
\begin{align*}
\begin{cases}
\dot{\gamma}^{(1)}_i = \partial_j b_i(\phi) \gamma^{(1)}_j + \partial_j \sigma_{ik}(\phi)
\eta_k \gamma^{(1)}_j  +  \sigma_{ij}(\phi) \delta \eta_j\,, \quad & \gamma^{(1)}_i(0) = 0\,,\\
\dot{\zeta}_i \quad = - \partial_i
b_j(\phi) \zeta_j  - \partial_i \sigma_{jk}(\phi) \eta_k  \zeta_j - 
\partial_i \sigma_{jk} (\phi) \delta \eta_k \theta_j &\\
\qquad \quad   - \partial_i 
\partial_j b_k(\phi)
\theta_k \gamma^{(1)}_j - \partial_i \partial_j \sigma_{kl}(\phi) \eta_l  \gamma^{(1)}_j \theta_k\,, \quad &
\zeta_i(T) = \lambda \partial_i \partial_j f(\phi) \gamma^{(1)}_j(T)\,.
\end{cases}
\end{align*}


\section{Parallels to Gauss--Newton approximation}
\label{app:bip:laplace}
Contrasting the decomposition and properties of the second variation~$A_\lambda$ of the noise-to-observable map~$\lambda F$ with the well-known Gauss--Newton approximation of the Hessian in inverse problems in infinite dimensions provides further insights into the construction entering theorem~\ref{thm:prefac}. Consider a Gaussian prior $\pi_{\text{pr}} = {\cal N}\left( m_{\text{pr}} , C_{\text{pr}} \right)$ on a Hilbert space ${\cal H}$, with $C_{\text{pr}}$ symmetric positive definite and TC, and a possibly nonlinear forward map $F \colon {\cal H} \to {\cal H}_{\text{obs}}$ with additive Gaussian noise $\Delta \sim {\cal N}\left( 0, C_{\text{obs}} \right)$ from the parameter $x$ to observations $y = F(x) + \Delta \in {\cal H}_{\text{obs}}$. Then given an observation $y$, the Laplace approximation to the posterior measure~$\pi_{\text{pos}}^y$ with Radon--Nikodym derivative~\cite{stuart:2010}
\begin{align*}
\dv{\pi_{\text{pos}}^y}{\pi_{\text{pr}}} (x)  \propto \exp \left\{-\frac{1}{2} \norm{F(x) - y}^2_{C_{\text{obs}}^{-1}} \right\}
\end{align*}
at the maximum a posteriori (MAP) point $x_{\text{MAP}}$ with
\begin{align*}
x_{\text{MAP}} = \argmin_x \frac{1}{2} \norm{F(x) - y}^2_{C_{\text{obs}}^{-1}} + \frac{1}{2} \norm{x - m_{\text{pr}}}^2_{C_{\text{pr}}^{-1}}
\end{align*}
reads $x \sim {\cal N}\left(x_{\text{MAP}}, C_{\text{pos}} \right)$ with covariance operator $C_{\text{pos}} = C_{\text{pr}}^{1/2} \left[\text{Id}_{\cal H} + B \right]^{-1 }C_{\text{pr}}^{1/2}$, where
\begin{align}
B := C_{\text{pr}}^{1/2} \left(\nabla F\left(x_{\text{MAP}}\right)^\top C_{\text{obs}}^{-1}  \nabla F\left(x_{\text{MAP}}\right) + \left \langle \nabla^2 F\left(x_{\text{MAP}}\right), C_{\text{obs}}^{-1} \left(F\left(x_{\text{MAP}}\right) - y \right) \right \rangle_{{\cal H}_{\text{obs}}}\right) C_{\text{pr}}^{1/2}  \,.
\label{eq:post-covariance}
\end{align}
The properties and possible low-rank approximations of the update operator $B$~\cite{spantini-solonen-cui-etal:2015} are of practical interest here, e.g.\ to sample from the Laplace approximation of the posterior~\cite{bui-tanh-ghattas-martin-etal:2013}, to use Newton Markov Chain Monte Carlo~\cite{martin-wilcox-burstedde-etal:2012}, or to identify a subspace for dimension reduction~\cite{cui-martin-marzouk-etal:2014}. The commonly used Gauss--Newton approximation consists of dropping the second derivative of $F$ and using only
\begin{align*}
B - \tilde{B} = C_{\text{pr}}^{1/2} \nabla F\left(x_{\text{MAP}}\right)^\top C_{\text{obs}}^{-1}  \nabla F\left(x_{\text{MAP} }\right) C_{\text{pr}}^{1/2}\,,
\end{align*}
which should be a good approximation if the mismatch $F\left(x_{\text{MAP}}\right) - y$ is small, and in particular will always define a symmetric positive semidefinite operator, even if used at any $x$ away from the MAP point. By writing the noise-to-observable map~\eqref{eq:noise-to-obs-mult} as $\lambda F = \lambda f \circ \Phi_T$ with $\Phi_T : L^2([0,T], \RR^n) \to \RR^n$ the solution map of the differential equation in~\eqref{eq:noise-to-obs-mult} until time $t = T$, we have
\begin{align*}
A_\lambda = \left( \nabla \Phi_T(\eta_\lambda) \right)^\top \lambda \nabla^2 f(\phi_\lambda(T)) \nabla \Phi_T(\eta_\lambda) + \left \langle  \nabla^2 \Phi_T(\eta_\lambda) ,\lambda \nabla f(\phi_\lambda(T)) \right \rangle\,.
\end{align*}
As we see in appendix~\ref{app:theory}, the operator $\tilde{A}_\lambda$ for multiplicative noise that is responsible for the generically lower regularity of $A_\lambda$ corresponds to (a part of) the second term here, and the TC operator $A_\lambda - \tilde{A}_\lambda$ is hence akin to a Gauss--Newton approximation. The key difference of our setup to inverse problems is that in the latter, it is often (except for questions of optimal experimental design~\cite{alexanderian:2021}) only important that $B$ is compact, meaning that its eigenvalues converge to zero in any way. As an example, we refer to~\cite{bui-thanh-ghattas:2012a,bui-thanh-ghattas:2012b} for an analysis of Hessians in inverse scattering problems from this perspective. Note also the benign influence in the definition of $B$ in~\eqref{eq:post-covariance} of the TC prior covariance preconditioning $C_{\text{pr}}^{1/2} \dots C_{\text{pr}}^{1/2}$, which is absent in the present extreme event setting, where we further require more stringent HS and TC properties of $A_\lambda$ and its constituents beyond compactness.


\section{Theoretical background and derivation of the main result}
\label{app:theory}

In this appendix, we provide both a heuristic motivation, and then a more detailed mathematical derivation, following the seminal work~\cite{ben-arous:1988}, of the general expression~\eqref{eq:tail-prefac-projected} for the leading prefactor in precise Laplace asymptotics~\eqref{eq:tail-asymp} of small multiplicative noise SDEs~\eqref{eq:SDE}. We first focus on the computation of sharp asymptotics of the MGF $J^\eps \colon \RR \to [0, \infty]$ here, instead of tail probabilities~\eqref{eq:tail-asymp}, with
\begin{align}
J^\eps(\lambda) = \EE \left[\exp \left\{\frac{\lambda}{\eps} f\left(X_T^\eps\right) \right\} \right] \overset{\eps \downarrow 0}{\sim} R_\lambda \exp \left\{\frac{1}{\eps} I^*(\lambda) \right\}\,,
\label{eq:mgf-def-asymp}
\end{align}
where $I^*$ denotes the LF transform of $I$. In particular, we will revisit and illustrate key steps in the proof of the following result from~\cite{ben-arous:1988} in this appendix:

\begin{theorem}
\label{thm:mgf-prefac}
The leading-order MGF prefactor~$R_\lambda$ in the asymptotic expansion~\eqref{eq:mgf-def-asymp} is given by
\begin{align}
R_\lambda = {\det}_2\left(\Id - A_\lambda \right)^{-1/2} \exp\left\{\frac12 \trace\left[A_\lambda-\tilde{A}_\lambda\right]\right\}\,,
\label{eq:mgf-prefac}
\end{align}
for the multiplicative noise SDE~\eqref{eq:SDE}, generalizing the corresponding expression for additive noise SDEs~\eqref{eq:sde-additive}
\begin{align*}
R_\lambda = \det \left(\Id - A_\lambda \right)^{-1/2}\,.
\end{align*}
\end{theorem}

We address the transformation onto tail probabilities~\eqref{eq:tail-asymp} to deduce the tail probability prefactor~\eqref{eq:tail-prefac-projected} from~\eqref{eq:mgf-prefac} afterwards in subsection~\ref{sec:mgf-to-prob}. We assume in this entire appendix that the observable rate function $I \colon \RR \to \RR$ is strictly convex and differentiable, and hence that $J^\eps(\lambda) < \infty$ for all $\lambda \in \RR$ as $\eps \downarrow 0$. Importantly, there is a bijective relation between $z$ and a $\lambda = \lambda_z$ in this case, such that the small-noise asymptotics of the tail probability $P^\eps(z)$ at any given $z \in \RR$ can be recovered from calculating the $J^\eps(\lambda)$ asymptotics at $\lambda = \lambda_z$. Similarly, as a stronger assumption compared to~\eqref{eq:sufficent-optimal}, we assume that $\text{Id} - A_\lambda$ is positive definite on all of $L^2([0,T],\RR^n)$ in the following. We will comment on these points in subsections~\ref{sec:mgf-to-prob} and~\ref{sec:convex-rf}.\\

Starting with a heuristic motivation of~\eqref{eq:mgf-def-asymp} and~\eqref{eq:mgf-prefac} in subsection~\ref{sec:heur} in the Stratonovich picture, we will see that~\eqref{eq:mgf-prefac} can formally be obtained based on introducing an infinite ``counter-term'' in a functional integral representation of the MGF, which regularizes the Gaussian functional integral in case the second variation operator $A_\lambda$ is only HS instead of TC, as it was for additive noise in~\cite{schorlepp-tong-grafke-etal:2023}. After this intuitive derivation, we will explain certain aspects of the underlying mathematical theory of~\cite{ben-arous:1988}, where theorem~\ref{thm:mgf-prefac} was first stated and proven, to provide a clearer picture why the result~\eqref{eq:mgf-prefac} -- and thus ultimately \eqref{eq:tail-prefac-projected} -- is actually quite natural. This is based on well-known facts from functional analysis and probability theory about (regularized) determinants of $p$-Schatten class operators (see subsection~\ref{sec:operator-dets} and~\cite{simon:1977,dunford-schwartz:1988,simon:2005,mckean:2011}), Wiener--It{\^o} chaos decompositions of random variables (see subsection~\ref{sec:wiener-chaos} and~\cite{janson:1997,nualart:2006}), and the exponential integrability of random variables in the second inhomogeneous Wiener--It{\^o} chaos (see subsection~\ref{sec:exp-int-cf-det} and \cite{janson:1997,grasselli-hurd:2005,glimm-jaffe:2012}). We explicitly introduce and discuss all of these ingredients here for a self-contained presentation.\\

After these prerequisites have been clarified, and still following~\cite{ben-arous:1988}, we first compute the leading-order term in the Laplace asymptotics of the MGF~\eqref{eq:mgf-def-asymp} in subsection~\ref{sec:leading-MGF-contr}, which is a standard result in Freidlin--Wentzell theory in the sense of Varadhan's lemma, and then localize and perform a ``stochastic'' Taylor expansion around the large deviation minimizer in subsection~\ref{sec:loc-taylor-MGF}. At that point, given the prerequisites we recalled in subsections~\ref{sec:operator-dets} to~\ref{sec:exp-int-cf-det}, it will become evident that the main step to obtain~\eqref{eq:tail-prefac-projected} is finding the Wiener--It{\^o} chaos decomposition of the second variation of the noise-to-event map $\lambda F = \lambda f \circ \Phi_T$, evaluated formally at white noise $\xi$, which we denote by $\text{``}\hat{Q}_\lambda = \dd^2 \left(\lambda f \circ \Phi_T\right)(\xi, \xi)\text{''}$. For this, one needs to show that the classical second variation operator $A_\lambda$ is indeed HS -- as we had assumed in the heuristic derivation of~\eqref{eq:mgf-def-asymp} and~\eqref{eq:mgf-prefac} in subsection~\ref{sec:heur} already -- and that the expectation $\EE \left[\hat{Q}_\lambda \right]$ is finite (and in fact given by the trace of the regularized operator $A_\lambda - \tilde{A}_\lambda$). This is detailed in subsection~\ref{sec:q-hat-wiener-chaos}, and, among other things, rests upon Goodman's theorem for abstract Wiener space~\cite{gross:1967,kuo:2006}, which we recall in subsection~\ref{sec:wiener-space-goodman} for this purpose. Putting everything together, we obtain the result~\eqref{eq:mgf-prefac} for the prefactor.

\subsection{Heuristic derivation and intuition via path integrals}
\label{sec:heur}

We motivate equations~\eqref{eq:mgf-def-asymp} and~\eqref{eq:mgf-prefac} for the asymptotic expansion of the MGF including the leading-order prefactor here through formal manipulations of functional integrals. While these calculations are purely formal, they do, however, illustrate the main idea. Following common wisdom, we first convert the It{\^o} SDE~\eqref{eq:SDE} to an equivalent Stratonovich SDE, similar to remark~\ref{rem:ito-strato}, in order to be able to use the standard rules of calculus for the Taylor expansion in the functional integral below. On the level of the noise-to-observable map $F = f \circ \Phi_T$, this means that we replace the solution map $\Phi_T \colon L^2([0,T], \RR^n) \to \RR$ by an $\eps$-dependent map $\hat{\Phi}_T^\eps$ with $\hat{\Phi}^\eps[\eta] = \hat{\phi}^\eps$, which solves
\begin{align*}
\begin{cases}
\dv{}{t} \hat{\phi}^\eps = b \left(\hat{\phi}^\eps \right) - \frac{\eps}{2} \trace \left[ \sigma^\top  \left(\hat{\phi}^\eps \right) \nabla \sigma  \left(\hat{\phi}^\eps \right)\right] + \sigma  \left(\hat{\phi}^\eps \right) \eta\,,\\
\hat{\phi}^\eps(0) = x\,.
\end{cases}
\end{align*}
Similarly, we define the corresponding noise-to-observable map as
$\hat{F}^\eps \colon L^2([0,T],\RR^n) \to \RR$, $\hat{F}^\eps[\eta]\allowbreak =  f\left(\hat{\Phi}^\eps_T[\eta]\right)$.
With this, we formally write the MGF~\eqref{eq:mgf-def-asymp} as an integral
\begin{align}
J^\eps(\lambda) = \EE \left[ \exp \left\{\frac{\lambda}{\eps} f(X^\eps_T) \right\} \right] = \int D \eta \;  \exp \left\{-\frac{1}{\eps} \left(\frac{1}{2}\norm{\eta}_{L^2}^2 - \lambda \hat{F}^\eps[\eta] \right) \right\}
\label{eq:mgf-funcint}
\end{align}
against the non-existent infinite-dimensional Lebesgue measure on path space. Still, we immediately see that the integrand has a form that is amenable to approximating it with Laplace's method as $\eps \downarrow 0$. The large deviation minimizer or instanton (noise) for the MGF at $\lambda$ is hence given by considering the dual objective to~\eqref{eq:min-prob-eta},
\begin{align*}
\eta_\lambda = \argmin_{\eta \in L^2([0,T], \RR^n)} \frac{1}{2}\norm{\eta}_{L^2}^2 - \lambda F[\eta]\,,
\end{align*}
where we used $F[\eta] = \hat{F}^{\eps = 0}[\eta]$. The leading exponential term will then give the LF transform of $I$ at $\lambda$, since
\begin{align}
\begin{aligned}
I^*(\lambda) &= \max_{z \in \RR} \left( \lambda z - I(z) \right) =\max_{z \in \RR} \left( \lambda z - \min_{\eta : F[\eta] = z} \frac{1}{2} \norm{\eta}^2_{L^2}\right) \\&= -\min_{z \in \RR} \min_{\eta : F[\eta] = z} \left( \frac{1}{2} \norm{\eta}^2_{L^2} - \lambda F[\eta]\right) = - \left( \frac{1}{2} \norm{\eta_\lambda}^2_{L^2} - \lambda F[\eta_\lambda] \right)\,.
\end{aligned}
\label{eq:legendre}
\end{align}
Now, expanding $\eta = \eta_\lambda + \sqrt{\eps} \delta \eta$ around the minimizer to second order in the fluctuations $\delta \eta$ in the functional integral representation~\eqref{eq:mgf-funcint} of the MGF yields
\begin{align*}
J^\eps(\lambda) &\overset{\eps \downarrow 0}{\sim} \exp\{+\eps^{-1} I^*(\lambda)\} \int D(\delta \eta) \exp \bigg\{-\frac{1}{2} \bigg(\left \langle \delta \eta, \left[\text{Id} - A_\lambda \right] \delta \eta \right \rangle_{L^2}\\ &\qquad \qquad \qquad - 2 \big \langle \left. \tfrac{\partial}{\partial \left(\sqrt{\eps} \right)} \right|_{\eps = 0} \left(\tfrac{\delta \left( \lambda \hat{F}^\eps \right)}{\delta \eta} \right) [\eta_\lambda], \delta \eta \big \rangle_{L^2}  - \left( \tfrac{\partial^2}{\partial \left(\sqrt{\eps} \right)^2} \left( \lambda \hat{F}^\eps\right)\right)_{\eps = 0}[\eta_\lambda]\bigg) \bigg\}\,,
\end{align*}
where $A_\lambda$ is the second-variation operator defined in~\eqref{eq:second-var-mult} and~\eqref{eq:second-order-adjoint-mult-noise}, without the It{\^o}--Stratonovich correction term since we expand around $\eps = 0$. The last two terms above can be evaluated with the adjoint state method~\cite{plessix:2006} again. This shows that
\begin{align*}
\left. \frac{\partial}{\partial \left(\sqrt{\eps} \right)} \right|_{\eps = 0} \left(\frac{\delta \left( \lambda \hat{F}^\eps \right)}{\delta \eta} \right) [\eta_\lambda] = 0\,,
\end{align*}
which is clear even without calculation, since only a term quadratic in $\sqrt{\eps}$ appeared in the original functional, as well as
\begin{align*}
\left( \frac{\partial^2}{\partial \left(\sqrt{\eps} \right)^2} \left( \lambda \hat{F}^\eps \right) \right)_{\eps = 0}[\eta_\lambda] = -\int_0^T \trace\left[\sigma^\top(\phi_\lambda(t))\nabla \sigma(\phi_\lambda(t)) \theta_\lambda(t) \right] \dd t\,.
\end{align*}
With this, we end up with
$
J^\eps(\lambda) \overset{\eps \downarrow 0}{\sim} R_\lambda \exp \left\{\eps^{-1} I^*(\lambda) \right\}\,,
$
where the leading prefactor is
\begin{align}
R_\lambda = \exp \left\{-\frac{1}{2} \int_0^T \trace\left[\sigma^\top(\phi_\lambda)\nabla \sigma(\phi_\lambda) \theta_\lambda \right] \dd t \right\} \int D(\delta \eta) \exp \left\{-\frac{1}{2} \left \langle \delta \eta, \left[\text{Id} - A_\lambda \right] \delta \eta \right \rangle_{L^2} \right\}\,.
\label{eq:rlambda-funcint}
\end{align}
If we know that the second variation operator $A_\lambda$ is a TC operator, we can simply evaluate
\begin{align*}
\int D(\delta \eta) \exp \left\{-\frac{1}{2} \left \langle \delta \eta, \left[\text{Id} - A_\lambda \right] \delta \eta \right \rangle_{L^2} \right\} = {\det} \left(\text{Id} - A_\lambda \right)^{-1/2}
\end{align*}
in terms of the Fredholm determinant $\det$ and conclude, cf.~\eqref{eq:tail-prefac-projected-mult-tc}. However, if the operator $A_\lambda$ is not TC, such that $\trace \left[ A_\lambda \right]$ is not well-defined, but we do already know that it is still HS (e.g.\ from the intuitive considerations related to figure~\ref{fig:flowcharts}), such that $\trace \left[A_\lambda^p\right] < \infty$ for $p \geq 2$, we may formally add and subtract the counter term $\pm \tfrac{1}{2} \trace \left[ A_\lambda \right]$ in the exponent of~\eqref{eq:rlambda-funcint} to make the determinant well defined. This is because, as detailed in the following subsection, we have
\begin{align}
 {\det} \left(\text{Id} - A_\lambda \right) = \exp \left\{ \trace \left[ \log \left(\text{Id} - A_\lambda\right) \right]\right\} =  \exp \left\{- \sum_{p = 1}^\infty \trace \left[ A_\lambda^p \right]\right\}\,,
 \label{eq:exptrlog-motivation}
\end{align}
and hence the counter-term removes the ill-defined trace of $A_\lambda$ itself, leaving only well-defined traces of higher powers of~$A_\lambda$. Note that this corresponds precisely to \textit{Wick ordering} the exponent~\cite{glimm-jaffe:2012}, and renders
\begin{align*}
&\int D(\delta \eta) \exp \left\{-\frac{1}{2} \left \langle \delta \eta, \left[\text{Id} - A_\lambda \right] \delta \eta \right \rangle - \frac{1}{2} \trace \left[ A_\lambda \right] \right\} \\ &\quad \text{\;``$=$''\;} \left[\det \left(\left(\text{Id} - A_\lambda \right) \exp \left\{A_\lambda \right\} \right) \right]^{-1/2} = {\det}_2 \left(\text{Id} - A_\lambda \right)^{-1/2}
\end{align*}
in~\eqref{eq:rlambda-funcint} finite. Importantly, it turns out that at the same time, the remaining trace term of opposite sign that we introduced in~\eqref{eq:rlambda-funcint} to make the determinant well-defined also results, morally speaking, in
\begin{align}
\exp \left\{\frac{1}{2}\left(\trace \left[ A_\lambda \right]  -  \int_0^T \trace\left[\sigma^\top(\phi_\lambda)\nabla \sigma(\phi_\lambda) \theta_\lambda \right] \dd t \right) \right\}  \text{\;``$=$''\;} \exp \left\{\frac{1}{2} \trace \left[A_\lambda - \tilde{A}_\lambda \right] \right\} < \infty\,,
\label{eq:exp-traatilde-cancellation}
\end{align}
if the renormalization procedure is carried out correctly, as we will see in the following subsections. It yields the trace of a TC operator $A_\lambda - \tilde{A}_\lambda$, exactly removing the part of $A_\lambda$ that leads to the second variation being only HS. Let us briefly show that the identity~\eqref{eq:exp-traatilde-cancellation} above holds at least for TC $\tilde{A}_\lambda$, in order to complete our heuristic reasoning: Introducing the propagator or resolvent ${\cal P}(t,t') = {\cal T} \exp \left\{ \int_{t'}^t L[\eta(t''), \phi(t'')] \dd t''\right\}$ where ${\cal T}$ denotes time-ordering, we can write
\begin{align*}
\gamma^{(1)}(t) &= \int_0^t {\cal P}(t,t') \sigma(\phi(t')) \delta \eta(t') \dd t'\\
\zeta_{\text{sing}}(t) &= \int_t^T {\cal P}^\top(t',t) \left( \nabla \sigma(\phi(t')) \delta \eta(t') \right)^\top \theta(t') \dd t'
\end{align*}
for the evaluation of~$\tilde{A}_\lambda \delta \eta$ according to~\eqref{eq:atilde-return} and~\eqref{eq:a-tilde-adjoint-mult-noise}. Rewriting $\tilde{A}_\lambda \delta \eta(t) = \int_0^T \tilde{\cal A}_\lambda(t,t') \delta \eta(t') \dd t'$ with a symmetric integral kernel $\tilde{\cal A}_\lambda \colon [0,T]^2 \to \RR^{n \times n}$ and assuming that $\tilde{A}_\lambda$ is TC, such that~\cite{brislawn:1988}
\begin{align*}
\trace \left[\tilde{A}_\lambda \right] = \int_0^T \trace \left[\tilde{\cal A}_\lambda \right](t,t) \; \dd t\,,
\end{align*}
enables us to derive~\eqref{eq:trace-atilde-integral} from the explicit form of the integral kernel
\begin{align}
\tilde{{\cal A}}_{\lambda, ij}(t,t') = \begin{cases}
\theta_{\lambda, k}(t) \partial_r \sigma_{ki}(\phi_\lambda(t)) {\cal P}_{rl}(t,t') \sigma_{lj}(\phi_\lambda(t'))\,, \quad & t' < t\,,\\
\theta_{\lambda, k}(t') \partial_r \sigma_{kj}(\phi_\lambda(t')) {\cal P}_{rl}(t',t) \sigma_{li}(\phi_\lambda(t))\,, \quad & t' > t\,.
\end{cases}
\label{eq:a-tilde-kernel}
\end{align}
On a high level, it now only remains to understand exactly (i) why $A_\lambda$ is indeed HS, as then the counter-term is natural to introduce given what we have already learned about CF determinants, and (ii) why the remainder $A_\lambda - \tilde{A}_\lambda$ happens to be TC. Based on the heuristic derivation of this subsection, the end result~\eqref{eq:mgf-prefac} for the MGF prefactor may appear hard to anticipate, and the cancellation which is happening to give both a finite (CF) determinant and trace at the same time may seem almost coincidental. However, treating the same problem in a proper mathematical framework in the following subsection will hopefully make~\eqref{eq:mgf-prefac} -- and~\eqref{eq:tail-prefac-projected} for tail probabilities -- more natural.

\subsection{Prerequisites}

In this subsection, with the aim of providing a self-contained exposition of theorem~\ref{thm:prefac}, we introduce a lot of the (in principle well-known) definitions and underlying theorems that we will need in the subsequent subsections to summarize the arguments of~\cite{ben-arous:1988}.

\subsubsection{Trace-class and Hilbert--Schmidt operators and their determinants}
\label{sec:operator-dets}

Let ${\cal H}$ be a real separable Hilbert space, and $B \colon {\cal H} \to {\cal H}$ a bounded linear and symmetric operator, such that $\left \langle f, B g \right \rangle_{{\cal H}} = \left \langle B f,  g \right \rangle_{{\cal H}}$ for all $f,g \in {\cal H}$. In the concrete setting of this paper, this could correspond to ${\cal H} = L^2([0,T], \RR^n)$, or, in the underlying classical Wiener space setting to be introduced in subsection~\ref{sec:wiener-space-goodman}, the space ${\cal H} = H_0^{1,2}([0,T], \RR^n)$ of absolutely continuous functions with square integrable weak derivative and initial value $0$. In particular, we are exclusively concerned with infinite-dimensional Hilbert spaces in this paper; the following definitions are trivial for $\text{dim}\, {\cal H} < \infty$.

\begin{definition}
We call $B$ trace-class (TC) iff
\begin{align*}
\trace \left[ B \right] := \sum_{i \in I} \left \langle e_i, B e_i \right \rangle_{{\cal H}}
\end{align*}
converges absolutely and independently of the choice of the (at most countable) orthonormal basis $\left(e_i \right)_{i \in I}$ of $\cal H$.
\end{definition}

Denoting the eigenvalues of $B$, potentially repeated according to their multiplicity, by $\mu_i \in \RR$ with corresponding orthonormal eigenvectors $\hat{e}_i \in {\cal H}$ for $i \in I$, we have $\trace \left[ B \right] = \sum_{i \in I} \mu_i \in (-\infty, \infty)$, where the series converges absolutely. We will \textit{always} work under the following assumption in this subsection:

\begin{assumption}
\label{assump-posdef}
The bounded linear and symmetric operator $B\colon {\cal H} \to {\cal H}$ is such that $\text{Id}_{{\cal H}} - B$ is positive definite, or, equivalently, its eigenvalues satisfy $\mu_i \in (-\infty, 1)$ for all $i \in I$.
\end{assumption}

This is because, with regard to theorem~\ref{thm:prefac} and the condition~\eqref{eq:sufficent-optimal}, the (projected and scaled) Hessian $B = B_z$ in the application considered in this paper has to satisfy assumption~\ref{assump-posdef}, which is equivalent to a second-order sufficient optimality condition or ``nondegeneracy'' of the instanton solution. For TC operators, the following definition then makes sense:

\begin{definition}
For TC operators $B$, we define the Fredholm determinant
\begin{align*}
\det \left(\text{Id}_{\cal H} - B \right) := \prod_{i \in I} \left(1 - \mu_i \right) \in (0, \infty)
\end{align*}
in terms of the eigenvalues $\mu_i$ of $B$.
\end{definition}

The convergence of the product in this determinant definition is guaranteed by the sandwich criterion, as we have for $I_0 \subset I$ finite
\begin{align}
\exp \left\{-\sum_{i \in I_0} \frac{\mu_i}{1 - \mu_i} \right\} \leq \prod_{i \in I_0} \left(1 - \mu_i \right) = \exp \left\{\sum_{i \in I_0}  \log \left(1 - \mu_i \right)\right\} \leq \exp \left\{-\sum_{i \in I_0}  \mu_i\right\}\,,
\label{eq:fred-det-converges-sandwich}
\end{align}
and the lower and upper bounds converge since $B$ is TC. Here, we only used the straightforward inequalities $x/(1+x)\leq \log (1+x) \leq x$ for $x \in (-1, \infty)$. TC operators and their Fredholm determinants are what we encountered when considering well-posed additive-noise SDEs in the introduction, following~\cite{schorlepp-tong-grafke-etal:2023}, and studied the second variation or Hessian of the noise-to-observable map $F$ in~\eqref{eq:noise-to-obs-additive}. Furthermore, Fredholm determinants are of general interest in the literature because of their relation to certain probability distributions that arise in random matrix theory and growth processes. In this context, the precise numerical evaluation of Fredholm determinants through quadrature in low-dimensional state spaces is addressed in~\cite{bornemann:2010}. In the high-dimensional state spaces that are considered in~\cite{schorlepp-tong-grafke-etal:2023,burekovic-schaefer-grauer:2024} and others, one typically finds the leading $M$ eigenvalues of $B$ with largest absolute value in a matrix-free way, and approximates
\begin{align*}
\det \left(\text{Id}_{\cal H} - B \right) \approx \prod_{i = 1}^M \left(1 - \mu_i \right)\,.
\end{align*}

Now, in the predator-prey example in the introduction and in the general formulation of theorem~\ref{thm:prefac}, we encounter so-called HS operators:

\begin{definition}
We call $B$ Hilbert--Schmidt (HS) iff
\begin{align*}
\norm{B}^2_{\text{HS}} := \trace \left[B^2 \right] = \sum_{i \in I} \norm{B e_i}^2_{{\cal H}}
\end{align*}
is finite for any orthonormal basis $\left(e_i \right)_{i \in I}$.
\end{definition}

In terms of the eigenvalues $\mu_i$ of $B$, this amounts to demanding that the sequence of eigenvalues is square summable. Obviously, $B$ being TC implies that $B$ is also HS, while the converse is not true. As equation~\eqref{eq:exptrlog-motivation} in the heuristic motivation above suggests, we need another notion of a (regularized) operator determinant to make sense of determinants of HS operators in a nontrivial way:

\begin{definition}
For HS operators $B$, we define the Carleman--Fredholm (CF) determinant
\begin{align*}
{\det}_2 \left(\text{Id}_{\cal H} - B \right) := \det \left(\left(\text{Id}_{\cal H} - B \right) \exp \left\{B \right\} \right) = \prod_{i \in I} \left(1 - \mu_i \right) \exp \left\{\mu_i \right\}
\end{align*}
in terms of the eigenvalues $\mu_i$ of $B$.
\end{definition}

We commented under~\eqref{eq:exptrlog-motivation} that the ``counter-term'' in the exponent removes the ill-defined trace of $B$. More precisely, for $I_0 \subset I$ finite we now have
\begin{align*}
\exp \left\{-\sum_{i \in I_0} \frac{\mu_i^2}{1 - \mu_i} \right\} \leq \prod_{i \in I_0} \left(1 - \mu_i \right) \exp \left\{\mu_i \right\} = \exp \left\{\sum_{i \in I_0}  \log \left(1 - \mu_i \right) + \mu_i\right\} \leq 1
\end{align*} 
using the same straightforward inequalities for the log as in~\eqref{eq:fred-det-converges-sandwich}, such that the CF determinant ${\det}_2 \left(\text{Id}_{\cal H} - B \right)$ indeed converges to a number in $(0,1]$ for HS operators $B$. If $B$ is TC, the CF determinant factorizes, since both terms converge individually, and we have ${\det}_2 \left(\text{Id}_{\cal H} - B \right) = \det \left(\text{Id}_{\cal H} - B \right) \exp \left\{\trace \left[B \right] \right\}$.

\begin{example}
Consider an operator $B$ with eigenvalues $\mu_i = 1 / (2i)^2$ for $i = 1, 2, 3, \dots$. Then $B$ is TC with $\trace \left[B \right] = \pi^2 / 24$ and $\det \left(\text{Id}_{\cal H} - B \right) = 2 / \pi$, and the CF determinant factorizes as ${\det}_2 \left(\text{Id}_{\cal H} - B \right) = \det \left(\text{Id}_{\cal H} - B \right) \exp \left\{\trace \left[B \right] \right\} = 2/\pi \cdot \exp \left\{\pi^2 / 24 \right\}$. Now consider an operator~$B$ with eigenvalues $\mu_i = 1 / 2i$ for $i = 1, 2, 3, \dots$ instead. Then $B$ is only HS, but not TC, and while $\prod_{i=1}^M (1 - \mu_i)$ just converges to $0$ as $M \to \infty$ (and hence does not provide a useful notion of a determinant), the CF determinant is ${\det}_2 \left(\text{Id}_{\cal H} - B \right) = 1 / \sqrt{\pi} \cdot \exp \left\{\gamma / 2 \right\} \in (0,1)$, where $\gamma$ is the Euler--Mascheroni constant.
\end{example}

The numerical computation of CF determinants through quadratures in the context of low-dimensional rough SDEs is addressed in~\cite{friz-gassiat-pigato:2022}. In contrast to this, we will again find the dominant $M$ eigenvalues of $B$ in terms of their absolute value in the following in a matrix-free way, and approximate
\begin{align*}
{\det}_2 \left(\text{Id}_{\cal H} - B \right) \approx \prod_{i =1}^M \left(1 - \mu_i \right) \exp \left\{\mu_i \right\}\,.
\end{align*}

\begin{remark}
One can generalize the definition of ${\det}_2$ in a straightforward way to regularized $p$-determinants ${\det}_p$ for $p$-Schatten class operators with $p = 1,2,3,\dots$. We refer the interested reader to~\cite{simon:2005}, where, furthermore, an example from quantum field theory where ${\det}_3$ appears after renormalization is given. Due to their special role in determining characteristic functions of elements in the second Wiener--It{\^o} chaos, as discussed below, and hence their appearance in theorem~\ref{thm:prefac}, CF determinants ${\det}_2$ are perhaps the most important generalization of Fredholm determinants, certainly for our present application.
\end{remark}

We collect a few standard results about TC and HS operators, which we will need in the following, below. The first one concerns a simple characterization of HS operators on $L^2$-spaces through integral kernels:

\begin{theorem}
\label{thm:hs-kernel}
Let ${\cal H} = L^2\left(\Omega, \RR^n\right)$ for $\Omega \subseteq \RR^d$.
\begin{enumerate}
\item Let ${\cal B} \in L^2\left(\Omega \times \Omega, \RR^{n \times n}\right)$ be symmetric. Then the operator $B \colon {\cal H} \to {\cal H}$, defined via
\begin{align*}
(Bf)(x) = \int_{\Omega} {\cal B}(x,y) f(y) \dd y
\end{align*}
for $f \in {\cal H}$ defines a symmetric HS operator, and $\norm{B}_{\text{HS}} = \norm{{\cal B}}_{L^2\left(\Omega \times \Omega, \RR^{n \times n}\right)}$.
\item Conversely, let $B \colon {\cal H} \to {\cal H}$ be a symmetric HS operator. Then there exists a symmetric integral kernel ${\cal B} \in L^2\left(\Omega \times \Omega, \RR^{n \times n}\right)$ such that $(Bf)(x) = \int_{\Omega} {\cal B}(x,y) f(y) \dd y$ for all $f \in {\cal H}$, and $\norm{B}_{\text{HS}} = \norm{{\cal B}}_{L^2\left(\Omega \times \Omega, \RR^{n \times n}\right)}$.
\end{enumerate}
Hence, the space of symmetric functions in $L^2\left(\Omega \times \Omega, \RR^{n \times n}\right)$, and the space of symmetric HS operators on $L^2\left(\Omega, \RR^n\right)$, are isometrically isomorphic.
\end{theorem}

This so-called Hilbert--Schmidt kernel theorem typically makes it straightforward to prove that a given operator on $L^2\left(\Omega, \RR^n\right)$ is HS if the associated kernel can be computed, whose square integrability is then usually easy to see. For the proof of theorem~\ref{thm:hs-kernel}, we only remark that 1.\ follows by computation and using Fubini's theorem, while for 2.\ one sets ${\cal B} = \sum_{i\in I} \mu_i \hat{e}_i \otimes \hat{e}_i$ in terms of the eigenvalues and orthonormal eigenfunctions of $B$.\\

On the other hand, showing that a given operator $B$ on ${\cal H} = L^2\left(\Omega, \RR^n\right)$ is TC is often harder. If we already now that $B$ is TC and has an associated integral kernel ${\cal B}$, then it is known that the following holds~\cite{brislawn:1988}:

\begin{theorem}
Let ${\cal B}$ be the integral kernel of a TC operator $B \colon L^2\left(\Omega, \RR^n\right) \to L^2\left(\Omega, \RR^n\right)$. Then
\begin{align}
\trace \left[B \right] = \int_{\Omega} \lim_{r \downarrow 0} \ \trace\left( \mathrm{boxav}_r \left(  {\cal B} \right) \right)(x,x) \; \dd x\,,
\label{eq:trace-integral-boxav}
\end{align}
where $\mathrm{boxav}_r \left(  {\cal B} \right) (x,x) $ denotes the box average over a box with side lengths $2r$ centered at $(x,x)$.
\end{theorem}

The problem here is of course that a priori ${\cal B}$ is not even defined pointwise and hence this box averaging procedure will generally be necessary. Conversely, showing that an operator $B$, potentially with known kernel ${\cal B}$, is TC in the first place is more difficult. A classical Fourier series example given by Carleman~\cite{carleman:1917} shows that it is not enough for ${\cal B}$ to be continuous and $\Omega$ compact (which renders the right-hand side of~\eqref{eq:trace-integral-boxav} finite) to conclude that the trace of the operator $B$ is finite. In line with this and for the present application, it is easily shown below in subsection~\ref{sec:q-hat-wiener-chaos} from the HS kernel theorem~\ref{thm:hs-kernel} that~$\tilde{A}_\lambda$ is HS, while proving that $A_\lambda - \tilde{A}_\lambda$ is TC makes use of the special problem structure in the form of Goodman's theorem for abstract Wiener space, cf.\ subsection~\ref{sec:wiener-space-goodman}.

\subsubsection{Wiener--It{\^o} chaos decomposition}
\label{sec:wiener-chaos}

We briefly introduce the notion of the Wiener--It{\^o} chaos decomposition of a real-valued random variable that depends on an isonormal Gaussian process here. The latter will be given by a standard $n$-dimensional Brownian motion $W = \left(W_t \right)_{t \in [0,T]}$ for our application, the former by (the Taylor expansion of) the noise-to-observable map $F$. We can think of this as a useful choice of basis for the space of ${\cal G}$-measurable random variables, where ${\cal G}$ is the sigma algebra generated by $W$. This choice of basis will make the subsequent calculation of the expectation that determines the MGF prefactor~$R_\lambda$ easier. This subsection mainly follows~\cite{nualart:2006}, where the corresponding proofs can be found.

\begin{definition}
Let ${\cal H}$ be a real separable Hilbert space with inner product $\left \langle \cdot, \cdot \right \rangle_{\cal H}$. A stochastic process $W = \{W(h) \mid h \in {\cal H} \}$ indexed by ${\cal H}$ on a probability space $(\Omega, {\cal F}, \PP)$ is called isonormal Gaussian process if the real-valued random variables~$W(h)$ are Gaussian with
\begin{align}
\EE \left[W(h) \right] = 0 \text{ and } \EE \left[W(h) W(g) \right] = \left \langle h, g \right \rangle_{\cal H}
\label{eq:isonorm-def}
\end{align}
for all $g, h \in H$.
\end{definition}

The definition implies that $h \mapsto W(h)$ is linear, and hence defines a linear isometry from ${\cal H}$ onto a subspace $\mathscr{H}_1 \subset L^2(\Omega, {\cal G}, \PP)$. The most important example for us is indeed given by a Brownian motion, for which ${\cal H} = L^2 \left([0,T], \RR^n \right)$. Writing $W_i(t) := W\left(\mathds{1}_{[0,t] \times \{i\}}  \right)$ for an isonormal Gaussian process $W$ indexed by this ${\cal H}$, the random variables $W_i(t)$ will need to be centered Gaussian by definition, and furthermore
\begin{align*}
\EE \left[ W_i(t) W_j(s) \right] = \delta_{ij} \int_0^T \mathds{1}_{[0,t]}(t') \mathds{1}_{[0,s]}(t') \dd t' = \min \{t,s \} \delta_{ij}
\end{align*}
by~\eqref{eq:isonorm-def}. By finding a continuous version of this process, we see that $W$, constructed in this way, defines a standard $n$-dimensional Brownian motion, with $W(h) = \int_0^T h(t) \dd W_t$ for a general $h \in L^2\left([0,T], \RR^n \right)$.\\

Because we deal with Gaussian processes, a basis in terms of the (probabilist's) Hermite polynomials $\text{He}_m$ turns out to be convenient, since they are orthogonal with respect to the Gaussian weights. More precisely, we use the following

\begin{definition}
For $m = 0, 1, 2, \dots$, we define $\mathrm{He}_m \colon \RR \to \RR$ through the generating function
\begin{align*}
\exp \left\{tx - \frac{t^2}{2} \right\} = \sum_{m = 0}^\infty \frac{t^m}{m!} \mathrm{He}_m(x)\,,
\end{align*}
for all $x, t \in \RR$. Explicitly, $\mathrm{He}_0(x) = 1$, $\mathrm{He}_1(x) = x$, $\mathrm{He}_2(x) = x^2 - 1 $, and so forth.
\end{definition}

and have the following key property, which follows immediately from the generating function:

\begin{lemma}
Let $X,Y$ be bivariate centered Gaussian with $\EE\left[ X^2 \right] = \EE \left[Y^2 \right] = 1$. Then
\begin{align*}
\EE \left[\mathrm{He}_m(X) \mathrm{He}_l(Y) \right] = \begin{cases}
0\,, \quad & m \neq l\,,\\
m! \EE \left[XY \right]^m\,, \quad & m = l\,.
\end{cases}
\end{align*}
\end{lemma}

\begin{definition}
For all $m \in \NN$, we define the $m$-th (homogeneous) Wiener chaos $\mathscr{H}_m \subset L^2(\Omega, {\cal G}, \PP)$ as the closed linear subspace generated by $\{\mathrm{He}_m(W(h)) \mid h \in {\cal H}\}$. In particular, $\mathscr{H}_0$ is the set of constant random variables, $\mathscr{H}_1$ is the set of random variables generated by $\{W(h) \mid h \in {\cal H}\}$, and so forth.
\end{definition}

The most important property of these subspaces is the following decomposition:

\begin{theorem}
The Wiener chaos subspaces give an orthogonal decomposition $L^2(\Omega, {\cal G}, \PP) = \bigoplus_{m = 0}^\infty \mathscr{H}_m$.
\end{theorem}

It is also straightforward to give an orthonormal basis of $\mathscr{H}_m$ to make this more explicit. We use $\odot$ to denote the symmetrized tensor product in the following.

\begin{proposition}
\label{prop:wiener-iso}
There is an isometry $I_m \colon {\cal H}^{\odot m} \to \mathscr{H}_m$ that provides an orthonormal basis of the $m$-th Wiener--It{\^o} chaos~$\mathscr{H}_m$ as follows:
Let $\left(e_i \right)_{i \in I}$ be an orthonormal basis of ${\cal H}$, and let $a = \left( a_i \right)_{i \in I}$ with $a_i \in \NN$ for all $i \in I$ be a multi-index with $\sum_{i \in I} a_i = m$, then $I_m$ is defined on basis elements of ${\cal H}^{\odot m}$ as
\begin{align}
I_m\left(\bigodot_{i \in I} e_i^{\otimes a_i}\right) = \prod_{i \in I}  \left( \tfrac{1}{a_i!} \mathrm{He}_{a_i} \left(W(e_i)\right) \right)\,,
\label{eq:wiener-iso-def}
\end{align}
and extended linearly to general symmetric tensors.
\end{proposition}

Importantly, for $W$ a Brownian motion, the maps $I_m$ are exactly the (iterated) It{\^o} integrals against $W$, in the sense that for $f_m \colon [0,T]^m \to \RR^{n \times n \times \dots \times n} \cong \RR^{(n^m)}$ totally symmetric and square integrable
\begin{align*}
I_m(f_m) = \sum_{i_1, \dots, i_m = 1}^n \int_0^T \int_0^{t_m} \dots \int_{0}^{t_2} f_{m, i_1, \dots, i_m}(t_1, \dots, t_m) \dd W_{t_1}^{i_1} \dots \dd W_{t_m}^{i_m}\,.
\end{align*}
Here, the right-hand side is not just notation, but iterated It{\^o} integrals from stochastic integration theory. In particular, by the HS kernel theorem~\ref{thm:hs-kernel}, the coefficient function $f_2 \colon [0,T]^2 \to \RR^{n \times n}$ for the second homogeneous Wiener chaos component of a given ${\cal G}$-measurable random variable $Y$ will always define a HS operator with $f_2$ as its integral kernel. Conversely, any symmetric HS operator on $L^2\left([0,T], \RR^n \right)$ defines a random variable in the second homogeneous Wiener chaos $\mathscr{H}_2$.\\

More concretely, we remark that since $L^2(\Omega, {\cal G}, \PP)$ decomposes into the direct orthogonal sum of the Wiener chaos subspaces~$\mathscr{H}_m$, we can decompose any real-valued random variable $Y$ that is measurable with respect to the sigma algebra generated by a Brownian motion in terms of its chaos components:
\begin{align*}
Y = \sum_{m = 0}^\infty I_m\left(f_m^{Y} \right)\,,
\end{align*}
where we emphasize the dependence on $Y$. To find the (deterministic) coefficient functions $f_m^{Y} \colon\allowbreak [0,T]^m \to \RR^{n \times n \times \dots \times n}$ for a given~$Y$, one can either explicitly rewrite $Y$ in terms of iterated integrals, or, as e.g.\ in~\cite{grasselli-hurd:2005}, proceed via Fr{\'e}chet derivatives of the generating functional $Z_Y \colon {\cal H} \to \RR$ with
\begin{align}
Z_Y(h) = \EE \left[ Y \exp \left\{ W(h) - \frac{1}{2} \int_0^T \norm{h(t)}^2 \dd t \right\} \right]\,,
\label{eq:gen-func-chaos}
\end{align}
for which
$f_m^Y = \delta^m Z_Y / \delta h^m$ at $h = 0$.

\subsubsection{Exponentials of the second Wiener--It{\^o} chaos component and Carleman--Fredholm determinants}
\label{sec:exp-int-cf-det}

At this point, it should be clear that to prove~\eqref{eq:mgf-prefac} for the leading-order MGF tail prefactor $R_\lambda$, one performs a Taylor expansion of the noise-to-observable map $F$ up to second order around the large deviation minimizer, and then uses a Wiener-It{\^o} chaos decomposition of the resulting real-valued random variable $Y = \tfrac{1}{2}\hat{Q}_\lambda$ (later defined in~\eqref{eq:qhat}) in the exponent of $R_\lambda = \EE \left[\exp \left\{Y \right\} \right]$. Since the Taylor expansion terminates at second order, we will have $Y \in \mathscr{H}_0 \oplus \mathscr{H}_2$, with $f_1^Y = 0$ because of stationarity, where it then remains to show that the zeroth chaos component $f_0^Y = \EE \left[ Y \right]$ is indeed finite, and that $f_2^Y$ is square-integrable. We have the following proposition that allows for calculating $R_\lambda$ from this decomposition, which has repeatedly appeared throughout the literature~\cite{janson:1997,ben-arous:1988,dunford-schwartz:1988,grasselli-hurd:2005,ferreiro-castilla-utzet:2011,inahama:2013}:

\begin{proposition}
\label{prop:cf-det-expec}
Let $Y \in \mathscr{H}_0 \oplus \mathscr{H}_2$ with
\begin{align*}
Y = \EE \left[Y \right] + \int_{0}^T \int_0^{t_2} {\cal B}_{ij}(t_1,t_2) \dd W^i_{t_1} \dd W^j_{t_2}\,,
\end{align*}
where the symmetric HS integral operator $B$ associated with the integral kernel ${\cal B}$ is such that $\Id - B$ is positive definite (cf.\ assumption~\ref{assump-posdef}). Then
\begin{align*}
\EE \left[\exp \left\{Y \right\} \right] = {\det}_2 \left(\Id - B \right)^{-1/2} \exp \left\{\EE \left[Y \right] \right\}\,.
\end{align*}
\end{proposition}

\begin{proof}
Using the isometry from proposition~\ref{prop:wiener-iso} and decomposition of ${\cal B}$ in terms of eigenvalues $\mu_i$ and orthonormal eigenfunctions $\hat{e}_i$ of $B$, we have
\begin{align*}
I_2\left({\cal B} \right) \overset{\text{linear}}{=} \sum_{i = 1}^\infty \mu_i \, I_2\left(\hat{e}_i \otimes \hat{e}_i \right)\overset{\eqref{eq:wiener-iso-def}}{=} \frac{1}{2} \sum_{i = 1}^\infty \mu_i \, \mathrm{He}_2 \left(\int_0^T \left \langle \hat{e}_i(t) , \dd W_t \right \rangle \right) \overset{d}{=} \frac{1}{2} \sum_{i = 1}^\infty \mu_i \left(\xi_i^2 - 1 \right)
\end{align*}
with $\xi_i = \int_0^T \left \langle \hat{e}_i(t) , \dd W_t \right \rangle$ independent and standard normally distributed. But then
\begin{align*}
\EE \left[\exp \left\{Y \right\} \right] &= \exp \left\{\EE \left[Y \right] \right\} \prod_{i = 1}^\infty \EE \left[\exp \left\{ \frac{1}{2} \mu_i \left(\xi_i^2 - 1 \right) \right\} \right] \\
&= \exp \left\{\EE \left[Y \right] \right\} \prod_{i = 1}^\infty  \exp \left\{-\frac{1}{2} \mu_i \right\} \frac{1}{\sqrt{2 \pi}}\int_{-\infty}^\infty \exp \left\{ -\frac{1}{2} \left(1-\mu_i\right) \xi_i^2\right\} \dd \xi_i \\
&= \exp \left\{\EE \left[Y \right] \right\} \prod_{i = 1}^\infty  \exp \left\{-\frac{1}{2} \mu_i \right\} \frac{1}{\sqrt{1 - \mu_i}} = {\det}_2 \left(\Id - B \right)^{-1/2} \exp \left\{\EE \left[Y \right] \right\}
\end{align*}
by the definition and convergence of the CF determinant from subsection~\ref{sec:operator-dets}.
\end{proof}

\subsubsection{Classical and abstract Wiener space and Goodman's theorem}
\label{sec:wiener-space-goodman}

The general setting for the following calculations can be formulated as an abstract Wiener space~\cite{gross:1967,kuo:2006} (even though we will in fact only need the \textit{classical} Wiener space here). Let
\begin{enumerate}
\item ${\cal H}$ be a real, separable Hilbert space with inner product $\left \langle \cdot, \cdot \right \rangle_{\cal H}$ and norm $\norm{\cdot}_{\cal H}$ -- this is the Cameron--Martin space,
\item $\mathscr{B} \overset{i}{\hookleftarrow} {\cal H}$ be a Banach space, the completion of ${\cal H}$ with respect to a measurable (with respect to a Gaussian measure on~$\mathscr{B}$) semi-norm $\norm{\cdot}_{\mathscr{B}}$ -- this is the ambient space on which the Wiener measure can actually be defined,
\item $\mathscr{B}^*$ be the continuous dual space of $\mathscr{B}$, which by restriction of domain and the Riesz representation theorem is a subset of ${\cal H}$.
\end{enumerate}
Then the tuple $(i, {\cal H}, \mathscr{B})$ is called abstract Wiener space; schematically $\mathscr{B}^* \overset{r}{\rightarrow} {\cal H}^* \simeq {\cal H} \overset{i}{\hookrightarrow} \mathscr{B}$. The most important example, and indeed the one that we will exclusively consider in the following, is the classical Wiener space, where
\begin{enumerate}
\item ${\cal H} = H_{0}^{1,2}\left([0,T],\RR^n \right)$ with $\left \langle g,h \right \rangle_1 = \int_0^T \left \langle \dot{g}_t, \dot{h}_t \right \rangle \dd t$, i.e.\ the space of absolutely continuous functions with initial value $0$ whose weak derivative is square integrable,
\item $\mathscr{B} = C_0\left([0,T],\RR^n \right)$ with $\norm{h}_\infty = \sup_{t \in [0,T]} \norm{h_t}$, i.e.\ the space of continuous functions with initial value $0$ with the supremum norm.
\end{enumerate}

In our present setting, the solution map for the SDE~\eqref{eq:SDE} is for example actually defined as $\Phi \colon {\cal H} \to {\cal H}$, $h \mapsto \phi$ where $\dd \phi_t = b(\phi_t) \dd t + \sigma(\phi_t) \dd h_t$ with $\phi_0 = 0$. The (scaled) noise-to-observable map on the Cameron--Martin space is then defined as $\lambda F = \lambda \tilde{f} \circ \Phi \colon {\cal H} \to \RR$ with $\tilde{f} \colon {\cal H} \to \RR$, $\phi \mapsto \tilde{f}(\phi) =f(\phi_T)$. Its second variation operator maps $A_\lambda \colon {\cal H} \to {\cal H}$ by definition. Note that for the following subsections, we will work in terms of these \textit{redefinitions} of the noise-to-observable map $F$ and operators $A_\lambda$ and $\tilde{A}_\lambda$, now defined on ${\cal H}$, instead of $L^2([0,T],\RR^n)$ in the introduction and heuristic motivation in subsection~\ref{sec:heur}.\\

To show that part of the Hessian $A_\lambda - \tilde{A}_\lambda \colon {\cal H} \to {\cal H}$ is TC below in subsection~\ref{sec:q-hat-wiener-chaos}, we need Goodman's theorem for abstract Wiener space~\cite{kuo:2006}

\begin{theorem}
\label{thm:goodman}
A continuous linear operator $D \colon \mathscr{B} \to \mathscr{B}^*$ is TC when its domain is restricted to~${\cal H}$.
\end{theorem}

or more precisely the following corollary, for $D = A_\lambda - \tilde{A}_\lambda$:

\begin{corollary}
\label{cor:goodman}
Let $D \colon {\cal H} \to {\cal H}$ be a symmetric linear bounded operator (which, notably, is a priori only defined on ${\cal H}$). Assume that it satisfies the inequality
\begin{align}
\abs{\left \langle D g, h \right \rangle_{\cal H}} \leq c \cdot \norm{g}_{\mathscr{B}} \cdot \norm{h}_{\mathscr{B}}
\label{eq:quadrat-form-on-H}
\end{align}
for the bilinear form $g,h \mapsto \left \langle D g, h \right \rangle_{\cal H}$ on ${\cal H} \times {\cal H}$ for any $g, h \in {\cal H}$, where the constant $c \geq 0$ does not depend on $g$ and $h$. Then~$D$ is TC.
\end{corollary}

\begin{proof}
Assuming~\eqref{eq:quadrat-form-on-H} holds,
we can extend the bilinear form $(g, h) \mapsto \left \langle g, Dh \right \rangle_{\cal H}$, defined on ${\cal H} \times {\cal H}$, to a continuous bilinear form ${\cal Q} \colon \mathscr{B} \times \mathscr{B} \to \RR$ by density of ${\cal H}$ in $\mathscr{B}$.
Explicitly, for $\mathfrak{g}, \mathfrak{h} \in \mathscr{B}$, we can find sequences $(g_n)_{n \in \NN}, (h_n)_{n \in \NN} \subset {\cal H}$,
such that $g_n \to \mathfrak{g}$ and $h_n \to \mathfrak{h}$ with respect to $\norm{\cdot}_\mathscr{B}$.
Then we set ${\cal Q}\left(\mathfrak{g}, \mathfrak{h} \right) = \lim_{n\to \infty} \left \langle g_n, D h_n \right \rangle_H$,
which is well-defined and continuous because of~\eqref{eq:quadrat-form-on-H}. Now, the partial insertion map ${\cal Q}_1 \colon \mathscr{B} \to \mathscr{B}^*$, $\mathfrak{g} \mapsto {\cal Q}_1(\mathfrak{g}) = {\cal Q}\left(\mathfrak{g}, \cdot \right)$
defines a map into the (continuous) dual space $\mathscr{B}^*$ of $\mathscr{B}$,
and the map ${\cal Q}_1$ itself is bounded and linear, where linearity is clear and boundedness follows from
\begin{align*}
\norm{{\cal Q}_1(\mathfrak{g})}_{\mathscr{B}^*} = \inf \left\{ c' \geq 0 \mid \forall \mathfrak{h} \in \mathscr{B}: \; \abs{{\cal Q}_1(\mathfrak{g})(\mathfrak{h})} \leq c' \norm{\mathfrak{h}}_\mathscr{B} \right\} \leq c \norm{\mathfrak{g}}_\mathscr{B}\,.
\end{align*}
Then, by Goodman's theorem~\ref{thm:goodman}, restricting ${\cal Q}_1$ to ${\cal H}$ yields a TC operator, but when we are already on ${\cal H}$, ${\cal Q}_1(g) = \left \langle D g, \cdot \right \rangle_{\cal H}$ by construction, so after the Riesz isometry, we conclude that $D$ itself is TC.
\end{proof}

We remark that, as we will see, the inequality~\eqref{eq:quadrat-form-on-H} relates exactly to the reasoning presented in figure~\ref{fig:flowcharts} as to why $A_\lambda - \tilde{A}_\lambda$ is TC. This is because cutting off the dashed red arrows in figure~\ref{fig:flowcharts} to get from~$A_\lambda$ to $A_\lambda - \tilde{A}_\lambda$ makes the estimate~\eqref{eq:quadrat-form-on-H} possible in the first place, as there are in some sense enough time derivatives in the resulting expressions.

\subsection{Leading order contribution of minimizer}
\label{sec:leading-MGF-contr}

First, we note that the leading-order exponential term in the asymptotic expansion~\eqref{eq:mgf-def-asymp} is directly obtained from Varadhan's lemma applied to the family of $C_0([0,T],\RR^n)$-valued random variables $(\sqrt{\eps} W)_{\eps > 0}$. This family of Brownian motions satisfies a large deviation principle as $\eps \downarrow 0$ with Schilder's rate function $S[h] = \tfrac{1}{2} \norm{\dot{h}}_{L^2}^2 \mathds{1}_{h \in {\cal H}} + \infty  \mathds{1}_{h \not \in {\cal H}}$, where ${\cal H} = H_{0}^{1,2}\left([0,T],\RR^n \right)$ as in the previous subsection. With this, we immediately get
\begin{align*}
\lim_{\eps \downarrow 0} \left( \eps \log J^\eps(\lambda) \right) = \lim_{\eps \downarrow 0} \left( \eps \log \EE \left[\exp \left\{\frac{\lambda}{\eps} f\left(X_T^\eps\right) \right\} \right] \right) = - \inf_{h \in {\cal H}} \left( S[h] - \lambda F[h] \right) \overset{\eqref{eq:legendre}}{=} I^*(\lambda)\,,
\end{align*}
by following the same line of reasoning as in subsection~\ref{sec:heur}. We denote the large deviation minimizer for this problem, which we assume to exist, by $h_\lambda \in {\cal H}$, with $\eta_\lambda = \dot{h}_\lambda \in L^2([0,T],\RR^n)$ in the notation of the introduction.

\subsection{Taylor expansion around minimizer}
\label{sec:loc-taylor-MGF}

Skipping a number of technical details and localization arguments of~\cite{ben-arous:1988}, the first step in the proof of the MGF prefactor formula~\eqref{eq:mgf-prefac}, as the next-to-leading-order correction to the result from Varadhan's lemma, is now to shift the Brownian motion $\sqrt{\eps} W$ of the SDE~\eqref{eq:SDE} by the large deviation minimizer $h_\lambda$ using Girsanov's theorem. This leads to
\begin{align*}
J^\eps(\lambda) &=  \EE \left[\exp \left\{\frac{\lambda}{\eps} f\left(X_T^\eps\right) \right\} \right]\\&=  \EE \left[\exp \left\{\frac{\lambda}{\eps} f\left(Z_T^\eps\right) \right\} \exp \left\{ -\frac{1}{\sqrt{\eps}} \int_0^T \left \langle \dot{h}_\lambda, \dd W \right \rangle - \frac{1}{2 \eps} \int_0^T \norm{\dot{h}_\lambda}^2 \dd t \right\} \right]\,,
\end{align*}
where the shifted process is given by
\begin{align*}
\begin{cases}
\dd Z_t^\eps = \left[ b(Z_t^\eps) + \sigma(Z_t^\eps) \dot{h}_\lambda \right] \dd t + \sqrt{\eps} \sigma(Z_t^\eps) \dd W_t\,,\\
Z_0^\eps = x\,.
\end{cases}
\end{align*}
Then, given a realization of $Z^\eps$, an ordinary one-dimensional Taylor expansion of $M(\eps) := \lambda f(Z_T^\eps)$ up to second order in $\sqrt{\eps}$ around $\eps = 0$ is used, with remainder term in integral form:
\begin{align*}
M(\eps) = M(0) + \sqrt{\eps} M'(0) + \eps \int_0^1 M''(\eps u) (1 - u) \dd u\,.
\end{align*}
Using this expansion,~\cite{ben-arous:1988} goes on to show that asymptotically
\begin{align}
\begin{aligned}
J^\eps(\lambda) &\overset{\eps \downarrow 0}{\sim} \EE \bigg[\exp \left\{\frac{1}{\eps} \left(M(0) + \sqrt{\eps} M'(0) + \tfrac{1}{2} \eps \hat{Q}_\lambda \right) \right\} \times \\ &\qquad \times \exp \left\{ -\frac{1}{\sqrt{\eps}} \int_0^T \left \langle \dot{h}_\lambda, \dd W \right \rangle - \frac{1}{2 \eps} \int_0^T \norm{\dot{h}_\lambda}^2 \dd t \right\} \bigg]
\end{aligned}
\label{eq:taylor-mgf}
\end{align}
holds (where we write $\hat{Q}_\lambda := M''(0)$) by suitably bounding the remainder term. We will just use this here, and merely calculate the appearing terms. Obviously, we have $M(0) = \lambda f(\phi_\lambda(T))$, such that the exponential $O(1/\eps)$ terms in~\eqref{eq:taylor-mgf} give~$I^*(\lambda)$ as expected. The next term becomes
\begin{align*}
M'(0) = \left \langle \lambda \nabla f(\phi_\lambda(T)), \gamma^{(1)}_T \right \rangle = \int_0^T \left \langle \sigma^\top(\phi_\lambda) \theta_\lambda , \dd W \right \rangle\,,
\end{align*}
such that the exponential $O(1/\sqrt{\eps})$ terms in~\eqref{eq:taylor-mgf} cancel by the first-order optimality conditions. Note that $\gamma^{(1)} = \gamma^{(1)}[W]$ here is defined analogously to~\eqref{eq:second-order-adjoint-mult-noise}, but this is now an SDE driven by the Brownian motion $W$:
\begin{align*}
\begin{cases}
\dd \gamma^{(1)}_t = L\left[\dot{h}_\lambda, \phi_\lambda\right] \gamma^{(1)}_t \dd t + \sigma(\phi_\lambda) \dd W_t\,,\\
\gamma^{(1)}_0 = 0\,.
\end{cases}
\end{align*}
The second derivative of $M$ at $\eps = 0$ is given by
\begin{align}
\hat{Q}_\lambda = \left \langle \gamma^{(1)}_T, \lambda \nabla^2 f(\phi_\lambda(T)) \gamma^{(1)}_T \right \rangle + \left \langle \lambda \nabla f(\phi_\lambda(T)), \gamma^{(2)}_T \right \rangle\,,
\label{eq:qhat}
\end{align}
with
\begin{align}
\begin{cases}
\dd \gamma_t^{(2)} = \left[ L\left[\dot{h}_\lambda, \phi_\lambda\right] \gamma^{(2)}_t + \left(\nabla^2 b(\phi_\lambda) + \nabla^2 \sigma(\phi_\lambda) \dot{h}_\lambda \right) : \left(\gamma_t^{(1)} \right)^{\otimes 2} \right] \dd t + 2 \gamma^{(1)}_t \nabla \sigma(\phi_\lambda) \dd W_t\,,\\
\gamma^{(2)}_0 = 0\,.
\label{eq:gamma-2}
\end{cases}
\end{align}
Importantly, $\gamma^{(1)}$ is driven by the same Brownian motion $W$, so the last term in~\eqref{eq:gamma-2} corresponds to an iterated It{\^o} integral. Putting everything together, we arrive at
\begin{align}
J^\eps(\lambda) \overset{\eps \downarrow 0}{\sim} R_\lambda \exp \left\{\frac{1}{\eps} I^*(\lambda) \right\} \quad \text{with} \quad R_\lambda = \EE \left[\exp \left\{\frac{1}{2} \hat{Q}_\lambda \right\} \right]\,,
\label{eq:mgf-intermediate-expec}
\end{align}
and hence need to consider properties of $\hat{Q}_\lambda$ in the following.

\subsection{Wiener chaos decomposition of $\hat{Q}_\lambda$: Second variation $A_\lambda$ Hilbert--Schmidt, $A_\lambda-\tilde{A}_\lambda$  trace-class}
\label{sec:q-hat-wiener-chaos}

In order to conclude with the result~\eqref{eq:mgf-prefac} from~\eqref{eq:mgf-intermediate-expec}, it remains to analyze the real-valued random variable $\hat{Q}_\lambda$, which depends on the Brownian motion $W$, by performing a Wiener chaos decomposition (which can, due to the definition of $\hat{Q}_\lambda$ in~\eqref{eq:qhat}, only contain terms in the zeroth and second homogeneous chaos)
\begin{align*}
\tfrac{1}{2}\hat{Q}_{\lambda} = \tfrac{1}{2} \EE \left[\hat{Q}_{\lambda} \right] + \int_0^T \int_0^t f_2^{\hat{Q}_\lambda / 2}(t,s) \dd W_s \dd W_t = \tfrac{1}{2} \EE \left[\hat{Q}_{\lambda} \right] + I_2\left(f_2^{\hat{Q}_\lambda / 2}\right)\,.
\end{align*}
Specifically, one needs to (a) connect the chaos coefficients to the classical second variation operator $A_\lambda$ of the noise-to-observable map as well as the trace of $A_\lambda - \tilde{A}_\lambda$ from the introduction, and (b) show that these terms in the Wiener chaos expansion are actually well-defined and hence $\tfrac{1}{2}\hat{Q}_{\lambda} \in \mathscr{H}_0 \oplus \mathscr{H}_2$, with finite expectation $\EE \left[\hat{Q}_{\lambda} \right]$ and HS operator corresponding to the kernel $f_2^{\hat{Q}_\lambda / 2}$. Evaluating~$R_\lambda$ afterwards is then simply a matter of applying proposition~\ref{prop:cf-det-expec}.\\

To address (a), using the linearity of the equations, we split off the explicit iterated integral term in $\gamma^{(2)} = \gamma^{(2)}_{\text{reg}} +  \gamma^{(2)}_{\text{sing}}$ appearing in ~\eqref{eq:qhat} via
\begin{align*}
\begin{cases}
\dd \gamma_{\text{reg}}^{(2)} = \left[ L\left[\dot{h}_\lambda, \phi_\lambda\right] \gamma^{(2)}_{\text{reg}} + \left(\nabla^2 b(\phi_\lambda) + \nabla^2 \sigma(\phi_\lambda) \dot{h}_\lambda \right) : \left(\gamma_t^{(1)} \right)^{\otimes 2} \right] \dd t\,,\\
\gamma^{(2)}_{\text{reg},0} = 0\,,
\end{cases}
\end{align*}
and
\begin{align}
\begin{cases}
\dd \gamma_{\text{sing}}^{(2)} = L\left[\dot{h}_\lambda, \phi_\lambda\right] \gamma^{(2)}_{\text{sing}} \dd t + 2 \gamma^{(1)}_t \nabla \sigma(\phi_\lambda) \dd W_t\,,\\
\gamma^{(2)}_{\text{sing},0} = 0\,,
\label{eq:gamma-sing}
\end{cases}
\end{align}
and define $\hat{Q}_\lambda = \hat{Q}_\lambda^{\text{reg}} + \hat{Q}_\lambda^{\text{sing}}$, where
\begin{align*}
\hat{Q}_\lambda^{\text{sing}} &= \left \langle \lambda \nabla f(\phi_\lambda(T)), \gamma^{(2)}_{\text{sing},T} \right \rangle\\
\hat{Q}_\lambda^{\text{reg}} &= \left \langle \gamma^{(1)}_T, \lambda \nabla^2 f(\phi_\lambda(T)) \gamma^{(1)}_T \right \rangle + \left \langle \lambda \nabla f(\phi_\lambda(T)), \gamma^{(2)}_{\text{reg},T} \right \rangle\,.
\end{align*}
It is straightforward to see that $\EE \left[\hat{Q}_\lambda^{\text{sing}} \right] = 0$ because of the It{\^o} integral and non-anticipating integrand in~\eqref{eq:gamma-sing}; more explicitly, the random variable $\hat{Q}_\lambda^{\text{sing}}$ is already an iterated It{\^o} integral with
\begin{align*}
\hat{Q}_\lambda^{\text{sing}} = 2\int_0^T \int_0^t \tilde{\cal A}_\lambda (t,t') \dd W_{t'} \dd W_{t}
\end{align*}
in terms of the symmetric integral kernel $\tilde{\cal A}_\lambda \colon [0,T]^2 \to \RR^{n \times n}$ of $\tilde{A}_\lambda$ as introduced in~\eqref{eq:a-tilde-kernel} in subsection~\ref{sec:heur}. Denoting the integral kernel of the operator $D = A_\lambda - \tilde{A}_\lambda$ by ${\cal D} = {\cal A}_\lambda - \tilde{\cal A}_\lambda$, one can show through direct computation, using the definitions of $A_\lambda$ and $\tilde{A}_\lambda$ from the introduction, as well as the definition of $\hat{Q}_\lambda^{\text{reg}}$ above, that
\begin{align*}
\hat{Q}_\lambda^{\text{reg}} = \underbrace{\int_0^T \trace \left[{\cal D} \right](t,t) \dd t}_{= \EE \left[\hat{Q}_\lambda^{\text{reg}} \right] \text{ by It{\^o} isometry}} + \underbrace{2 \int_0^T \int_0^t {\cal D}(t,t') \dd W_{t'} \dd W_t}_{\text{e.g.\ via generating functional~\eqref{eq:gen-func-chaos}}}\,,
\end{align*}
such that, if $A_\lambda - \tilde{A}_\lambda$ is indeed trace class and $A_\lambda$ HS, we arrive at
\begin{align*}
\tfrac{1}{2}\hat{Q}_{\lambda} = \tfrac{1}{2} \trace \left[A_\lambda - \tilde{A}_\lambda \right] + I_2(A_\lambda) \in \mathscr{H}_0 \oplus \mathscr{H}_2\,,
\end{align*}
leading to theorem~\ref{thm:mgf-prefac} via proposition~\ref{prop:cf-det-expec} and equation~\eqref{eq:mgf-intermediate-expec}.\\

Finally, to address point (b) above,~\cite{ben-arous:1988} uses corollary~\ref{cor:goodman} of Goodman's theorem to show that $D = A_\lambda - \tilde{A}_\lambda \colon {\cal H } \to {\cal H}$ is TC. For $g, h \in {\cal H}$, the left-hand side of~\eqref{eq:quadrat-form-on-H}, using the definition of $D$ acting on functions in ${\cal H}$, becomes
\begin{align*}
\abs{\left \langle g, D h \right \rangle_{1}} \leq \abs{\left \langle \gamma^{(1)}[g](T), \lambda \nabla^2 f(\phi_\lambda(T)) \gamma^{(1)}[h](T) \right \rangle} + \abs{\left \langle \lambda \nabla f(\phi_\lambda(T)), \gamma^{(2)}_{\text{reg}}[g,h](T) \right \rangle}
\end{align*}
with
\begin{align*}
\begin{cases}
\dd \gamma^{(1)}[g] = L \gamma^{(1)}[g] \dd t + \sigma(\phi_\lambda) \dd g \,,\\
\gamma^{(1)}[g](0) = 0\,,
\end{cases}
\end{align*}
and
\begin{align*}
\begin{cases}
  \dd  \gamma^{(2)}_{\text{reg}}[g,h] = L \gamma^{(2)}_{\text{reg}}[g,h] \dd t + \left(\nabla^2 b + \nabla^2 \sigma \eta_\lambda\right):\left(\gamma^{(1)}[g] \otimes \gamma^{(1)}[h]\right) \dd t\,,\\
  \gamma^{(2)}_{\text{reg}}[g,h](0) = 0\,.
\end{cases}
\end{align*}
Writing the propagator for $L = L[\phi_\lambda, \eta_\lambda] = \nabla b(\phi_\lambda) + \nabla \sigma(\phi_\lambda) \eta_\lambda$ as ${\cal P}$ again, we see that
\begin{align*}
\norm{\gamma^{(1)}[g](T)|} &= \norm{\int_0^T {\cal P}(T,s) \sigma(\phi_\lambda(s)) \dot{g}(s) \, \dd s} \\
&\overset{\text{i.b.p.}}{=} \norm{\int_0^T \dv{}{s}\left( {\cal P}(T,s) \sigma(\phi_\lambda(s)) \right) g(s) \, \dd s} \leq \text{const} \cdot \norm{g}_\infty\,,
\end{align*}
so that $\abs{\left \langle D g, h \right \rangle_1} \leq c \cdot \norm{g}_\infty \cdot \norm{h}_\infty$ follows after a few triangle inequalities and supremum estimates (assuming $\eta_\lambda \in C([0,T],\RR^n)$, such that the propagator ${\cal P}$ is continuously differentiable~\cite{teschl:2024}). Note that the subtraction of $\tilde{A}_\lambda$ is indeed crucial for this estimate to go through, as otherwise the differential equation for the \textit{full} $\gamma^{(2)}[g,h]$ would contain $\left(\nabla\sigma \dot{g}\right)\gamma^{(1)}[h] + \left(\nabla\sigma \dot{h}\right)\gamma^{(1)}[g]$ with ``bare'' derivative terms $\dot{g}$ and $\dot{h}$, and hence cannot be bounded in the same way. This corresponds to the intuitive explanation in section~\ref{sec:theory-result} of why $A_\lambda - \tilde{A}_\lambda$ is TC, while $A_\lambda$ itself will not be TC in general, cf.\ figure~\ref{fig:flowcharts}.\\

Lastly, it is obvious that $\tilde{A}_\lambda\colon {\cal H} \to {\cal H}$ is HS by the HS kernel theorem~\ref{thm:hs-kernel}, since the integral kernel~$\tilde{\cal A}_\lambda$ as given explicitly in~\eqref{eq:a-tilde-kernel} is square-integrable, so $\tilde{A}_\lambda\colon L^2([0,T], \RR^n) \to L^2([0,T], \RR^n)$ acting on $L^2$ is HS, and ${\cal H}$ and $L^2([0,T], \RR^n)$ are isometric via integration / differentiation. Thus, $A_\lambda = A_\lambda - \tilde{A}_\lambda + \tilde{A}_\lambda$ as the sum of two HS operators is HS, thereby concluding our outline of the proof of theorem~\ref{thm:mgf-prefac} given in~\cite{ben-arous:1988}.

\subsection{Transforming from MGFs to tail probabilities}
\label{sec:mgf-to-prob}

After establishing the asymptotic expansion~\eqref{eq:mgf-def-asymp} with prefactor~\eqref{eq:mgf-prefac} for the MGF in case of a strictly convex rate function, it is -- at least on a formal level --  simple to conclude with theorem~\ref{thm:prefac} for tail probabilities under this assumption. We first state the following general

\begin{proposition}
\label{prop:tail-trafo}
If a sharp asymptotic expansion $J^\eps(\lambda) \sim R_\lambda \exp \left\{\eps^{-1} I^*(\lambda) \right\}$ of the MGF holds as $\eps \downarrow 0$, then the tail probability satisfies
\begin{align*}
P^\eps(z) \overset{\eps \downarrow 0}{\sim}  \eps^{1/2} (2 \pi)^{-1/2}
  \, C(z) \,\exp\left\{-\frac1\eps I(z)\right\}
\end{align*}
with
\begin{align}
C(z) = R_{\lambda_z} \sqrt{I''(z)} \lambda_z^{-1}\,,
\label{eq:cz-from-mgf-trafo}
\end{align}
where $\lambda_z = I'(z)$ is the Lagrange multiplier associated with $z$.
\end{proposition}

This does not rely on the specific form of the prefactor $R_\lambda$, but merely follows from successive applications of the saddlepoint method, as $\eps \downarrow 0$, to the contour integral for the inverse Laplace transform to convert from the MGF to the probability density function of $f(X_T^\eps)$, as well as Laplace's method for the integral of the probability density function to get the tail probability~\cite{deuschel-etal:2014,schorlepp-grafke-grauer:2023,schorlepp-tong-grafke-etal:2023}. The main mathematical difficulty in rigorously proving this proposition, which we do not address here, is controlling the error term in the asymptotic expansion.\\

It is desirable to simplify this result further, and hence obtain theorem~\ref{thm:prefac}, for at least two reasons: First, note that~\eqref{eq:cz-from-mgf-trafo} contains the second derivative of the rate function at $z$, which is unknown in general, and e.g.\ needs to be determined via finite differences of $\lambda_z$ for different $z$'s. While this is in principle feasible to do, it is in a certain sense against the spirit of having a complete tail probability estimate for a given $z$ from solving a single optimization problem and evaluating determinants for that particular~$z$. Second, and more importantly, as we comment below, the tail probability prefactor in terms of~\eqref{eq:cz-from-mgf-trafo} is only well-defined for strictly convex rate functions $I$ (with $I''(z) > 0$ and well-defined $R_{\lambda_z}$), whereas the expression given in theorem~\ref{thm:prefac} is in fact possible to evaluate independently of the convexity of the rate function, as long as the second-order sufficient optimality condition~\eqref{eq:sufficent-optimal} (and differentiability) holds.\\

Using proposition~\ref{prop:tail-trafo}, we can now finally derive theorem~\ref{thm:prefac}. As noted already in~\cite{schorlepp-tong-grafke-etal:2023}, by differentiating the first-order optimality condition for the instanton noise $\eta_z$
\begin{align*}
\eta_z = \lambda_z \left. \frac{\delta F}{\delta \eta} \right|_{\eta_z}\,,
\end{align*}
with respect to $z$ and recalling $\lambda_z = I'(z) = \left \langle \eta_z, \partial_z \eta_z \right \rangle$, we obtain that
\begin{align}
\partial_z \eta_z = \frac{I''(z)}{\lambda_z} \eta_z + A_{\lambda_z} \partial_z\eta_z \quad \Leftrightarrow \quad \left[\text{Id} -  A_{\lambda_z}\right] \partial_z \eta_z = \frac{I''(z)}{\lambda_z} \eta_z\,.
\label{eq:deriv-optimality}
\end{align}
From~\eqref{eq:cz-from-mgf-trafo}, we then see that
\begin{align}
\begin{aligned}
C(z) &= R_{\lambda_z} \sqrt{I''(z)} \lambda_z^{-1} \overset{\eqref{eq:deriv-optimality}}{=} \frac{R_{\lambda_z}}{\left \langle \eta_z, \left(\Id - A_{\lambda_z} \right)^{-1} \eta_z \right \rangle_{L^2}^{1/2}} \\
&\overset{\eqref{eq:mgf-prefac}}{=} \left[ 2 I(z) \left \langle  e_z, \left(\Id - A_{\lambda_z} \right)^{-1} e_z \right \rangle {\det}_2 \left( \Id - A_{\lambda_z} \right)  \right]^{-1/2} \exp \left\{\frac{1}{2} \trace \left[A_{\lambda_z} - \tilde{A}_{\lambda_z} \right] \right\}\,.
\end{aligned}
\label{eq:tail-prefac-intermediate}
\end{align}
To make progress at this point, we note that Cramer's rule for finite-rank or TC operators~$A_{\lambda_z}$ and unit vectors $e_z$ reads~\cite{mckean:2011,schorlepp-tong-grafke-etal:2023}
\begin{align}
\left \langle e_z, \left[ \text{Id} - A_{\lambda_z} \right]^{-1} e_z \right \rangle = \frac{\det \left(\text{Id} - \text{pr}_{\eta_z^\perp }A_{\lambda_z} \text{pr}_{\eta_z^\perp} \right)}{\det \left(\text{Id} -  A_{\lambda_z} \right)}\,.
\label{eq:trace-class-adjugate}
\end{align}
Based on this, and the observation that for TC operators $B:L^2\left([0,T],\RR^n\right) \to L^2\left([0,T],\RR^n\right)$, we have
\begin{align}
\begin{aligned}
\det \left( \exp \left\{B \right\} \right) = \exp \left\{ \trace [ B ] \right\} &= \exp \left\{ \trace \left[ \left(\text{Id} - e_z^{\otimes 2} + e_z^{\otimes 2} \right) B \left(\text{Id} - e_z^{\otimes 2} + e_z^{\otimes 2} \right) \right] \right\}\\ &= \exp \left\{ \trace \left[ \text{pr}_{\eta_z^\perp} B \text{pr}_{\eta_z^\perp} \right] \right\} \exp \left\{\left \langle e_z, B e_z \right \rangle \right\}\,,
\end{aligned}
\label{eq:trace-projection}
\end{align}
we propose the following generalization of Cramer's rule to HS operators  $A_{\lambda_z}$ (cf.~\cite{ikeda-manabe:1996})
\begin{align}
\left \langle e_z, \left[ \text{Id} - A_{\lambda_z} \right]^{-1} e_z \right \rangle = \frac{{\det}_2 \left(\text{Id} - \text{pr}_{\eta_z^\perp }A_{\lambda_z} \text{pr}_{\eta_z^\perp} \right)}{{\det}_2 \left(\text{Id} -  A_{\lambda_z} \right)} \exp \left\{  \left \langle e_z, A_{\lambda_z} e_z \right \rangle \right\}\,,
\label{eq:hilbert-schmidt-adjugate}
\end{align}
by requiring that the expression reduces to~\eqref{eq:trace-class-adjugate} if $A_{\lambda_z}$ does happen to be TC instead of just HS. More rigorously, one can use a continuity argument to show~\eqref{eq:hilbert-schmidt-adjugate} via finite-rank  approximations $A_{\lambda_z}^{(k)}$ to a HS operator~$A_{\lambda_z}$, since for all such approximations $A_{\lambda_z}^{(k)}$ the equality~\eqref{eq:hilbert-schmidt-adjugate} holds, and the left-hand side and right-hand side will converge as $k \to \infty$. To complete the derivation of the main result~\eqref{eq:tail-prefac-projected} of this paper from~\eqref{eq:tail-prefac-intermediate}, we simply apply Cramer's rule~\eqref{eq:hilbert-schmidt-adjugate}, as well as~\eqref{eq:trace-projection} with $B = A_{\lambda_z} - \tilde{A}_{\lambda_z}$, to equation~\eqref{eq:tail-prefac-intermediate}.

\subsection{Remarks on the convexity of the rate function}
\label{sec:convex-rf}

Note that throughout this appendix, we have assumed that the rate function is strictly convex and differentiable, in order to have a finite MGF for all $\lambda \in \RR$, and a bijective relation between $\lambda$ and $z$. We briefly outline an argument here that the main result, theorem~\ref{thm:prefac}, for the tail probability of $f(X_T^\eps)$ remains valid even in the non-convex case with $I''(z) \leq 0$ (as long as $\lambda_z > 0$), without going into more technical details.
Starting from equation~\eqref{eq:deriv-optimality}, which holds regardless of convexity of the rate function $z \mapsto I(z)$, we can take the inner product of~\eqref{eq:deriv-optimality} with~$\partial_z \eta_z$, resulting in
\begin{align*}
I''(z) = \left \langle \partial_z \eta_z , \left[\text{Id} -  A_{\lambda_z}\right] \partial_z \eta_z \right \rangle\,.
\end{align*}
Note that by the second-order sufficient optimality condition~\eqref{eq:sufficent-optimal}, the operator $\text{Id} -  A_{\lambda_z}$ is positive definite on the space~$\eta_z^\perp$, or, more generally, on the tangent space of the constraint hypersurface $\{ \eta \colon F[\eta] = z \}$ at $\eta = \eta_z$. Hence, assuming~\eqref{eq:sufficent-optimal}, the only direction in which $\text{Id} -  A_{\lambda_z}$ can fail to be positive definite is the surface normal $\delta F / \delta \eta |_{\eta_z} \propto \eta_z$. Because of $\langle \delta F / \delta \eta |_{\eta_z}, \partial_z \eta_z \rangle = 1$ by the chain rule, the vector $\partial_z \eta_z$ always has a nonzero component along the surface normal direction, and hence $I''(z) \leq 0$ implies that $\text{Id} -  A_{\lambda_z}$ will not be positive definite in the normal direction, rendering the MGF prefactor~$R_{\lambda_z}$ ill-defined.\\

To address this issue, we note that first of all, numerically solving the constrained optimization problem~\eqref{eq:min-prob-eta} for a $z$ with~$I''(z)\leq 0$ non-convex is not an issue in practice, as long as penalty-type methods are used to incorporate the constraint (this is in contrast to solving the dual problem, which may lead to a finite or infinite duality gap in this case). Having access to~$\eta_z$ and~$\lambda_z$ in this way, one can show that the prefactor~\eqref{eq:tail-prefac-projected} -- which is still \textit{possible} to evaluate by~\eqref{eq:sufficent-optimal} and leads to a well-defined result -- is in fact the correct expression to use, by employing a reparameterization argument, cf.~\cite{alqahtani-grafke:2021}. This reparameterization does not need to be computed explicitly in practice, where~\eqref{eq:tail-prefac-projected} can just be used directly, and is only a theoretical tool here to convexify the problem to apply the results of this appendix. The idea is to construct a function $g \colon \RR \to \RR$ with $g' > 0$ (orientation-preserving) and $g''(z) < 0$ (strictly concave at $z$), such that for the reparameterized observable $\hat{f} = g \circ f$, we have $\hat{I}''(\hat{z}) > 0$ with $\hat{I}$ the rate function of $\hat{f}$, at $\hat{z} = g(z)$. Furthermore, since $\hat{I}'(\hat{z}) = \hat{\lambda}_{\hat{z}} = \lambda_z / g'(z)$, and of course $\hat{\eta}_{\hat{z}} = \eta_z$ (as we only change the parameterization, not the constraint set itself in~\eqref{eq:min-prob-eta}), we can compute
\begin{align}
\text{Id} - \hat{A}_{\hat{\lambda}_{\hat{z}}} = \text{Id} - \left. \frac{\delta^2 \left(\hat{\lambda}_{\hat{z}} \cdot g \circ F \right)}{\delta \eta^2} \right|_{\hat{\eta}_{\hat{z}}} = \text{Id} - A_{\lambda_z} - \lambda_z  \frac{g''(z)}{g'(z)} \left( \left. \frac{\delta F}{\delta \eta} \right|_{\eta_z} \right)^{\otimes 2}\,,
\label{eq:reparam-perturb}
\end{align}
so the reparameterization leads to a rank-1 perturbation of the second variation operator $\text{Id} - A_{\lambda_z}$ precisely in the constraint surface's normal direction. For suitably constructed $g$, this renders the operator $\text{Id} - \hat{A}_{\hat{\lambda}_{\hat{z}}}$ for the reparameterized problem positive definite on the whole space. Then the derivation of this appendix applies and leads to a prefactor $\hat{C}(\hat{z})$ given by~\eqref{eq:tail-prefac-projected} for the asymptotic expansion of the tail probability $\PP \left[g \circ f(X_T^\eps) \geq \hat{z} \right] = \PP \left[f(X_T^\eps) \geq z \right]$, where all quantities in~\eqref{eq:tail-prefac-projected} are replaced by their ``hat-equivalents''. From~\eqref{eq:reparam-perturb}, we have $\text{pr}_{\hat{\eta}_{\hat{z}}^\perp} \hat{A}_{\hat{\lambda}_{\hat{z}}} \text{pr}_{\hat{\eta}_{\hat{z}}^\perp} = \text{pr}_{\eta_z^\perp} A_{\lambda_z} \text{pr}_{\eta_z^\perp}$, and since $\eta_z$, $\phi_z$ and $\theta_z$ remain invariant under the reparameterization, so does the potentially non-trace class part $\tilde{A}$ of the second variation defined through~\eqref{eq:atilde-return} and~\eqref{eq:a-tilde-adjoint-mult-noise}, leading to $\hat{C}(\hat{z}) = C(z)$. In total, this implies that~\eqref{eq:tail-prefac-projected} holds even for non-convex rate functions.


\bibliographystyle{siamplain}
\bibliography{../bib}

\end{document}